\documentclass[11pt]{article}

\usepackage[parfill]{parskip}
\usepackage[letterpaper, left=1in, right=1in, top=1in, bottom=1in]{geometry}

\usepackage{pgfplots}
\pgfplotsset{compat=1.18}
\usepackage{tikz}

\usetikzlibrary{patterns}

\usepackage{amsmath,amssymb}
\usepackage{xcolor}
\usepackage{bbm}
\usepackage{multicol}

\usepackage{natbib}
\bibpunct{(}{)}{;}{a}{,}{,}
 
\usepackage[utf8]{inputenc} 
\usepackage[T1]{fontenc}    
\usepackage[colorlinks=true, linkcolor=blue!70!black, citecolor=blue!70!black,urlcolor=black,breaklinks=true]{hyperref}

\usepackage{hyperref}

\usepackage{comment}
\usepackage{csquotes}

\usepackage{url}            
\usepackage{booktabs}       
\usepackage{amsfonts}       
\usepackage{nicefrac}       
\usepackage{microtype}      

\usepackage{mathtools, amsthm, bm}

\usepackage[frak=mma]{mathalfa}

\allowdisplaybreaks

\usepackage[linesnumbered,ruled,vlined]{algorithm2e}

\SetCommentSty{mycommfont}

\usepackage{subcaption}
\usepackage{multirow}
\usepackage{float}
\usepackage{xspace}

\usepackage{enumitem}

\usepackage{xfrac}

\usepackage{cleveref}

\theoremstyle{plain}
\newtheorem{theorem}{Theorem}[section]

\usepackage{thm-restate}
\newtheorem{lemma}[theorem]{Lemma}

\newtheorem{proposition}[theorem]{Proposition}

\theoremstyle{plain}
\newtheorem{definition}{Definition}[section] 
\newtheorem{example}[definition]{Example}
\newtheorem{remark}[definition]{Remark}

\newtheorem{observation}[definition]{Observation}

\theoremstyle{plain}

\newtheorem{assumption}{Assumption}

	\newcommand{\reals}{\mathbb{R}}
	\newcommand{\positivereals}{\tiny \reals_{+}}

\makeatletter
\newcommand{\subalign}[1]{%
  \vcenter{%
    \Let@ \restore@math@cr \default@tag
    \baselineskip\fontdimen10 \scriptfont\tw@
    \advance\baselineskip\fontdimen12 \scriptfont\tw@
    \lineskip\thr@@\fontdimen8 \scriptfont\thr@@
    \lineskiplimit\lineskip
    \ialign{\hfil$\m@th\scriptstyle##$&$\m@th\scriptstyle{}##$\hfil\crcr
      #1\crcr
    }%
  }%
}
\makeatother

	\DeclareMathOperator{\OPT}{\texttt{OPT}}

	\providecommand{\given}{}
	\DeclarePairedDelimiterX{\set}[1]\{\}{\renewcommand\given{\nonscript\:\delimsize\vert\nonscript\:\mathopen{}}#1}
	\let\Pr\relax
	\DeclarePairedDelimiterXPP{\Pr}[1]{\mathbb{P}}[]{}{\renewcommand\given{\nonscript\:\delimsize\vert\nonscript\:\mathopen{}}#1}
	\DeclarePairedDelimiterXPP{\Ex}[1]{\mathbb{E}}[]{}{\renewcommand\given{\nonscript\:\delimsize\vert\nonscript\:\mathopen{}}#1}

	\newcommand{\primed}{^\dagger}

    \newcommand{\Profit}[2][]{\text{\bf Profit}\ifthenelse{\not\equal{}{#1}}{_{#1}}{}\!\left[{\def\givenn{\middle|}#2}\right]}
    \newcommand{\Cost}[2][]{\text{\bf COST}\ifthenelse{\not\equal{}{#1}}{_{#1}}{}\!\left[{\def\givenn{\middle|}#2}\right]}
    \newcommand{\CostF}[2][]{\text{\bf COST}^{\bf{F}}\ifthenelse{\not\equal{}{#1}}{_{#1}}{}\!\left[{\def\givenn{\middle|}#2}\right]}
    \newcommand{\CostC}[2][]{\text{\bf COST}^{\bf{C}}\ifthenelse{\not\equal{}{#1}}{_{#1}}{}\!\left[{\def\givenn{\middle|}#2}\right]}

	\newcommand{\Rev}[2][]{\text{\bf Rev}\ifthenelse{\not\equal{}{#1}}{_{#1}}{}\!\left[{\def\givenn{\middle|}#2}\right]}

    \newcommand{\Obj}[2][]{\text{\bf OBJ}\ifthenelse{\not\equal{}{#1}}{_{#1}}{}\!\left[{\def\givenn{\middle|}#2}\right]}
    
    \newcommand{\SFRP}[1]{\text{\hyperref[eq:strongly factor-revealing quadratic program]{$\mathcal{P}_{\texttt{SFR}}(#1)$}}}
    \newcommand{\WFRP}[1]{\text{\hyperref[eq:weakly factor-revealing program]{$\mathcal{P}_{\texttt{FR}}(#1)$}}}

    \newcommand{\SFRPMFLP}[1]{\text{\hyperref[eq:strongly factor-revealing program MFLP]{$\mathcal{P}_{\texttt{SFR-MFLP}}(#1)$}}}
    \newcommand{\WFRPMFLP}[1]{\text{\hyperref[eq:weakly factor-revealing program MFLP]{$\mathcal{P}_{\texttt{FR-MFLP}}(#1)$}}}

\newcommand{\LBLP}[1]{\text{\hyperref[eq:lower bound LP]{$\mathcal{P}_{\texttt{LB}}(#1)$}}}

	\newcommand{\alloc}{x}

	\newcommand{\noaccents}[1]{#1}
	
	\newcommand{\newagentvar}[3][\noaccents]{%
		\expandafter\newcommand\expandafter{\csname #2\endcsname}{#1{#3}}%
		\expandafter\newcommand\expandafter{\csname #2s\endcsname}{#1{\boldsymbol{#3}}}%
		\expandafter\newcommand\expandafter{\csname #2smi\endcsname}[1][i]{#1{\boldsymbol{#3}}_{-##1}}%
		\expandafter\newcommand\expandafter{\csname #2i\endcsname}[1][i]{#1{#3}\agind[##1]}%
		\expandafter\newcommand\expandafter{\csname #2ith\endcsname}[1][i]{#1{#3}_{(##1)}}%
	}
	\newagentvar{typespace}{{\cal Z}}
	\newagentvar{typesubspace}{S}
	
	\newagentvar{type}{z}
	\newagentvar{othertype}{s}
	\newagentvar{val}{v}
	\newagentvar{hval}{\bar \val}
	\newagentvar{hbudget}{\bar \wealth}
	\newagentvar{budget}{B}
	\newagentvar{lbudget}{\underaccent{\bar}{ \wealth}}
	\newagentvar{lowestval}{0}
	\newagentvar{cumval}{V}
	\newagentvar{cumprice}{P}
	\newagentvar{welcurve}{W}
	\newagentvar{revcurve}{R}
	\newagentvar{outcome}{w}
	\newagentvar{outcomespace}{{\cal W}}

	\newcommand{\distance}{d}
\newcommand{\flow}{\tau}
\newcommand{\flowi}{\flow_{ij}}

\newcommand{\pop}{N}
\newcommand{\popi}{\pop_i}

\newcommand{\stores}{\SOL}

\newcommand{\opencost}{f}
\newcommand{\opencosti}{\opencost_i}
\newcommand{\opencosts}{\{\opencosti\}}

\newcommand{\instance}{I}

\newcommand{\SOL}{\texttt{SOL}}
\renewcommand{\OPT}{\texttt{OPT}}

\newcommand{\approxratio}{\Gamma}
\newcommand{\SFRapproxratio}{\approxratio_{\texttt{SFR}}(\discountfactor,\costscalar)}
\newcommand{\WFRapproxratio}{\approxratio_{\texttt{FR}}(\discountfactor, \costscalar)}

\newcommand{\WFRMFLPapproxratio}{\approxratio_{\texttt{FR-MFLP}}}

\newcommand{\connectioncost}{c}

\newcommand{\distanceHat}{{\distance}}
\newcommand{\costHat}{{\opencost}}
\newcommand{\distanceHatStar}{{\connectioncost}}
\newcommand{\muHat}{{\alpha}}

\newcommand{\permu}{\chi}

\newcommand{\greedyapprox}{\text{2.497}}
\newcommand{\greedyapproxMFLP}{\text{1.819}}
\newcommand{\greedyapproxLB}{2.428}

\newcommand{\Home}{\texttt{H}}
\newcommand{\Work}{\texttt{W}}

\newcommand{\TWOLFLP}{2-LFLP}
\newcommand{\edge}{e}
\newcommand{\edgehome}{\edge_\Home}
\newcommand{\edgework}{\edge_\Work}
\newcommand{\flowe}{\flow_\edge}
\newcommand{\ALG}{\texttt{ALG}}

\newcommand{\plus}[1]{{\left( #1 \right)^+}}

\newcommand{\dual}{\mu}

\newcommand{\starinstance}{S}

\newcommand{\edgeindex}{\texttt{L}}

\newcommand{\greedycounter}{\alpha}
\newcommand{\firstchanceset}{U}

\newcommand{\connectmapping}{\psi}
\newcommand{\Edge}{E}
\newcommand{\config}{\sigma}

\newcommand{\timeindex}{t}
\newcommand{\cost}{c}
\newcommand{\coststar}{\cost(\starinstance)}
\newcommand{\allocstar}{\alloc(\starinstance)}
\newcommand{\duale}{\dual(\edge)}
\newcommand{\starspace}{\mathcal{S}}
\newcommand{\singlecountedgeset}{\Edge^{(1)}}
\newcommand{\kcountedgeset}{\Edge^{(k)}}
\newcommand{\doublecountedgeset}{\Edge^{(2)}}

\newcommand{\mHat}{{m}}
\newcommand{\nHat}{{n}}
\newcommand{\qHat}{{q}}

\newcommand{\qHatBf}{\mathbf{q}}
\newcommand{\muHatBf}{\boldsymbol{\muHat}}
\newcommand{\distanceHatBf}{{\mathbf{\distance}}}
\newcommand{\distanceHatStarBf}{{\mathbf{\connectioncost}}}
\newcommand{\zerobf}{\mathbf{0}}

\newcommand{\naturals}{\mathbb{N}}

\newcommand{\discountfactor}{\gamma}
\newcommand{\nullsymbol}{\bot}
\newcommand{\hwset}{\{\Home,\Work\}}

\newcommand{\normalizefactor}{N}
\newcommand{\SFRCmonotonicity}{\texttt{(SFR.i)}}
\newcommand{\SFRCtriangleineq}{\texttt{(SFR.ii)}}
\newcommand{\SFRCcontribution}{\texttt{(SFR.iii)}}
\newcommand{\SFRCcontributionstronger}{\texttt{(SFR.iii.0)}}
\newcommand{\SFRCdistance}{\texttt{(SFR.iv)}}
\newcommand{\SFRCconnectcost}{\texttt{(SFR.v)}}
\newcommand{\SFRCtotalcost}{\texttt{(SFR.vi)}}
\newcommand{\SFRCdensity}{\texttt{(SFR.vii)}}

\newcommand{\WFRCtriangleineq}{\texttt{(FR.i)}}
\newcommand{\WFRCcontribution}{\texttt{(FR.ii)}}
\newcommand{\WFRCconnectcost}{\texttt{(FR.iii)}}
\newcommand{\WFRCtotalcost}{\texttt{(FR.iv)}}

\newcommand{\EdgeTilde}{\tilde{\Edge}}
\newcommand{\edgetoellmapping}{\kappa}

\newcommand{\empattract}{\rho}
\newcommand{\flowconstructionfactor}{\iota}
\newcommand{\costconstructionfactor}{\bar{\opencost}}
\newcommand{\zhvi}{z}
\newcommand{\CFAlg}{\texttt{2-GR}}
\newcommand{\CFAlgP}{\texttt{2-GRP}}
\newcommand{\GDH}{\texttt{GR-H}}
\newcommand{\GDW}{\texttt{GR-W}}

\newcommand{\candidatecosttext}{candidate cost}

\newcommand{\KLFLP}{$K$-LFLP}
\newcommand{\RelaxSFRapproxratio}{\approxratio_{\texttt{SFR-K}}}
\newcommand{\RelaxSFRP}[1]{\text{\hyperref[eq:relax strongly factor-revealing quadratic program]{$\mathcal{P}_{\texttt{SFR-K}}(#1)$}}}

\newcommand{\permuTGD}{\mathcal{Y}}

\newcommand{\costscalar}{\eta}

\newcommand{\argmin}{\mathop{\mathrm{arg\,min}}}
\newcommand{\argmax}{\mathop{\mathrm{arg\,max}}}

\newcommand{\prob}[2][]{\text{Pr}\ifthenelse{\not\equal{}{#1}}{_{#1}}{}\!\left[{\def\givenn{\middle|}#2}\right]}
\newcommand{\expect}[2][]{\mathbb{E}\ifthenelse{\not\equal{}{#1}}{_{#1}}{}\!\left[{\def\givenn{\middle|}#2}\right]}

\newcommand{\tparen}{\big}
\newcommand{\tprob}[2][]{\text{Pr}\ifthenelse{\not\equal{}{#1}}{_{#1}}{}\tparen[{\def\given{\tparen|}#2}\tparen]}
\newcommand{\texpect}[2][]{\mathbb{E}\ifthenelse{\not\equal{}{#1}}{_{#1}}{}\tparen[{\def\given{\tparen|}#2}\tparen]}

\newcommand{\sprob}[2][]{\text{Pr}\ifthenelse{\not\equal{}{#1}}{_{#1}}{}[#2]}
\newcommand{\sexpect}[2][]{\mathbb{E}\ifthenelse{\not\equal{}{#1}}{_{#1}}{}[#2]}

\newcommand{\indicator}[1]{{\mathbbm{1}\left\{ #1 \right\}}}

\title{Mobility Data in Operations: \\
The Facility Location Problem}
\author{Ozan Candogan\thanks{University of Chicago Booth School of Business. Email: \texttt{ozan.candogan@chicagobooth.edu}.}  \and Yiding Feng\thanks{University of Chicago Booth School of Business. Email: \texttt{yiding.feng@chicagobooth.edu}.} }
\date{}
\begin{document}





\maketitle

\begin{abstract}

The recent large scale availability of 
mobility data, which captures individual
mobility patterns, 
poses
novel operational problems that are
exciting and
 challenging.
Motivated by this, 
we introduce and study a variant of 
the (cost-minimization) 
facility location problem 
where each individual is endowed with
two locations (hereafter, her home and work locations),
and the connection cost is the minimum distance
between any of her locations and its closest facility.
We design a 
polynomial-time algorithm whose 
approximation ratio is at most
$\greedyapprox$.
We complement this positive result by showing that
the proposed algorithm 
is at least a $\greedyapproxLB$-approximation,
and there exists no polynomial-time algorithm
with approximation ratio $2-\epsilon$
under UG-hardness.
We further extend our results and analysis to 
the model where each individual 
is endowed with $K$ locations.
Finally, we 
conduct numerical experiments over both 
synthetic data and
US census data
(for NYC, greater LA, greater DC, Research Triangle)
and evaluate the performance of our  algorithms.
\end{abstract}

\setcounter{page}{0}
\newpage

\section{Introduction}
\label{sec:intro}

Individual mobility patterns have a first-order impact on many operational decisions, ranging from facility location decisions to optimization of transit systems. In the past, obtaining data on individual mobility patterns was a challenging task. However, in recent years, such data have been extensively collected through mobile phones, allowing large-scale analysis. Consequently, these valuable insights have become more accessible to decision makers, thanks to the efforts of various data providers such as Safegraph and Carto \citep{Safegraph,Carto}.



Due to privacy concerns and data limitations, many mobility datasets (e.g., see \citealp{Carto-ODMatrix}) only record the most frequently visited places by anonymous individuals, such as their home and work locations. Decision makers can use this information to improve their decision making.
One of the main applications emphasized by data providers for their mobility data is the facility location problem. An illustrative example can be found in the case of ASDA, one of the largest British supermarket chains, which leverages mobility data to make informed decisions about new facility selections \citep{Carto-ASDA}. 
Other practical examples include the last-mile delivery market, e.g., \citet{LT-22} incorporate the mobility data and build the network of public lockers in Singapore;
and the design of the bike sharing systems, for example \citet{AAFDB-21} uses mobility data to study bike sharing stations in Lisbon.

In our paper, we aim to introduce new methodologies and important algorithmic contributions. Specifically, our research undertakes the task of formalizing the facility location problem using mobility data, providing valuable insights into how such data can enhance decision-making processes and its overall significance. We thoroughly examine the shortcomings of traditional algorithms commonly employed for facility location, propose new algorithms supported by theoretical guarantees, and present extensive numerical studies.
In essence, our work delivers a comprehensive \emph{package} for the study of a challenging problem that arises from the availability of new data sources. The ambition of this paper is to initiate the study of using mobility data to improve firm's decisions.

\subsection{Our contributions and techniques}

We introduce and study the 
\emph{2-location facility location problem (\TWOLFLP)}:
A decision maker  wants to choose locations of his
facilities to minimize
the sum of (i) facility opening costs and (ii) 
each individual's
``connection cost''  to a facility.
In the \TWOLFLP,
each individual is endowed with 2 locations, and her 
connection cost is the \emph{minimum} of the distances
between any of her locations
and its closest facility.\footnote{The connection cost formulation 
in this paper is relevant to applications such as 
retail, food bank, COVID testing, vaccine, and dynamic fulfillment problem.}
The majority part of this paper focuses on the \TWOLFLP,
and refers to each individual's endowed locations as 
her \emph{home} and \emph{work} locations.
By restricting attention to instances where the home location 
matches the work location for every individual,
the \TWOLFLP\ recovers the classical (single-location) metric facility location problem (MFLP).
Our main result is the following:

\begin{displayquote}
\emph{\textbf{(Main Result)} For the \TWOLFLP, we propose 
a polynomial-time algorithm (\Cref{alg:greedy modified}) that obtains a constant approximation 
ratio of $\greedyapprox$.
}
\end{displayquote}

In all versions of the facility location problem (e.g., \TWOLFLP, MFLP), it is natural to consider a greedy algorithm:
(i) initialize all individual in the \emph{unconnected} state;
then
(ii) iteratively pick the most cost-effective choice for unconnected individuals,
i.e., either connect an unconnected individual 
to a facility that is already open,
or 
open a new facility and connect a subset of unconnected individuals
to this new facility.
In fact,
\citet{JMMSV-02} shows that this greedy algorithm, hereafter
 the JMMSV algorithm,
is a 1.861-approximation for the MFLP.
We show that
the JMMSV algorithm no longer attains
a constant-approximation guarantee
in the \TWOLFLP.\footnote{\Cref{example:JMMSV failure} shows that
the JMMSV algorithm is an $\Omega(\log(n))$-approximation 
where $n$ is the number of total locations.}
Loosely speaking, the JMMSV algorithm fails to achieve a constant approximation ratio
since
the algorithm iteratively and  greedily connects each unconnected individual to a facility  \emph{based solely on one of her locations}. However, since each individual in the \TWOLFLP\ is endowed with two locations, the location used for the initial connection may (ex post) turn out to be quite suboptimal, thereby increasing overall costs substantially.

Motivated by the JMMSV algorithm and its failure in the \TWOLFLP, 
we propose a new parametric family of algorithms,
referred to as the \emph{2-Chance Greedy Algorithm},
which generalizes the JMMSV algorithm 
by allowing
each individual to be connected \emph{twice} (once
through the home and once through the work location).
This algorithm has two tuning parameters
$\discountfactor\in[0, 1]$ and $\costscalar \in\reals_+$.
The discount factor $\discountfactor$ discounts the cost improvement obtained by connecting individuals who were previously connected to a facility to a new one.
Since each individual can be connected twice, the 2-Chance Greedy Algorithm may open more facilities than expected and lead to a high facility opening cost. To handle this, the opening cost scalar $\costscalar$ controls the tradeoff between the facility opening cost and the individual connection cost in the algorithm.
When the discount factor $\discountfactor$
is set to zero and opening cost scalar $\costscalar$ is set to one, 
our algorithm recovers the JMMSV algorithm, but by choosing $\discountfactor>0$ appropriately, substantial cost savings can be achieved.

To characterize the performance  
of the 2-Chance Greedy Algorithm 
 with discount factor $\discountfactor$ and opening cost scalar $\costscalar$,
we design a family of \emph{strongly factor-revealing quadratic programs} 
$\{\SFRP{\nHat,\discountfactor,\costscalar}\}_{\nHat\in\naturals}$.
We show that for any parameter $n$ the 
optimal objective of the corresponding program 
upper bounds the 
 approximation ratio achieved by our 
 algorithm.
 Using this result and setting $n=2$ 
 we obtain
 a loose analytical upper bound 
 of 
 $\sfrac{6(1+\discountfactor)}{\discountfactor^2}$
 for our algorithm.
Numerically solving our program we show that the optimal objective value of $\SFRP{25, 1, 2}$ is upper bounded by 
$\greedyapprox$, thereby yielding an  approximation ratio of $\greedyapprox$ for the 2-Chance Greedy Algorithm
with $\discountfactor = 1$ and $\costscalar = 2$, as also stated above. 

One implication of our approximation ratio result is as follows. In practice, the social planner can implement the 2-Chance Greedy Algorithm by varying the discount factor $\discountfactor$ and the opening cost scalar $\costscalar$, and then select the solution with the minimum cost.
Our theoretical result ensures that there exists at least one parameter assignment (i.e., $\discountfactor = 1$ and $\costscalar = 2$) with an approximation ratio of $\greedyapprox$. Nonetheless, it is possible that for the particular instance of the social planner, a better solution can be found by varying $\discountfactor$ and $\costscalar$.


It is worth highlighting that 
our results diverge from 
the prior works \citep[e.g.,][]{JMMSV-02,MYZ-06}
in the literature where
the approximation ratios are upperbounded by the \emph{supremum}
of some families of factor-revealing programs.
Consequently, unlike prior works which need to
analyze their factor-revealing programs for \emph{all} parameters,
in this work
evaluating program~\ref{eq:strongly factor-revealing quadratic program}
for an \emph{arbitrary} $\nHat\in\naturals$
provides an  upperbound on the approximation ratio (and the best upper bound can be obtained by taking infimum over $n$).
We refer to program~\ref{eq:strongly factor-revealing quadratic program}
as the \emph{strongly} factor-revealing quadratic program 
to emphasize this distinction.

\paragraph{New analysis framework: primal-dual \& strongly factor-revealing quadratic program.}
We rely on a primal-dual analysis to prove the bounds on 
the approximation ratio of the 2-Chance Greedy Algorithm.
Similar to prior work, we consider a linear programming relaxation 
of the optimal solution and its dual program.
In this dual program,
there is a non-negative dual variable 
associated with each individual.
The dual objective function 
is simply the sum of all these dual variables,
and dual constraints impose restrictions on
the sum of dual variables for every subset of individuals.
The goal of the primal-dual method 
is to construct a dual assignment 
such that (i) its objective value 
is weakly larger than the total cost 
of the solution outputted by the 2-Chance Greedy Algorithm;
and (ii) every dual constraint
is approximately satisfied.
To achieve this goal, loosely speaking,
we decompose the total cost from the algorithm over individuals,
and let each dual variable take the value equal to the  cost 
of its corresponding individual.
In this way, (i) is satisfied automatically.
We then establish a set of structural properties
of feasible execution paths generated 
by the 2-Chance Greedy Algorithm (\Cref{lem:structural lemma}),
and  use those 
to characterize 
the approximation factor of the dual constraints.
Formally, this yields an approximation ratio for our algorithm over  \emph{supremum} of  optimal objectives of factor-revealing programs \{\ref{eq:weakly factor-revealing program}\}.
Finally, we upperbound the 
optimal objective value of
program~\ref{eq:weakly factor-revealing program}
in terms of the \emph{infimum} of 
\emph{strongly factor-revealing quadratic programs} \{\ref{eq:strongly factor-revealing quadratic program}\}.

Our method 
is (at least superficially) similar 
to the primal-dual analysis plus factor-revealing program
for the JMMSV algorithm in the MFLP \citep{JMMSV-02}.
However, there are two important differences in our setting that  make the analysis of the 2-Chance Greedy Algorithm challenging,
and differentiate our analysis from prior work
:
(i) \emph{Violation of triangle inequality:}
in the \TWOLFLP,
though the distance over locations in a metric space
satisfies the triangle inequality,
the connection cost over individuals
may not satisfy the triangle inequality,\footnote{Namely, 
suppose an individual $i$ has a low connection cost to a facility $j$ close to her home location,
and another individual~$i'$ is ``close'' to individual $i$ since they share the same work location.
Nonetheless, this does  not guarantee that 
the connection cost of individual $i'$ to facility $j$ is low as well.}
(ii) \emph{Non-monotonicity of decomposed costs:} 
the decomposed cost for each individual 
is not monotone increasing with respect to the time when each individual 
is connected (in particular, her second connection) in the 2-Chance Greedy Algorithm. 
Notably, while challenge (ii) is 
an issue for the 2-Chance Greedy Algorithm
(and possibly for greedy-style algorithms in general),
challenge (i) is an issue for the \TWOLFLP\ itself (regardless of the algorithms).
Both challenges
have been crucial for deriving structural properties 
of greedy  algorithms and their variants (including the JMMSV algorithm),
as well as obtaining
some families of factor-revealing \emph{linear} program 
that can be analytically evaluated
in prior work.

To address challenge (i), 
we replace a structural property 
induced by the triangle inequality 
over individuals in prior work 
with a weaker version induced 
by the triangle inequality over locations.
Combining this weaker structural property 
and other structural properties of the 2-Chance Greedy Algorithm,
we upperbound the approximation ratio of the algorithm in terms of  the supremum of a family of factor-revealing programs~\{\ref{eq:weakly factor-revealing program}\}.
These programs are  \emph{non-linear/non-quadratic/non-convex}, and thus hard to  solve.
To side step this difficulty, 
we further upperbound this supremum of
programs~\{\ref{eq:weakly factor-revealing program}\}
through a strongly  factor-revealing 
quadratic program~\ref{eq:strongly factor-revealing quadratic program},
which can be
both analytically analyzed (albeit leading to loose bounds) and numerically computed.\footnote{(Non-convex) quadratic programs are supported by optimization solvers
such as Gurobi, Matlab fmincon; 
or bounded
through semidefinite relaxations and then solved by common 
SP solver such as Mosek, SDPA, CSDP, SeDuMi.}

The concept of strongly factor-revealing program was originally introduced by \citet{MY-11} for the competitive ratio analysis 
for online bipartite matching with random arrivals.
To the best of our knowledge, all previous works 
with strongly factor-revealing programs
start with a family of factor-revealing \emph{linear} programs,
consider some ad-hoc relaxations, and then use
a \emph{naive batching argument} to compress multiple variables in 
the original linear program into a single variable in the relaxed linear program.
Invoking the linearity of both programs, 
the feasibility of 
the constructed solution from the naive batching argument 
is guaranteed.
As a warm-up exercise, in this paper we  illustrate how to 
use such a naive batching argument to obtain 
a strongly factor-revealing linear program and reprove the $\greedyapproxMFLP$
for the JMMSV algorithm
in the MFLP (\Cref{sec:strongly factor-revealing program MFLP}).
However, as we mentioned above, the factor-revealing program~\ref{eq:weakly factor-revealing program}
for the 2-Chance Greedy Algorithm in the \TWOLFLP\ is 
non-linear/non-quadratic/non-convex, and thus the naive batching argument fails
(\Cref{example: failure of naive batching}).
Yet, we introduce a new 
\emph{solution-dependent batching argument},
which enables us to obtain the strongly factor-revealing quadratic program~\ref{eq:strongly factor-revealing quadratic program}.
Given the popularity of using factor-revealing programs in the algorithm design literature, we believe our novel solution-dependent batching idea
might be of 
independent interest.

\paragraph{Approximation hardness.}
We complement our main results with two hardness results.

Our first hardness result is the existence of 
a \TWOLFLP\ instance 
such that 
the approximation ratio of the 2-Chance Greedy Algorithm 
with discount factor $\discountfactor = 1$ and opening cost scalar $\costscalar = 2$
is at least $\greedyapproxLB$.
We obtain this hardness result with a three-step approach.
First, we construct a linear program~$\LBLP{\nHat}$
by introducing additional constraints into 
the strongly 
factor-revealing quadratic program~$\SFRP{\nHat, 1, 2}$.
Second, we argue that the optimal solution in program~$\LBLP{\nHat}$
can be converted into a \TWOLFLP\ instance where the 
approximation ratio 
of the algorithm
equals to the optimal objective 
of program~$\LBLP{\nHat}$. 
We then numerically evaluate program~$\LBLP{500}$
and obtain $\greedyapproxLB$ as an approximation lower bound.\footnote{We 
conjecture that the gap between lower bound $\LBLP{\nHat}$ and upper bound $\SFRP{\nHat, 1, 2}$ decreases to zero as $\nHat$ goes to infinity.
}

Our second hardness result is that 
no polynomial-time algorithm
can obtain a approximation ratio of $2-\epsilon$ 
under UG-hardness \citep{kho-02} in the \TWOLFLP.
This approximation hardness result is directly implied
by the same hardness result in 
the vertex cover problem \citep{KR-08},
since the former
generalizes the latter:
consider each vertex/edge in the latter as a location/individual in the former,
and let distances between locations be infinite.
Notably, the 2-Chance Greedy Algorithm recovers
the classic dual-fitting algorithm
for the vertex cover problem instances and attains an approximation ratio of 2 for the vertex cover instances.

\paragraph{Extension to $K$-location FLP.}
Our model admits a natural extension, which we refer to as the \emph{$K$-location facility location problem (\KLFLP)}. 
In \KLFLP, each individual is endowed with $K$ locations, e.g., her top $K$ most frequently visited places in the mobility data. The individual's connection cost is the minimum of the distances of any of her $K$ locations and its closest facility.
For this extension model, we introduce the $K$-Chance Greedy Algorithm (\Cref{alg:K greedy modified}).
It is a natural generalization of the 2-Chance Greedy Algorithm, where each individual can be connected at most $K$ times through each of her locations. 
By properly picking the parameters of the $K$-Chance Greedy Algorithm and invoking a similar analysis, we show that the approximation ratio of the algorithm is at most the infimum of the strongly quadratic factor-revealing programs \{\ref{eq:relax strongly factor-revealing quadratic program}\} over $\nHat\in\naturals$.
We summarize these upper bounds of the approximation ratio for $K$ between 1 and 20 in \Cref{table:approx ratio KFLP}. For example, for $K = 3, 4, 5$, the upper bounds of the approximation ratio are 3.538, 4.58, 5.611, respectively. For $K \geq 21$, the upper bound of the approximation ratio is $1.059K$.
On the other hand, using the reduction from the vertex cover problem for $K$-uniform hypergraph \citep{KR-08}, we show that there exists no polynomial-time algorithm 
with approximation ratio better than $K - \epsilon$ under UG-hardness. Therefore, the linear dependence of $K$ in our approximation ratio guarantee of the $K$-Chance Greedy Algorithm is necessary and order optimal.\footnote{We 
conjecture that the approximation ratio upper bound \ref{eq:relax strongly factor-revealing quadratic program} converges to $K + o(1)$ as $\nHat$ goes to infinity, and thus the $K$-Chance Greedy Algorithm attains asymptotic optimal approximation ratio of $K + o(1)$.}

\paragraph{Numerical simulations.}
We provide a numerical justification for the performance 
of the 2-Chance Greedy Algorithm.
We construct numerical experiments over both randomly-generated 
synthetic data and the US census data. 
For the latter, we construct \TWOLFLP\ instances 
for four cities in the US:
New York City (NYC), Los Angeles metropolitan area (greater LA), Washington metropolitan area
(greater DC), and Raleigh-Durham-Cary CSA (Research Triangle). 

We 
discretize the set of discount factors and opening cost scalars and
compute the performance of the 2-Chance Greedy Algorithm
with discount factor $\discountfactor \in \{0, 0.2, 0.4, 0.6, 0.8, 1\}$ and opening cost scalar $\costscalar\in\{1, 1 + 0.5\discountfactor, 1 + \discountfactor\}$.
We observe that different parameters $(\discountfactor,\costscalar)$ attain the best cost for different instances, which highlights the necessity of implementing the 2-Chance Greedy Algorithm with different parameter assignments.
We also obtain some structural observations: the
number of opened facilities is increasing w.r.t.\ $\discountfactor$ and decreasing w.r.t.\ $\costscalar$.
This aligns with the construction of the algorithm, that is, the algorithm with larger $\discountfactor$ (smaller $\costscalar$)
opens facilities more aggressively since $\discountfactor$
amplifies ($\costscalar$ discounts) the value of
  opening additional facilities.
Motivated by this, 
we design a post-processing step,
which myopically ``prunes'' the solution by checking if 
the objective can be improved by removing
any facility from the solution,
and combine it with the 2-Chance Greedy Algorithm. 
Theoretically, adding this post-processing step does not 
worsen the approximation guarantee.
Numerically, we observe that
combing with myopic pruning,
the 2-Chance Greedy Algorithm with myopic pruning further improves the performance.

Finally, we analyze the \emph{value of mobility data}.
Specifically, we consider a scenario where
the mobility data (recording the pair of home and work locations
for each individual) are missing, and
a decision maker
only has the residential population (resp. employment) information in each location, 
and then implements the JMMSV algorithm pretending
that each individual can only be connected through her home (resp. work) location.
In most experiments, 
we observe that the 2-Chance Greedy Algorithm that utilizes mobility data achieves substantially better performance 
than the performance of the JMMSV algorithm without 
mobility data.

\paragraph{Organization.}
We start by formalizing the model and providing necessary preliminaries and
notations in \Cref{sec:prelim}. 
In \Cref{sec:2-Chance Greedy}, we introduce the 2-Chance Greedy Algorithm and discuss its connection to the classic JMMSV algorithm. 
In \Cref{sec:greedy two location analysis}, we present the approximation results of the 2-Chance Greedy Algorithm. 
We conduct numerical experiments over both synthetic data and US census
data in \Cref{sec:numerical}. 
Finally, in \Cref{apx:k location} we extend our model, algorithm and approximation guarantee to the $K$-location FLP.

\subsection{Further related work}
There has been a long line of research on
the
(single-location) facility location problem.
The \TWOLFLP\ (and \KLFLP) generalizes 
the single-location metric facility location problem
(MFLP),
and can be thought as a special case of 
single-location non-metric facility location problem
(NMFLP).
\citet{hoc-82} presents a greedy algorithm with $O(\log n)$ 
approximation guarantee for the NMFLP.
Since the first constant-approximation algorithm
given by 
\citet{STA-97}
for the MFLP,
several techniques and improved results have been developed around this problem
\citep[e.g.,][]{KPR-00,AGKMMP-01,JMMSV-02,CS-03,MYZ-06,byr-07}.
Currently, the best approximation guarantee of 1.488 is due to \citet{li-11}; and 
\citet{GK-99} show that it is hard to approximate within a factor of 1.463.
There is another line of research on
the $K$-level facility location problem,
where each individual is endowed with a single location
and needs to be connected with $K$ facilities in a hierarchical order. 
For the 2-level (resp.\ $K$-level)
facility location problem,
\citet{zha-06} 
(resp.\ \citet{ACS-99})
proposes a 1.77-approximation 
(resp.\ 3-approximation) 
algorithm.
See survey by \citet{OCL-18}
for a comprehensive discussion.
\citet{DPW-22} study another variant of revenue-maximizing facility location problem with capacity constraint, which the authors refer as $K$-sided facility location problem.
All these variants 
are fundamentally different from 
the $K$-location facility location problem
studied in this paper.
For example,
it is hard to approximate
within a factor of $K-\epsilon$
in the \KLFLP.
In contrast, a 3-approximation 
polynomial time algorithm exists for 
the $K$-level facility location problem
\citep{ACS-99}.
Due to both theoretical and practical importance of the facility location problem, there are many other variant models studied in the literature. For example, \citet{PT-13,ABGOT-22} considering strategic individuals, \citet{Mey-01,KNR-23} considering online decision making, \citet{WLHT-22} combining FLP with other combinatorial optimization problems.

There have been numerous works on different problems where the approximation or competitive ratio is determined by the infimum of a class of factor-revealing programs across a potentially large parameter space. For example, \citet{JMMSV-02,MYZ-06} for the facility location problem, \citet{msvv-05} for the AdWords problem, \citet{MY-11,GT-12} for the online matching problem, \citet{AHNPY-18,AB-20,ABB-22} in mechanism design, \citet{CSZ-21} for the prophet secretary problem. 
Most of these works involve intricate and potentially imprecise analyses that explore the entire parameter space to obtain their final results. In contrast, \citet{MY-11,GT-12} employ the concept of strongly factor-revealing programs, which provide upper bounds on the infimum of the original factor-revealing programs by utilizing the supremum of the strongly factor-revealing programs. As a result, a single evaluation of a strongly factor-revealing program can yield the desired outcome for any given parameter. Previously, the use of strongly factor-revealing programs was primarily confined to linear programs due to their simplicity. To the best of our knowledge, our work represents the first instance of introducing a technique (the solution-dependent batching argument) that enables the derivation of a strongly non-linear factor-revealing program. Given the widespread adoption of the factor-revealing program approach, we believe that our technique holds independent interest.


\section{Preliminaries}
\label{sec:prelim}

In this paper, 
we study the cost-minimization
for the \emph{$2$-location (uncapacitated)
facility location problem (2-LFLP)} defined as follows.\footnote{In \Cref{apx:k location}, we
extend our model 
to a more general setting where each individual is associated with $K$ locations,
and explain how our algorithm and results carry over to this richer setting.}

There is a metric space 
$([n], \distance)$
with $n$ locations
indexed by $[n]\triangleq \{1,\dots, n\}$ and a (metric)
distance function  $\distance:[n]\times[n] 
\rightarrow
\reals_+$.
Let $\Edge \triangleq [n]^2$.
For each pair of locations $(i, j)\in\Edge$,
we use $\flowi$ to denote the  number of 
individuals who reside in location $i$
and work in location $j$.
We also use 
the notation $\edge = (\edgehome, \edgework)
\in\Edge$
to denote individuals who reside in location~$\edgehome$ 
and work in location~$\edgework$, and
refer to $e$ as an  \emph{edge} between these two locations. 
Throughout the paper, we use the notation $(i, j)$
and $\edge$ interchangeably.
With a slight abuse of notation, 
we define $\flowe \triangleq \flow_{\edgehome\edgework}$ 
as the number of individuals on edge $\edge$,
and $\distance(\edge, i) \triangleq 
\min\{\distance(\edgehome, i),\distance(\edgework, i)\}$
as the smallest distance between location $i\in[n]$
and individuals on edge $\edge$.

There is a social planner who
wants to choose a subset 
$\stores\subseteq[n]$
as the locations of her facilities 
to minimize the total cost.
Specifically, 
given an arbitrary solution $\stores\subseteq [n]$
as the location of facilities, 
the total cost $\Cost{\stores}$ is defined as
a combination of the facility opening cost and 
the
individual connection cost:\footnote{Namely, for each individual who resides in location $i$ and works in location $j$,
there is a connection cost which is equal to 
the minimum of the distances between her home or work location to the closest 
facility in solution $\stores$.}
\begin{align*}
    \Cost{\stores}
    \triangleq
    \displaystyle\sum_{i\in\stores}\opencosti
    +
    \sum_{\edge\in\Edge}
    \flowe \cdot 
    \min_{i\in\stores}
    \distance(e, i)
\end{align*}
where 
$\opencosti$ is the facility opening cost 
for location $i$.



\paragraph{Approximation.}
Given a problem instance $\instance$,
we say a solution $\stores\subseteq[n]$
is a \emph{$\approxratio$-approximation} to the optimal solution $\OPT$
if 
\begin{align*}
    \Cost{\stores} \leq \approxratio \cdot 
    \Cost{\OPT}
\end{align*}
where optimal solution $\OPT \triangleq 
\argmin\limits_{\stores'\subseteq[n]}\Cost{\stores'}$
is the solution minimizing the total cost.

We say an algorithm $\ALG$ is a \emph{$\approxratio$-approximation}
if for every problem instance $\instance$,
the solution $\ALG(\instance)$ computed from algorithm $\ALG$
is a $\approxratio$-approximation.
In this paper, our goal is to design 
polynomial-time algorithms with constant approximation 
guarantees.
\begin{proposition}
    In the \TWOLFLP, there exists no polynomial-time algorithm
    with a $2-\epsilon$ approximation guarantee 
    under the unique game conjecture.
\end{proposition}
This hardness result is 
directly implied by
the hardness result of the vertex cover problem
\citep{KR-08}.
Note that the \TWOLFLP\ generalizes 
the weighted vertex cover problem.
See \Cref{apx:vertex cover}
for more discussion.

\paragraph{Single-location facility location problem.}
The \TWOLFLP\ generalizes 
the classical \emph{(single-location) metric facility location problem (MFLP)} \citep{STA-97}.
In particular, the former problem becomes 
the latter problem when we further impose 
the restriction on the problem instances such that
the home location of 
every individual
is the same as her workspace,
i.e., $\flowi = 0$ for every~$i\not= j$.

The \TWOLFLP\ can also be considered as 
a special case of 
the classical \emph{(single-location) non-metric facility location problem (NMFLP)} \citep{hoc-82}.
In particular, every problem instance 
in 
the \TWOLFLP\ is equivalent to a problem instance
in the NMFLP where we consider each edge $\edge$
as a new location
with opening cost $\infty$.\footnote{Recall that 
distance $\distance$ is defined as $
\distance(\edge, i) = \min\{\distance(\edgehome,i),
\distance(\edgework,i)\}$
which may not satisfy the triangle inequality. 
To complete the instance construction in NMFLP, 
we note that the distance $\distance(\edge,\edge')$
between location $\edge$ and $\edge'$
does not matter (and thus can be set arbitrarily), 
since
their opening costs are set to be $\infty$.}

The (single-location) MFLP and NMFLP
have been studied extensively in the literature. 
Under the standard computation complexity hardness assumption, 
it is known that there exists no polynomial-time algorithm
which can compute the optimal solution even for MFLP.

\section{The 2-Chance Greedy Algorithm for the 
\TWOLFLP}
\label{sec:2-Chance Greedy}

The main result of this paper is a $\greedyapprox$-approximation algorithm for the \TWOLFLP.
In this section, 
we describe this algorithm and  
all of its ingredients.
Its   approximation ratio analysis
with a novel strongly factor-revealing 
quadratic program 
is deferred
to \Cref{sec:greedy two location analysis}. 

\label{sec:greedy two location}

Our algorithm is a natural generalization of 
the classic greedy algorithm 
 designed by \citet{JMMSV-02} for the 
single-location metric facility location problem.
We first present an overview
of our algorithm,
followed by 
the formal description in 
\Cref{alg:greedy modified}.
Finally, we provide  intuition behind our algorithm
by comparing it with the classic greedy algorithm
for the
 single-location metric/non-metric facility
location problems. 

\paragraph{Overview of the algorithm.}
The \emph{2-Chance Greedy Algorithm}
can be described as
a continuous procedure,
where we gradually identify the facilities to open 
in
a \emph{greedy} fashion
and connect each edge $\edge \in\Edge$ to the opened facilities.
Each edge $\edge$ can be connected \emph{twice}\footnote{Therefore,
we name our algorithm \emph{2-Chance} Greedy Algorithm.}
through both its home $\edgehome$ 
and work location~$\edgework$.

In this algorithm, 
we use $\stores$ to denote the set of opened facilities.
We maintain
$\connectmapping(\edge,\edgeindex)$
to
record the opened facility to which 
edge $\edge\in\Edge$ is connected through its home/work location $\edge_\edgeindex$
for each $\edgeindex\in\hwset$.
Initially, we set $\connectmapping(\edge, \edgeindex) \gets \nullsymbol$,
meaning edge $\edge$ has not been connected to
any facility through home/work location $\edge_\edgeindex$ yet.
We further introduce an auxiliary variable $\firstchanceset
= \{\edge \in\Edge: \connectmapping(\edge,\Home) = \nullsymbol
\land \connectmapping(\edge,\Work) = \nullsymbol\}$ to denote 
the subset of edges which have not been connected to any facilities 
through their home nor work locations.
We say an edge $\edge$ is \emph{unconnected} 
if $\edge \in \firstchanceset$,
\emph{partially connected} if 
$\edge \not\in\firstchanceset$ 
but $\connectmapping(\edge,\Home) = \nullsymbol$ or
$\connectmapping(\edge,\Work) = \nullsymbol$,
and \emph{fully connected} if 
$\connectmapping(\edge,\Home) \not= \nullsymbol$ 
and $\connectmapping(\edge,\Work) \not= \nullsymbol$.
Finally, we also maintain an \emph{\candidatecosttext} 
$\greedycounter(\edge)$
for each edge $\edge \in\Edge$,
which is initialized to zero, and continuously increases over time\footnote{Suppose $\greedycounter_\timeindex(\edge)$
 	    is the value of $\greedycounter(\edge)$ at time
 	    stamp $\timeindex$,
 	    then the algorithm updates $\greedycounter_\timeindex(\edge)$
 	    with $\frac{\partial \greedycounter_\timeindex(\edge)}{
 	    \partial \timeindex} = 1$ for each edge $\edge\in\firstchanceset$.}
as long as edge $\edge$ is unconnected, i.e.,
$\edge\in\firstchanceset$. 
The algorithm terminates when all edges are partially or fully connected,
i.e., $\firstchanceset = \emptyset$.

There are two possible events which might happen
as we   increase \candidatecosttext~$\greedycounter(\edge)$
for each unconnected edge $\edge \in\firstchanceset$.

\begin{displayquote}
    \textbf{Event~(a)}: For an unconnected 
edge $\edge\in\firstchanceset$
the \candidatecosttext~$\greedycounter(\edge)$ equals
to its per-individual 
connection cost $\distance(\edge, i)$
for some \emph{opened} facility $i\in\stores$.
\end{displayquote}
When \textbf{Event~(a)} happens, 
we remove unconnected edge $e$ from $\firstchanceset$ and 
connect it with facility $i$ 
through its home/work location $\edge_\edgeindex$ 
(i.e., update $\connectmapping(\edge,\edgeindex) \gets i$)
if $\greedycounter(\edge) = \distance(\edge_\edgeindex, i)$
for each $\edgeindex\in\hwset$. By
definition,
$\greedycounter(\edge) \leq \distance(\edgehome, i)$,
$\greedycounter(\edge) \leq \distance(\edgework, i)$
and at least one of them holds with equality.

Before we define the second event, we introduce the notion of cost improvement. For
an unconnected edge $e\in U$ and 
an unopened facility $i$
we
define per-individual cost improvement as
$(\alpha(e) - d(e,i))^+$,
where  operator $\plus{x} \triangleq \max\{x, 0\}$.
Similarly, for a partially connected edge $e$, we define discounted per-individual
cost improvement as $(\gamma\alpha(e) - d(e,i))^+$ where
 discount factor $\discountfactor\in[0, 1]$ 
is a pre-specified parameter of the algorithm.
Intuitively, the per-individual cost improvement captures the reduction in the distance costs
(relative to \candidatecosttext) 
of an edge $e$ when a new facility $i$ opens and this edge is connected to it.
The second event, hereafter \textbf{Event~(b)} is as follows.

\begin{displayquote}
    \textbf{Event~(b)}: For an unopened location~$i\not\in\SOL$,
    its facility opening cost ${\opencost_i}$ scaled by ${\costscalar}$
    equals to 
    the total cost improvement associated with opening a facility at $i$. Namely,
\begin{align*}
    \displaystyle\sum\nolimits_{\edge\in \firstchanceset} 
 	    \flowe\cdot 
 	    \plus{\greedycounter(\edge) - \distance(\edge, i)}
 	    +
 	    \displaystyle\sum\nolimits_{
 	    \subalign{
 	    &\edge\not\in \firstchanceset,\\
 	    &\edgeindex\in\hwset:
 	    \connectmapping(\edge,\edgeindex) = \nullsymbol
 	    }}
 	    \flowe\cdot 
 	    \plus{\discountfactor\cdot \greedycounter(\edge) - \distance(\edge_\edgeindex,i)}
 	    = {\costscalar\cdot \opencost_i}{}
\end{align*}
\end{displayquote}
Here $\discountfactor$ and $\costscalar$ are two pre-specified parameters of the algorithm. The discount factor $\discountfactor$ controls the impact of partially connected individuals in opening new facility,
while the opening cost scalar $\costscalar$ controls the tradeoff between facility opening cost and the individual connection cost in the solution.

When \textbf{Event~(b)} happens,
we open facility $i$ 
(i.e., update $\stores \gets 
\stores \cup \{i\}$),
and connect each unconnected edge $\edge\in\firstchanceset$
(resp.\ partially connected edge $\edge\not\in \firstchanceset:
\connectmapping(\edge,\edgeindex) = \nullsymbol$ for some 
$\edgeindex\in\hwset$)
with facility $i$ if 
$\greedycounter(\edge) \geq \distance(\edge, i)$
(resp.\ $\greedycounter(\edge) \geq \distance(\edge_\edgeindex, i)$).
See \Cref{alg:greedy modified} for a formal description.

\begin{algorithm}
    \SetKwInOut{Input}{input}
    \SetKwInOut{Output}{output}

    \Input{discount factor $\discountfactor\in[0, 1]$,
    opening cost scalar $\costscalar\in\reals_+$,    \TWOLFLP\ instance $\instance = 
    (n, \distance, 
    \{\flowe\}_{\edge\in\Edge}, \opencosts_{i\in[n]})$}
    \Output{subset $\stores \subseteq[n]$
    as the locations of opened facilities.}
 	\caption{2-Chance Greedy Algorithm}
 	\label{alg:greedy modified}
 	
 	initialize $\stores \gets \emptyset$
 	
 	initialize $\greedycounter(\edge) \gets 0$ for each edge $\edge\in\Edge$
 	
 	
 	initialize $\firstchanceset \gets \Edge$,

    initialize $\connectmapping(\edge,\edgeindex) \gets \nullsymbol$
    for each edge $\edge \in\Edge$,
    each $\edgeindex\in\hwset$
 	
 	\While{$\firstchanceset \not= \emptyset$}{ 
 	    increase $\greedycounter(\edge)$ by 
 	    $d\greedycounter$
 	    for every edge $\edge \in \firstchanceset$

 	    \tcc{Event (a)}
 	    
 	    \While{there exists edge $\edge \in \firstchanceset$
 	    and location $i\in\stores$ 
 	    s.t.\ $\greedycounter(\edge) = \distance(\edge, i)$ 
 	    }
 	    {
 	    
 	  
 	    remove edge $\edge$ from $\firstchanceset$,
 	    i.e., $\firstchanceset\gets \firstchanceset\backslash\{\edge\}$
 	    
 	     \For{each $\edgeindex \in \hwset$}{
 	    \If{$\greedycounter(\edge) = \distance(\edge_\edgeindex,i)$}{
 	    
 	    connect edge $\edge$ and facility $i$ through 
      home/work location $\edge_\edgeindex$, i.e., 
 	    $\connectmapping(\edge,\edgeindex) \gets i$
 	    }
 	    }

 	    }

 	    \tcc{Event (b)}
 	    
 	    \While{there exists location $i\not\in\stores$ s.t.\
 	    $\sum_{\edge\in \firstchanceset} 
 	    \flowe\cdot 
 	    \plus{\greedycounter(\edge) - \distance(\edge, i)}
 	    +
 	    \sum_{
 	    \subalign{
 	    &\edge\not\in \firstchanceset,\\
 	    &\edgeindex\in\hwset:
 	    \connectmapping(\edge,\edgeindex) = \nullsymbol
 	    }}
 	    \flowe\cdot 
 	    \plus{\discountfactor\cdot \greedycounter(\edge) - \distance(\edge_\edgeindex,i)}
 	    = \costscalar\cdot \opencost_i$
 	    }{
 	    
 	    \vspace{1mm}
 	    
 	    add location $i$ into $\stores$, i.e.,
 	    $\stores \gets \stores \cup \{i\}$

 	    \For{each edge $\edge\not\in \firstchanceset$}{
 	    \For {each $\edgeindex \in \hwset$}{
 	  
 	    \If{$\connectmapping(\edge,\edgeindex) = \nullsymbol$ 
 	    and $\discountfactor\cdot \greedycounter(\edge) \geq \distance(\edge_\edgeindex, i) $}{
 	  
 	    connect edge $\edge$ and facility $i$ through 
      home/work location 
      $\edge_\edgeindex$, i.e., 
 	    $\connectmapping(\edge,\edgeindex) \gets i$
 	   
 	    }
 	    }
 	    }
 	    
 	    \For{each edge $\edge\in \firstchanceset$}{
 	    
 	    \If{$\greedycounter(\edge) - \distance(\edge, i)\geq 0$}{
 	    
 	    remove edge $\edge$ from $\firstchanceset$,
 	    i.e., $\firstchanceset\gets \firstchanceset\backslash\{\edge\}$
 	    
 	    \For{each $\edgeindex \in \hwset$}{
 	    \If{$\greedycounter(\edge) \geq \distance(\edge_\edgeindex,i) $}{
 	    
 	    connect edge $\edge$ and facility $i$ through 
      home/work location 
      $\edge_\edgeindex$, i.e., 
 	    $\connectmapping(\edge,\edgeindex) \gets i$
 	    }
 	    }
 	    }
 	    }
 	    }
 	   
 	}
 	\textbf{return} $\stores$
\end{algorithm}

Although \Cref{alg:greedy modified}
is described as a deterministic continuous time procedure (due to the 
continuous update of  
\candidatecosttext~$\{\greedycounter(\edge)\}$),
it is straightforward to see that 
the algorithm 
eventually terminates by construction.
Furthermore, it 
can be executed
in polynomial time. 
Note that 
Events~(a) and (b) 
(i.e., the conditions in 
two while loops in lines 7 and 12)
happen
$O(n^2)$ times. Thus, at any stage of the algorithm,
with polynomial running time,
we can identify the next time stamp in which 
Event (a) or (b) happens, and
update 
all variables (i.e, $\stores$, $\{\connectmapping(\edge,\edgeindex)\}$,
$\firstchanceset$, $\{\greedycounter(\edge)\}$)
accordingly.

\paragraph{Connection to the JMMSV algorithm.}
For the classic single-location facility location problem,
\citet{JMMSV-02} design a greedy-style algorithm.
Throughout the paper, we denote this algorithm by the {JMMSV} algorithm.\footnote{\citet{JMMSV-02} and
follow-up works \citep[e.g.,][]{MYZ-06} 
introduce additional technical modifications to 
the JMMSV algorithm and achieve better approximation guarantees. 
This paper mainly compares the JMMSV algorithm with 
the 2-Chance Greedy Algorithm.
Whether similar technique can be applied to
the \TWOLFLP, is left as an future direction.}
For completeness, we include a formal description of
the JMMSV algorithm in \Cref{apx:single location FLP}.
Loosely speaking, the JMMSV algorithm 
and 
the 2-Chance Greedy Algorithm 
(\Cref{alg:greedy modified})
use the same greedy procedure
to identify facilities to open.
However, the JMMSV algorithm connects each individual to a single facility, 
while \Cref{alg:greedy modified}
connects each individual to one or two locations (through both 
her home or work location).
For the single-location metric facility location problem
(MFLP) where the distance function satisfies the triangle inequality,
\citet{JMMSV-02} show that the JMMSV algorithm 
is a 1.861-approximation.\footnote{In \Cref{sec:strongly factor-revealing program MFLP}, 
we 
further show that the JMMSV algorithm is indeed a 
$\greedyapproxMFLP$-approximation in the MFLP
(\Cref{prop:approx ratio general gamma MFLP}).
We use this 
as a warm-up exercise to illustrate 
our technique (i.e., strongly factor-revealing quadratic program)
for the $\greedyapprox$-approximation guarantee of \Cref{alg:greedy modified}
in the \TWOLFLP.}

As we mentioned in \Cref{sec:prelim}, the \TWOLFLP\ 
can also be thought as a special case of 
the single-location non-metric facility location problem (NMFLP),
where we 
create a new location for each edge
(but the distance between 
the original edges and new locations and facilities may no longer 
satisfy the triangle inequality). 
In this sense, the JMMSV algorithm is also well-defined 
for the \TWOLFLP.
Moreover,
\Cref{alg:greedy modified}
with discount factor $\discountfactor = 0$ and opening cost scalar $\costscalar = 1$
recovers the JMMSV algorithm for general \TWOLFLP\ instances.
To see this, observe that 
in \Cref{alg:greedy modified} with 
$\discountfactor = 0$,
all partially connected edges have no impact on 
whether new facilities will be opened at any time stamp,
and
every edge will be connected to exactly one facility 
in the end.
\begin{observation}
\label{ob:JMMSV equivalence TWOLFLP}
In the \TWOLFLP, 
the 2-Chance Greedy Algorithm with discount factor $\discountfactor = 0$ and opening cost scalar $\costscalar = 1$
is equivalent to the JMMSV algorithm.
\end{observation}

It is known that for general NMFLP instances, 
the JMMSV algorithm is not constant-approximation \citep{hoc-82}.
Here we present an example 
to illustrate that 
\Cref{alg:greedy modified} 
with $\discountfactor = 0$ and $\costscalar = 1$, thus
the JMMSV algorithm, is an $\Omega(n)$-approximation in the \TWOLFLP.

\begin{example}
\label{example:JMMSV failure}
Given an arbitrary $n_0\in\naturals$, 
consider a \TWOLFLP\ instance as follows:
There are $n \triangleq n_0 + 1$ locations,
and $n_0$ 
individuals. Each individual $i\in[n_0]$ 
resides at location $i$.
All of them work at location $n_0 + 1$.
Namely, the numbers of individuals $\{\flowe\}_{\edge\in\Edge}$
are
\begin{align*}
    \flowe = 
    \left\{
    \begin{array}{ll}
    1   & \quad \text{if $\edgehome\in[n_0],\edgework = n_0 + 1$} \\
    0     & \quad \text{o.w.} 
    \end{array}
    \right.
\end{align*}
The facility opening costs $\{\opencost_i\}_{i\in[n]}$ and  
distance function $\distance(\cdot,\cdot)$ are 
\begin{align*}
    \opencost_i =
    \left\{
    \begin{array}{ll}
    \frac{1}{n_0 - i + 1}  - \epsilon    & \quad \text{if $i\in[n_0]$} \\
    1     & \quad \text{if $i = n_0 + 1$} 
    \end{array}
    \right.
    \qquad 
    \distance(i, j) = 
    \left\{
    \begin{array}{ll}
    \sfrac{1}{\costscalar}   & \quad \text{if $i \not = j$} \\
    0     & \quad \text{if $i = j$} 
    \end{array}
    \right.
\end{align*}
where $\epsilon \in \reals_+$ is a sufficiently small positive constant.

The optimal solution opens a facility at the common work location
(i.e., $\OPT = \{n_0 + 1\}$) 
with optimal total cost $\Cost{\OPT} = 1$.

In contrast, consider the 2-Chance Greedy Algorithm 
with $\discountfactor = 0$ and $\costscalar\in\reals_+$.
At time stamp $\frac{1}{n_0 - 1 + 1} - \epsilon$,
for each edge $\edge$,
the \candidatecosttext~$\greedycounter(\edge) = \frac{1}{n_0 - 1 + 1} - \epsilon$, and the condition of Event (b) for opening
facility 1 is satisfied due to individual~1 residing at location 1.
Thus, the algorithm opens facility $1$ and 
partially connects individual $1$ to facility $1$ through her home.
Then, at time stamp $\frac{1}{n_0 - 2 + 1} - \epsilon$,
the \candidatecosttext~$\greedycounter(\edge) = \frac{1}{n_0 - 2 + 1} - \epsilon$
for every edge $\edge$ except for the partially connected individual $1$,
thus
facility $2$ is opened, 
and individual~$2$ is partially connected to facility $2$. Proceeding similarly, it can be seen that
  when the algorithm terminates,
it
opens $n_0$ facilities at every individual's home (i.e., $\stores = \{1, \dots, n_0\}$) with total cost $\Cost{\stores} = \Theta(\log(n)) - n_0\cdot \epsilon$. 

Putting all pieces together, as $\epsilon$ goes to zero, the approximation of the 2-Chance Greedy Algorithm 
with $\discountfactor = 0$ converges to $\Theta(\log(n))$
for this example.
\end{example}

\Cref{example:JMMSV failure} also illustrates 
the necessity of positive 
discount factor $\discountfactor$
of \Cref{alg:greedy modified}
to achieve a constant approximation in the \TWOLFLP.
Specifically,
when the discount factor $\discountfactor$ is set as zero,
\Cref{alg:greedy modified}
never
opens facility $n_0 + 1$ in the optimal solution, 
since the facility opening cost at location $n_0 + 1$
never meets the total cost improvement induced by \candidatecosttext s $\{\greedycounter(\edge)\}$.
In particular, the \candidatecosttext s for all partially connected individuals 
make no positive contribution to the total cost improvement since the discount factor $\discountfactor = 0$.
Note that this issue persists 
if the discount factor $\discountfactor = o(\frac{\epsilon}{\log(n)})$.
On the other hand, it is straightforward to verify that
\Cref{alg:greedy modified} achieves a constant approximation 
 for \Cref{example:JMMSV failure}
when the discount factor $\discountfactor$ is set to a positive constant.\footnote{As a sanity check, 
when $\epsilon \leq \frac{\opencost_1}{n_0 - 1}$,
\Cref{alg:greedy modified} with discount factor $\discountfactor = 1$
returns solution $\stores = \{1, n_0 + 1\}$ with total cost 
$\Cost{\stores} = 1 + \Theta(\frac{1}{n})$,
and thus is an $(1 + \Theta(\frac{1}{n}))$-approximation 
for \Cref{example:JMMSV failure}.}
In the next subsection,
we further show its sufficiency
by proving that 
\Cref{alg:greedy modified} with positive constant
discount factor $\discountfactor$
is a constant-approximation for all \TWOLFLP\ instances,
e.g., a $\greedyapprox$-approximation for $\discountfactor = 1$
(\Cref{prop:approx ratio gamma one}, \Cref{prop:approx ratio general gamma}).\footnote{In practice, 
the social planner can implement the 2-Chance Greedy Algorithm
by varying discount factor $\discountfactor$ and opening cost scalar $\costscalar$,
and then select the solution with the minimum cost.}

We finish this subsection by noting that 
the discount factor $\discountfactor$ in \Cref{alg:greedy modified}
plays no role when we restrict attention to MFLP instances.
As we mentioned in \Cref{sec:prelim},
the MFLP can be thought of as a special case 
of the \TWOLFLP\ where every individual's home is equivalent
to her work  location, i.e., $\flowe = 0$ if 
$\edgehome \not=\edgework$.
Therefore, \Cref{alg:greedy modified} is well-defined 
for MFLP instances.
Moreover, observe that 
for MFLP instances,
no edge $\edge$ with positive $\flowe > 0$
is ever partially connected in the algorithm.
Instead, 
both its home $\edgehome$ and work location $\edgework$
will be connected to the same facility in the same period.
Therefore, 
for any fixed MFLP instance,
\Cref{alg:greedy modified}
identifies the same set of facilities $\stores$ 
regardless of its discount factor $\discountfactor$.
Thus, in this case for any discount factor $\discountfactor\in[0, 1]$,
\Cref{alg:greedy modified} with $\costscalar = 1$
is equivalent to the JMMSV algorithm.

\begin{observation}
\label{ob:JMMSV equivalence MFLP}
In the MFLP, 
the 2-Chance Greedy Algorithm with an arbitrary 
discount factor $\discountfactor \in [0, 1]$ and opening cost scalar $\costscalar = 1$
is equivalent to the JMMSV algorithm.
\end{observation}

To summarize our discussion, 
the 2-Chance Greedy Algorithm
gets the best of both worlds:
In the classic MFLP, 
for any discount factor $\discountfactor\in[0, 1]$ and opening cost scalar $\costscalar = 1$,
the 2-Chance Greedy Algorithm recovers the JMMSV algorithm 
and thus is a $\greedyapproxMFLP$-approximation.
In the \TWOLFLP,
while the JMMSV algorithm degenerates to
an $\Omega(\log(n))$-approximation,
the 2-Chance Greedy Algorithm with a positive constant 
discount factor  
$\discountfactor$ 
remains a constant-approximation, e.g., $\greedyapprox$-approximation for $\discountfactor = 1$ and $\costscalar = 2$, as the next section formalizes.

\section{Obtaining an Approximation Ratio via Strongly 
Factor-Revealing Quadratic Programs}
\label{sec:greedy two location analysis}

In this section, we analyze the approximation ratio
of the 2-Chance Greedy Algorithm. 
To obtain our result, we first introduce 
the following quadratic program \ref{eq:strongly factor-revealing quadratic program}
parameterized by 
$\nHat\in\naturals$, $\discountfactor \in[0, 1]$, and $\costscalar\in[1, 1 + \discountfactor]$:
\begin{align}
\tag{$\SFRP{\nHat,\discountfactor,\costscalar}$}
\label{eq:strongly factor-revealing quadratic program}
&\arraycolsep=1.4pt\def\arraystretch{2.2}
    \begin{array}{llll}
    \max\limits_{
    \substack{\costHat\geq 0,
    \\
    \qHatBf,
    \muHatBf,
    \distanceHatBf,
    \distanceHatStarBf \geq 
    \zerobf
    }} &
     \displaystyle\sum\nolimits_{a\in[\nHat],b\in[a]}
     \qHat(a,b)\cdot 
    \left(\frac{1 + \discountfactor}{\costscalar}
    \cdot
    \muHat(a,b) - 
    \left(
    \frac{1 + \discountfactor}{\costscalar} - 1\right)\cdot\distanceHatStar(a,b)\right)
    & \text{s.t.} & \\
    \SFRCmonotonicity &
    \muHat(a, b) \leq \muHat(a', b')
    &
    a, a'\in[\nHat],
    b\in[a], b'\in[a'], b < b'
    \\
    \SFRCtriangleineq &
    \discountfactor\cdot \muHat(a', b') \leq 
    \distanceHatStar(a, b) + \distanceHat(a, b)
    +
    \distanceHat(a', b')
    &
    a,a'\in[\nHat],
    b\in[a], b'\in[a'],
    a < a'
    &
    \\
    \SFRCcontribution &
    \displaystyle\sum\nolimits_{\substack{a' \in [a + 1:\nHat]\\
    b' \in [a]}}
    \qHat(a',b')\cdot \plus{\discountfactor\cdot \muHat(a', b') - \distanceHat(a',b')}
    &
    \\
    &
    \quad
    +\displaystyle\sum\nolimits_{
    \substack{a' \in [a+1:\nHat]\\
    b' \in [a + 1:a']}}
    \qHat(a',b')\cdot \plus{\discountfactor\cdot \muHat(a, a) - \distanceHat(a',b')}
    \leq \costscalar\cdot \costHat
    \quad
    &
    a\in[\nHat - 1]
    \\
    \SFRCdistance &
    \distanceHat(a, b) \leq \muHat(a, b)
    &
    a\in[\nHat],
    b\in[a]
    \\
    \SFRCconnectcost &
    \distanceHatStar(a, b) \leq \muHat(a, b)
    &
    a\in[\nHat],
    b\in[a]
    \\
     \SFRCtotalcost &
     {\costHat}
     +
     \displaystyle\sum\nolimits_{a\in[\nHat],b\in[a]}
     \qHat(a,b)\cdot \distanceHat(a,b)
     \leq 1
     &
     \\
     \SFRCdensity &
     \displaystyle\sum\nolimits_{a \in [b:\nHat]}
     \qHat(a, b) = 1
     &
     b\in[\nHat]
    \end{array}
\end{align}
Given parameters $\nHat\in\naturals$
and $\discountfactor\in[0, 1]$,
program
\ref{eq:strongly factor-revealing quadratic program}
with non-negative variables\footnote{We reuse notation $\distance,\costHat$ of the \TWOLFLP\
and $\greedycounter$ of the 2-Chance Greedy Algorithm in program~\ref{eq:strongly factor-revealing quadratic program},
since the program is constructed from the analysis of the algorithm.}
$\costHat$, $\{\qHat(a, b), \muHat(a, b),\distanceHat(a, b),
\distanceHatStar(a, b)\}_{a\in[\nHat], b \in [a]}$
maximizes its quadratic objective 
over linear/quadratic constraints \SFRCmonotonicity --\SFRCdensity.\footnote{By introducing  auxiliary variables and inequalities,
we 
can change constraint~\SFRCcontribution\ with a quadratic constraint.}
At a high level, constraint~\SFRCtriangleineq\ is derived from the triangle inequality over locations and the condition of Event~(a) in the 2-Chance Greedy Algorithm, and
\SFRCcontribution\ is derived from the condition of Event~(b) in the algorithm.
Both the objective function and the constraints are explained 
in more detail in \Cref{sec:primal dual framework}
and \Cref{sec:strongly factor-revealing program}.
Given any program~$\mathcal{P}$,
we denote its optimal objective value by 
$\Obj{\mathcal{P}}$.

We are now ready to present the main result of this paper.
\begin{restatable}{theorem}{greedymodified}
\label{thm:greedy modified}
In the \TWOLFLP,
the approximation ratio of 
the 2-Chance Greedy Algorithm
with discount factor $\discountfactor\in[0, 1]$ and opening cost scalar $\costscalar\in[1, 1 + \discountfactor]$
is 
at most equal to $\SFRapproxratio$,
where 
\begin{align*}
    \SFRapproxratio = 
    \inf\limits_{\nHat\in\naturals}
    \Obj{\text{\ref{eq:strongly factor-revealing quadratic program}}}
\end{align*}
\end{restatable}
We emphasize that 
approximation ratio $\SFRapproxratio$
in the above theorem is taking the \emph{infimum} of 
$\Obj{\text{\ref{eq:strongly factor-revealing quadratic program}}}$
over all possible $\nHat\in\naturals$.
Therefore, 
$\SFRapproxratio$
is
upperbounded by 
$\Obj{\text{\ref{eq:strongly factor-revealing quadratic program}}}$
for every $\nHat\in\naturals$.
This is different from the literature \citep[e.g.,][]{JMMSV-02,MYZ-06}
where 
the approximation ratio is upperbounded by the 
\emph{supremum} of a family of factor-revealing programs over its parameter space.
Therefore, we refer to program~\ref{eq:strongly factor-revealing quadratic program}
as the \emph{strongly} factor-revealing quadratic program
to highlight 
this distinction.

In practice, 
the social planner can implement the 2-Chance Greedy Algorithm
by varying the discount factor $\discountfactor$ and the opening cost scalar $\costscalar$,
and then select the solution with the minimum cost. See \Cref{sec:numerical} for more discussion.
Nonetheless, to obtain a theoretical approximation upperbound, we numerically evaluate
program $\Obj{\SFRP{25,\discountfactor,\costscalar}}$ for every $\discountfactor \in \{0.1, 0.2, \dots, 1\}$ and $\costscalar\in\{1, 1.1, \dots, 1 + \discountfactor\}$ with software Gurobi. Among all parameter choices in this grid search, we observe that the 2-Chance Greedy Algorithm with $\discountfactor = 1$ and $\costscalar = 2$ attains the best numerical approximation upperbound as follows.

\begin{restatable}{proposition}{approxratiogammaone}
\label{prop:approx ratio gamma one}
    In the \TWOLFLP,
    the approximation ratio of 
    the 2-Chance Greedy Algorithm
with discount factor $\discountfactor=1$ and opening cost scalar $\costscalar = 2$
is at most equal to 
$\SFRapproxratio \leq 
\Obj{\SFRP{25,1,2}} \leq \greedyapprox$.
\end{restatable}

We complement \Cref{prop:approx ratio gamma one} with the following 
lower bound, and defer its 
formal proof to \Cref{apx:approx ratio gamma one lower bound}.
\begin{restatable}{proposition}{approxratiogammaonelowerbound}
\label{prop:approx ratio gamma one lower bound}
There exists a \TWOLFLP\ instance
such that the approximation ratio of 
the 2-Chance Greedy Algorithm with discount factor $\discountfactor = 1$
and opening cost scalar $\costscalar = 2$
is at least $\greedyapproxLB$.
\end{restatable}

In the end of \Cref{sec:greedy two location}, we discuss the necessity of having a positive constant discount factor $\discountfactor$ to obtain a constant approximation ratio in the \TWOLFLP. Here, we echo this insight by showing the following analytical but loose approximation upperbounds. See its formal proof in \Cref{apx:approx ratio general gamma}.\footnote{Interestingly, the approximation upperbound in \Cref{prop:approx ratio general gamma} does not depend on $\costscalar\in[1, 1 + \discountfactor]$. However, this upperbound is loose, and picking a proper $\costscalar$ is important both for our theoretical result (\Cref{prop:approx ratio gamma one}) and our numerical experiment (\Cref{sec:numerical}).}

\begin{restatable}{proposition}{approxratiogeneralgamma}
\label{prop:approx ratio general gamma}
    In the \TWOLFLP,
    the approximation ratio of 
    the 2-Chance Greedy Algorithm 
with discount factor $\discountfactor\in[0, 1]$ and 
opening cost scalar $\costscalar \in [1, 1 + \discountfactor]$
is at most equal to 
$\SFRapproxratio \leq 
\Obj{\SFRP{2,\discountfactor,\costscalar}} \leq \frac{6(1+\discountfactor)}{\discountfactor^2}$.
\end{restatable}

In the remainder of this section, 
we sketch the proof of \Cref{thm:greedy modified}
and defer some technical details to \Cref{apx:missing proofs}.
In \Cref{sec:primal dual framework},
we first use a primal-dual framework 
to upperbound the approximation ratio 
of the 2-Chance Greedy Algorithm as 
the supremum over
a different
(and more complicated -- see \ref{eq:weakly factor-revealing program})
family of 
factor-revealing programs,
which are
non-linear, non-quadratic, and non-convex.
Then,
we upperbound this family of 
factor-revealing programs
through the strongly factor-revealing quadratic 
program~\ref{eq:strongly factor-revealing quadratic program}
and complete the proof of \Cref{thm:greedy modified}
in \Cref{sec:strongly factor-revealing program}.

\subsection{Construction of factor-revealing program~\ref{eq:weakly factor-revealing program}}
\label{sec:primal dual framework}

In this part,
we use the primal-dual analysis framework
(initially developed by \citealp{JMMSV-02}
for the JMMSV algorithm in the MFLP)
and 
provide  
three lemmas 
that shed light on the structure of  the 2-Chance Greedy Algorithm. We then leverage these to obtain our main approximation result  for \ref{eq:weakly factor-revealing program}, stated next:
\begin{lemma}
\label{lem:weakly factor-revealing program}
In the \TWOLFLP,
the approximation ratio of 
the 2-Chance Greedy Algorithm 
with discount factor $\discountfactor\in[0, 1]$
and opening cost scalar $\costscalar \in[1, 1 + \discountfactor]$
is 
at most equal to $\WFRapproxratio$ where
\begin{align*}
    \WFRapproxratio=
    \sup\limits_{\mHat\in\naturals, 
\permu
:[\mHat]\rightarrow\positivereals}
\Obj{\text{\ref{eq:weakly factor-revealing program}}}
\end{align*}
Here \ref{eq:weakly factor-revealing program}
is the maximization program parameterized by 
$\mHat\in\naturals$ 
and $\permu:[\mHat]\rightarrow\positivereals$
defined as follows:
\begin{align}
\tag{$\WFRP{\mHat,\permu,\discountfactor,\costscalar}$}
\label{eq:weakly factor-revealing program}
&\arraycolsep=1.4pt\def\arraystretch{2.2}
    \begin{array}{llll}
    \max\limits_{
    \substack{\costHat\geq 0,
    \\
    \muHatBf,
    \distanceHatBf,
    \distanceHatStarBf \geq 
    \zerobf
    }} &
     \displaystyle\sum\nolimits_{\ell\in[\mHat]}
     \frac{1 + \discountfactor}{\costscalar}
    \cdot
    \muHat(\ell) - 
    \left(\frac{1 + \discountfactor}{\costscalar} - 1\right)\cdot\distanceHatStar(\ell)
    & \text{s.t.} & \\
    \WFRCtriangleineq &
    \discountfactor\cdot \muHat(\ell') \leq 
    \distanceHatStar(\ell) + \distanceHat(\ell)
    +
    \distanceHat(\ell')
    &
    \ell,\ell'\in[\mHat],
    \permu(\ell) < \permu(\ell')
    &
    \\
    \WFRCcontribution &
    \displaystyle\sum\nolimits_{\ell'\in[\mHat]:
    \permu(\ell') \geq \permu(\ell)}
    \plus{\discountfactor\cdot 
    \min\left\{\muHat(\ell), \muHat(\ell')\right\} - \distanceHat(\ell')}
    \leq \costscalar\cdot \costHat\quad
    &
    \ell \in[\mHat]
    \\
    \WFRCconnectcost &
    \distanceHatStar(\ell) \leq \muHat(\ell)
    &
    \ell\in[\mHat]
    \\
     \WFRCtotalcost &
     {\costHat}{} 
     +
     \displaystyle\sum\nolimits_{\ell\in[\mHat]}
     \distanceHat(\ell)
     \leq 1
     &
    \end{array}
\end{align}
\end{lemma}
In the above lemma,
program~\ref{eq:weakly factor-revealing program}
is parameterized by a natural number $\mHat\in\naturals$,
and function~$\permu\in\positivereals^{\mHat}$ which maps 
$[\mHat]$ to ${\positivereals}$.
It has non-negative variables 
$\costHat$, $\{\muHat(\ell),\distanceHat(\ell),
\distanceHatStar(\ell)\}_{\ell\in[\mHat]}$.
Function $\permu$ as a parameter of program~\ref{eq:weakly factor-revealing program}
essentially specifies 
a total preorder\footnote{Total preorder 
is a binary relation 
satisfying reflexivity, transitivity, and is total.}
over $[\mHat]$, each of which in turn corresponds to different sets
of constraints in~\WFRCtriangleineq\ and \WFRCcontribution.

At a high level, constraint~\WFRCtriangleineq\ comes from the triangle inequality over locations and the condition of Event~(a) of the algorithm,
\WFRCcontribution\ comes from the condition of Event~(b), 
\WFRCtotalcost\ comes from the algorithm construction, and \WFRCtotalcost\ is a normalization.
It is worth highlighting that all constraints other than constraint~\WFRCcontribution ~are linear, as is the objective function.

As we discussed after \Cref{thm:greedy modified},
though \Cref{lem:weakly factor-revealing program}
already provides an upperbound $\WFRapproxratio$ 
on the approximation ratio
of the 2-Chance Greedy Algorithm, 
$\WFRapproxratio$
is hard to be evaluated,  since it is the supremum of 
non-linear/non-quadratic/non-convex\footnote{Due to term $\min\{\muHat(\ell),\muHat(\ell')\}$,
it is hard to change 
constraint~\WFRCcontribution\ with a linear/quadratic constraint by 
introducing auxiliary variables and inequalities.
Furthermore, because of constraint~\WFRCcontribution, 
the set of feasible solutions is not convex.}
programs~\{\ref{eq:weakly factor-revealing program}\}
over all possible $\mHat\in\naturals$ and $\permu\in \positivereals^{\mHat}$.
In \Cref{sec:strongly factor-revealing program},
we thus upperbound this family of factor-revealing 
programs
~\{\ref{eq:weakly factor-revealing program}\}
through the strongly factor-revealing programs~\{\ref{eq:strongly factor-revealing quadratic program}\},
which can be both analytically analyzed (albeit leading to loose bounds)
and numerically computed.

To prove \Cref{lem:weakly factor-revealing program},
we start by identifying the structural properties of
the feasible execution paths generated by the 2-Chance Greedy Algorithm.
In this structural lemma, we introduce two auxiliary notations 
$\permuTGD(\cdot,\cdot)$,
$\{\config_i(\cdot)\}_i$
as follows.
For each edge $\edge\in\Edge$ and $\edgeindex\in \hwset$,
let
$\permuTGD(\edge,\edgeindex)$
denote the time stamp when edge $\edge$ is connected 
to facility $\connectmapping(\edge,\edgeindex)$
through home/work location $\edge_\edgeindex$
in the 2-Chance Greedy Algorithm. 
If $\connectmapping(\edge,\edgeindex) = \nullsymbol$,
we let $\permuTGD(\edge,\edgeindex)$ be the time stamp of the termination of 
the algorithm.
For each location $i\in[n]$, 
let
$\config_i:\Edge\rightarrow\hwset$
be the mapping such that 
$\config_i(\edge) = h$
if and only if $\distance(\edgehome, i) <
\distance(\edgework, i)$.

\begin{restatable}[Structural properties of the 2-Chance Greedy Algorithm]{lemma}
{structural}
\label{lem:structural lemma}
Given any \TWOLFLP\ instance, 
after the termination of the 2-Chance Greedy Algorithm
with discount factor $\discountfactor \in [0, 1]$ and opening cost scalar $\costscalar\in\reals_+$:
\begin{enumerate}
\item [(i)]
for every location $i\in[n]$,
edges $\edge$, $\edge'\in \Edge$, and $\edgeindex$, 
$\edgeindex'\in\hwset$, 
if $\permuTGD(\edge,\edgeindex) < \permuTGD(\edge',\edgeindex')$,
then $\connectmapping(\edge,\edgeindex) \not=\nullsymbol$
and 
\begin{align*}
    \discountfactor\cdot \greedycounter(\edge') \leq 
\distance(\edge_\edgeindex, \connectmapping(\edge,\edgeindex))
+
\distance(\edge_\edgeindex, i)
+
\distance(\edge'_{\edgeindex'}, i)
\end{align*}

\item [(ii)] 
for every location $i\in[n]$, and edge $\edge\in \Edge$,
\begin{align*}
    \sum
    _{\edge'\in\Edge:\permuTGD(\edge',\config_i(\edge')) 
    \geq 
    \permuTGD(\edge,\config_i(\edge))}
    \plus{
    \discountfactor\cdot 
    \min\left\{
    \greedycounter(\edge),
    \greedycounter(\edge')
    \right\}
    -
    \distance(\edge'_{\config_i(\edge')}, i)
    }
    \leq
    \costscalar\cdot \opencost_i
\end{align*}

\item [(iii)] for every edge $\edge\in\Edge$
and $\edgeindex\in\hwset$,
if $\connectmapping(\edge,\edgeindex) \not=\nullsymbol$,
then 
\begin{align*}
\distance(\edge_\edgeindex, \connectmapping(\edge,\edgeindex))
\leq 
    \greedycounter(\edge)
\end{align*}
\end{enumerate}
\end{restatable}
At a high-level, property~(i) exploits the triangle inequality 
(of metric distance $\distance$ over locations) and condition of Event~(a)
in the algorithm;
property~(ii) exploits the condition of Event~(b) in the algorithm;
and property~(iii) is implied directly from
the updating rule of $\connectmapping(\edge,\edgeindex)$.
We defer the formal proof of \Cref{lem:structural lemma}
to \Cref{apx:structural lemma}.

\begin{remark}
Properties~(i) -- (ii) in \Cref{lem:structural lemma}
share similar formats as
constraints~\WFRCtriangleineq--\WFRCconnectcost\
of program~\ref{eq:weakly factor-revealing program}.
Loosely speaking, program~\ref{eq:weakly factor-revealing program},
uses constraints~\WFRCtriangleineq--\WFRCconnectcost\
to capture possible execution paths in the 2-Chance Greedy Algorithm.
As we will see subsequently, its objective and constraint~\WFRCtotalcost\ 
is relevant for characterizing the approximation guarantee for a feasible execution path.
\end{remark}

\begin{remark}
\label{remark:misalignment of greedy counter}
In the MFLP, \citet{JMMSV-02} identify a 
structural lemma with similar properties~(i) and (ii)
for the JMMSV algorithm.
The main difference, which becomes 
a new significant technical challenge  in our setting,
is the \emph{misalignment between 
$\greedycounter(\edge)$ and
$\permuTGD(\edge,\edgeindex)$}.
Specifically, 
in the 2-Chance Greedy Algorithm
with strictly positive discount factor $\discountfactor \in (0, 1]$
in the \TWOLFLP,
each edge $\edge$ may be connected to two facilities,
and thus $\greedycounter(\edge)$
is not monotone with respect to the total preorder
specified by 
$\permuTGD(\edge,\edgeindex)$.
Due to the non-monotonicity of $\greedycounter(\edge)$,
the property~(ii) in the structural lemma and
constraint \WFRCcontribution\ in the factor-revealing program~\ref{eq:weakly factor-revealing program}
include
a non-linear term $\min\{\greedycounter(\edge), \greedycounter(\edge')\}$.
Consequently, program~\ref{eq:weakly factor-revealing program}
becomes harder to analyze.
We side step this difficulty
through a major technical contribution of~\Cref{sec:strongly factor-revealing program},
which allows us to obtain a
 strongly factor-revealing quadratic  program~\ref{eq:strongly factor-revealing quadratic program} that can be both analytically analyzed and numerically computed.
By contrast, 
the JMMSV algorithm
only connects each edge to a single facility,
and thus $\greedycounter(\edge) \equiv
\permuTGD(\edge,\edgeindex)$ 
if $\connectmapping(\edge,\edgeindex)\not=
\nullsymbol$.
Therefore, the monotonicity of $\greedycounter(\edge)$
with respect to $\permuTGD(\edge,\edgeindex)$
is guaranteed, and
the non-linear term 
$\min\{\greedycounter(\edge), \greedycounter(\edge')\}$
in property~(ii)
simply becomes the linear term $\greedycounter(\edge)$.
Consequently, 
both relatively simple analytical analysis and numerically-aided
analysis are possible.
\end{remark}

\begin{remark}
    Recall that 
    the JMMSV algorithm is equivalent to 
    the 2-Chance Greedy Algorithm with an arbitrary 
    discount factor $\discountfactor\in[0,1]$ and an opening cost scalar $\costscalar = 1$
    (\Cref{ob:JMMSV equivalence MFLP})
    in the MFLP.
    Incorporating the monotonicity of $\greedycounter(\edge)$
    discussed in \Cref{remark:misalignment of greedy counter},
    the factor-revealing program~\ref{eq:weakly factor-revealing program}
    with $\discountfactor = 1$ and $\opencost = 1$
    recovers the factor-revealing program from \citet{JMMSV-02}
    and guarantees the $\greedyapproxMFLP$-approximation 
    of the JMMSV algorithm in the MFLP.
\end{remark}

In what follows, we sketch a 
three-step argument for~\Cref{lem:weakly factor-revealing program},
and defer some details to \Cref{apx:missing proofs}.

\paragraph{Step 1- lower bound of the optimal scaled cost via configuration LP.} 
Given a \TWOLFLP\ instance,
the optimal solution can be formulated as an integer program as 
follows.
We define a \emph{service region}
$\starinstance = (i,\EdgeTilde)$
as a tuple that 
consists a facility at location $i$ 
and an edge subset $\EdgeTilde\subseteq\Edge$
that are served by facility $i$.
Let $\starspace \triangleq [n]\times 2^{\Edge}$
be the set of all possible service regions.
With a slight abuse of notation, for every service region
$\starinstance = (i, \EdgeTilde)\in\starspace$,
we define its cost as
$\cost(\starinstance)
\triangleq {\opencost_i}
    +
    {
    \sum_{\edge\in\EdgeTilde}
    \flowe\cdot \distance(\edge, i)}$
    where the first term is the facility opening cost in location $i$,
    and the second term is the total 
    connection cost between the edge $\edge$ and facility $i$
    over all edge~$\edge\in\EdgeTilde$.
To capture the optimal solution, consider an integer program
where each service region $\starinstance=(i,\EdgeTilde)$ 
is associated with 
a binary variable $\allocstar$
indicating whether the optimal solution 
opens facility $i$ and serves every edge $\edge\in\EdgeTilde$
through facility $i$.
Below we present its linear program relaxation~\ref{eq:lp relaxtion}
as well as its dual program.\footnote{The optimal solution 
can also be formulated as an alternative integer program with polynomial number of
variables: 
each location $i$ is associated with a binary variable indicating 
whether the optimal solution opens facility $i$,
and each pair of edge $\edge$ and location $i$ is
associated with a binary variable indicating 
whether the optimal solution serves this edge $\edge$
through facility $i$.
Compared with its LP relaxation, the program~\ref{eq:lp relaxtion}
with service regions enables a relatively better dual assignment construction
for the approximation analysis.
}
\begin{align}
\tag{$\mathcal{P}_{\texttt{OPT}}$}
\label{eq:lp relaxtion}
    \arraycolsep=5.4pt\def\arraystretch{1}
    &\begin{array}{llllll}
    \min\limits_{\mathbf\alloc\geq \zerobf} &
    \displaystyle\sum\nolimits_{\starinstance\in\starspace} 
    \coststar \cdot \allocstar 
    & \text{s.t.} &
    \quad
    \max\limits_{\boldsymbol{\dual}\geq \zerobf} &
    \displaystyle\sum\nolimits_{\edge\in\Edge} 
    \duale
    & \text{s.t.} \\
    & \displaystyle\sum\nolimits_{\starinstance=(i, \EdgeTilde):
     \edge \in \EdgeTilde} 
     \allocstar \geq 1
     \quad
     & 
     \edge \in \Edge
     &
    &\displaystyle\sum\nolimits_{
     \edge \in \EdgeTilde} 
     \duale \leq \coststar
     \quad
     & 
     \starinstance = (i, \EdgeTilde)\in\starspace
    \end{array}
\end{align}
We formalize the connection between the optimal solution 
and program~\ref{eq:lp relaxtion}
in \Cref{lem:lp relaxation} 
and defer its formal proof into \Cref{apx:lp relaxation}.
\begin{restatable}{lemma}{lprelaxation}
\label{lem:lp relaxation}
Given any \TWOLFLP\ instance,
the optimal cost $\Cost{\OPT}$ is at least $\Obj{\text{\ref{eq:lp relaxtion}}}$.
\end{restatable}

\paragraph{Step 2-
dual assignment construction.}
In this step, we construct the dual assignment~\eqref{eq:dual assignment}
of program~\ref{eq:lp relaxtion}
based on the execution path of the 2-Chance Greedy Algorithm.
Then we argue that the total cost of the algorithm
is at most the objective value of the constructed dual assignment (\Cref{lem:dual assignment objective value}).
Combining with an argument that the constructed dual assignment is also 
approximately feasible presented in the step~3 (\Cref{lem:dual assignment approx feasible}), 
the weak duality of the linear program
completes the proof of \Cref{lem:weakly factor-revealing program}.

Let $\stores$ be the solution computed by the 2-Chance Greedy Algorithm.
Let $\{\greedycounter(\edge)\}$
and $\{\connectmapping(\edge,\edgeindex)\}_{
\edge\in\Edge,\edgeindex\in\hwset}$
be the values of these variables at the termination of the algorithm.
We partition all edges $\edge\in\Edge$ into two disjoint subsets  
$\singlecountedgeset$
and $\doublecountedgeset$ 
 as follows: 
\begin{align*}
    \doublecountedgeset &\triangleq
    \{\edge\in\Edge: \connectmapping(\edge,\Home) \not= \nullsymbol
    \land 
    \connectmapping(\edge,\Work)  \not= \nullsymbol
    \land 
    \connectmapping(\edge,\Home) \not = \connectmapping(\edge,\Work)\}
    \qquad
    \singlecountedgeset \triangleq
    \Edge\backslash\doublecountedgeset
\end{align*}
Namely, subset $\singlecountedgeset$
contains every edge $\edge\in\Edge$ 
that is either partially connected or fully connected to 
a single facility;
while subset $\doublecountedgeset$
contains every edge $\edge\in\Edge$ 
that is fully connected to 
two facilities.

To simplify the presentation, 
we assume $\connectmapping(\edge,\Home)\in\stores$ for every edge $\edge \in \singlecountedgeset$,
and $\distance(\edgehome, \connectmapping(\edge,\Home)) \leq 
 \distance(\edgework, \connectmapping(\edge,\Work))$
 for edge $\edge \in \doublecountedgeset$.
 This is without loss of generality, since the role of home location and work location are ex ante symmetric in our model.
Now, consider the dual assignment (which is not feasible in the dual problem in general)
constructed
as follows,
\begin{align}
\label{eq:dual assignment}
\begin{split}
    \edge\in\singlecountedgeset:&\qquad
    \duale \gets 
    \flowe\cdot 
    \left(\frac{1+\discountfactor}{\costscalar}
    \cdot \greedycounter(\edge)
    -
    \left(
    \frac{1 + \discountfactor}{\costscalar} - 1
    \right)
    \cdot 
    \distance(\edgehome, \connectmapping(\edge,\Home))
    \right)
    \\
    \edge\in\doublecountedgeset:&\qquad
    \duale \gets \flowe\cdot 
    \left(
    \frac{1 + \discountfactor}{\costscalar} \cdot \greedycounter(\edge) 
    - 
    \frac{1}{\costscalar}
    \left(
    \distance(\edgehome, \connectmapping(\edge,\Home))
    +
    \distance(\edgework, \connectmapping(\edge,\Work))
    \right)
    +
    \distance(\edgehome, \connectmapping(\edge,\Home))
    \right)
\end{split}
\end{align}
The construction of \Cref{alg:greedy modified}
ensures that the total cost of solution $\stores$
is upperbounded by the objective value of the constructed dual assignment.
We formalize this in \Cref{lem:dual assignment objective value}
and defer its formal proof to \Cref{apx:dual assignment objective value}.
\begin{restatable}{lemma}{dualassignmentobjectivevalue}
\label{lem:dual assignment objective value}
Given any \TWOLFLP\ instance,
the total cost $\Cost{\stores}$
of solution $\stores$ computed by the 2-Chance Greedy Algorithm
with discount factor $\discountfactor\in[0, 1]$ and opening cost scalar $\costscalar\in\reals_+$
is at most the objective value of program~\ref{eq:lp relaxtion}
with dual assignment \eqref{eq:dual assignment},
i.e., $\Cost{\stores} \leq \sum_{\edge\in\Edge}\duale$.
\end{restatable}

\paragraph{Step 3- approximate feasibility of dual assignment.}
It is straightforward to verify that the constructed dual assignment is non-negative, i.e., 
$\duale \geq 0$
for every edge $\edge\in\Edge$.
It remains to show that for each service region $\starinstance
=(i, \EdgeTilde)\in\starspace$,
the dual constraint associated with primal variable 
$\allocstar$
is approximately satisfied (with an approximation factor of $\Obj{\text{\ref{eq:weakly factor-revealing program}}}$
for some parameters $\mHat,\permu$ depending on $\starinstance$.
We formalize this
 in \Cref{lem:dual assignment approx feasible}.

\begin{restatable}{lemma}{dualassignmentapproxfeasible}
\label{lem:dual assignment approx feasible}
Given any \TWOLFLP\ instance, any $\discountfactor \in[0, 1]$ and $\costscalar\in[1, 1 + \discountfactor]$,
for each service region $\starinstance
=(i, \EdgeTilde)\in\starspace$,
the dual assignment \eqref{eq:dual assignment}
in program~\ref{eq:lp relaxtion}
is approximately feasible up to 
multiplicative factor $\WFRapproxratio$ where
$$\WFRapproxratio = \sup\nolimits_{\mHat\in\naturals,
\permu:[\mHat]\rightarrow\positivereals}
\Obj{\text{\ref{eq:weakly factor-revealing program}}}$$
In particular,  
fix an arbitrary
service region $\starinstance
=(i, \EdgeTilde)\in\starspace$,
let
$\edgetoellmapping:\EdgeTilde\rightarrow[\mHat]$ 
be an arbitrary bijection 
from $\EdgeTilde$ to $[\mHat]$.
Consider $\mHat = |\EdgeTilde|$,
and $\permu(\edgetoellmapping(\edge)) = 
\permuTGD(\edge,\config_i(\edge))$
for every edge $\edge\in\EdgeTilde$.\footnote{Recall
that $\permuTGD$, $\config_i$
are defined in \Cref{lem:structural lemma}.
For ease of presentation, in step 3 we make the assumption
that $\flowe \equiv 1$ for every edge $\edge\in\EdgeTilde$. 
This assumption is without loss of generality, since our argument
can be directly extended by resetting $\mHat$
equal to the total populations over all edges in $\EdgeTilde$,
i.e., $\mHat= \sum_{\edge\in\EdgeTilde}\flowe$, and treat each
of those individuals separately.}
\begin{align*}
    \sum\nolimits_{\edge\in\EdgeTilde}\duale \leq \Obj{\text{\ref{eq:weakly factor-revealing program}}}
    \cdot \coststar
\end{align*}
\end{restatable}
The formal proof of \Cref{lem:dual assignment approx feasible}
is deferred to \Cref{apx:dual assignment approx feasible}.
At a high level, 
for every \TWOLFLP\ instance,
and every service region $\starinstance
=(i,\EdgeTilde)$,
we construct a solution of program~\ref{eq:weakly factor-revealing program}
based on the \TWOLFLP\ instance
as well as the execution path of the 
2-Chance Greedy Algorithm as follows,
\begin{align*}
    &\qquad\costHat\primed \gets \normalizefactor \cdot {\opencost_i},
    \\
    \edge\in\EdgeTilde:&
    \qquad
    \muHat\primed(\edgetoellmapping(\edge)) \gets 
    \normalizefactor \cdot {\greedycounter(\edge)},
    \quad 
    \distanceHat\primed(\edgetoellmapping(\edge)) \gets 
    \normalizefactor \cdot 
    {\distance(\edge, i)},
    \\
    \edge\in\EdgeTilde\cap \singlecountedgeset:
    &\qquad
    \distanceHatStar\primed(\edgetoellmapping(\edge)) \gets 
    \normalizefactor\cdot \distanceHat(\edgehome,\connectmapping(\edge,\Home)),
    \\
    \edge\in\EdgeTilde\cap \doublecountedgeset:
    &\qquad
    \distanceHatStar\primed(\edgetoellmapping(\edge)) \gets 
    \normalizefactor\cdot \distance(
    \edge_{\config_i(\edge)}, \connectmapping(\edge, \config_i(\edge)))
\end{align*}
where $\normalizefactor = \frac{1}{\opencost_i+\sum_{\edge\in\EdgeTilde}
\distance(\edge, i)}$
is the normalization factor\footnote{If $\opencost_i+\sum_{\edge\in\EdgeTilde}\distance(\edge, i) = 0$,
the algorithm ensures that $\greedycounter(\edge) = 0$ and thus the dual constraint is satisfied with equality trivially.}
This constructed solution 
of program~\ref{eq:weakly factor-revealing program}
is feasible:
constraints~\WFRCtriangleineq--\WFRCconnectcost\ is satisfied due to \Cref{lem:structural lemma},
and constraint~\WFRCtotalcost\ 
is satisfied due to the normalization factor $\normalizefactor$.
Furthermore, its objective value equals to 
the ratio between $\sum_i\duale$ over
$\coststar$.

\begin{proof}[Proof of \Cref{lem:weakly factor-revealing program}]
Invoking 
\Cref{lem:lp relaxation,lem:dual assignment objective value,lem:dual assignment approx feasible}
and weak duality in linear programming finishes the proof.
\end{proof}

\subsection{Construction of strongly factor-revealing quadratic program~\ref{eq:strongly factor-revealing quadratic program}}

\label{sec:strongly factor-revealing program}

In \Cref{lem:weakly factor-revealing program}, we 
establish an
upper bound on the approximation ratio of 
the 2-Chance Greedy Algorithm with the supremum of optimal objectives of
factor revealing programs~\{\ref{eq:weakly factor-revealing program}\}
over all its possible parameters
$\mHat\in\naturals,\permu\in\positivereals^{\mHat}$.
However, program~\ref{eq:weakly factor-revealing program}
is nontrivial and
does not lend itself to a straightforward characterization of  this supremum. 
To overcome this obstacle, 
we upperbound the original factor-revealing 
program~\ref{eq:weakly factor-revealing program}
through 
a new strongly factor-revealing quadratic program~\ref{eq:strongly factor-revealing quadratic program},
and prove that the latter upperbounds 
the supremum of
the former
and, in turn, provides an approximation ratio for the algorithm. 
We formalize this in \Cref{lem:from weakly to strongly factor revealing program}. 

\begin{restatable}{lemma}{weaklytostronglyFRP}
\label{lem:from weakly to strongly factor revealing program}
For any $\discountfactor\in[0, 1]$,
$\nHat\in\naturals$,
$\mHat\in\naturals$, and $\permu:[\mHat] \rightarrow \positivereals$,
the optimal objective value of program~\ref{eq:weakly factor-revealing program}
is at most equal to the optimal objective value of program~\ref{eq:strongly factor-revealing quadratic program},
i.e.,
\begin{align*}
\Obj{\text{\ref{eq:weakly factor-revealing program}}}
\leq \Obj{\text{\ref{eq:strongly factor-revealing quadratic program}}}
\end{align*}
\end{restatable}

The formal proof of \Cref{lem:from weakly to strongly factor revealing program} is deferred to \Cref{apx:from weakly to strongly factor revealing program}.
In the remainder of this subsection,
we highlight the key steps of upperbounding 
factor-revealing program~\ref{eq:weakly factor-revealing program}
through the strongly factor-revealing quadratic program~\ref{eq:strongly factor-revealing quadratic program}.
Specifically, in \Cref{sec:strongly factor-revealing program MFLP},
as a warm-up exercise, 
we present a simple batching argument 
which upperbounds the factor-revealing program~\ref{eq:weakly factor-revealing program MFLP}
designed in \citet{JMMSV-02} for the MFLP
through the strongly factor-revealing \emph{linear} program~\ref{eq:strongly factor-revealing program MFLP}.
Consequently, we reprove the \greedyapproxMFLP-approximation 
in the MFLP.
Next, we discuss the main technical challenge 
for the \TWOLFLP, and our additional treatment 
which
upperbounds program~\ref{eq:weakly factor-revealing program}
through a strongly factor-revealing \emph{quadratic} program~\ref{eq:strongly factor-revealing quadratic program}
in \Cref{sec:strongly factor-revealing program 2-LFLP}.

\subsubsection{Warm-up: construction in the MFLP}
\label{sec:strongly factor-revealing program MFLP}

In the MFLP, the 2-Chance Greedy Algorithm with 
an arbitrary discount factor $\discountfactor \in[0, 1]$ and 
opening cost scalar $\costscalar = 1$
is equivalent to the JMMSV algorithm (\Cref{ob:JMMSV equivalence MFLP}),
whose approximation ratio is captured by the following 
factor-revealing program.

\begin{theorem}[adopted from \citealp{JMMSV-02}]
In the MFLP, the approximation ratio of the 2-Chance Greedy Algorithm
with discount factor $\discountfactor\in[0, 1]$ and opening cost scalar $\costscalar = 1$ is equal to $\WFRMFLPapproxratio$ where
\begin{align*}
    \WFRMFLPapproxratio=
    \sup\limits_{\mHat\in\naturals}
\Obj{\text{\ref{eq:weakly factor-revealing program MFLP}}}
\end{align*}
Here \ref{eq:weakly factor-revealing program MFLP}
is the maximization program parameterized by 
$\mHat\in\naturals$ 
defined as follows:
\begin{align}
\tag{$\WFRPMFLP{\mHat}$}
\label{eq:weakly factor-revealing program MFLP}
&\arraycolsep=1.4pt\def\arraystretch{2.2}
    \begin{array}{llll}
    \max\limits_{
    \substack{\costHat\geq 0,
    \muHatBf,
    \distanceHatBf \geq 
    \zerobf
    }} &
     \sum\nolimits_{\ell\in[\mHat]}
    \muHat(\ell)
    & \text{s.t.} & \\
    &
    \muHat(\ell) \leq \muHat(\ell')
    &
    \ell,\ell'\in[\mHat]
    \\
     &
    \muHat(\ell') \leq 
    \muHat(\ell) + \distanceHat(\ell)
    +
    \distanceHat(\ell')
    &
    \ell,\ell'\in[\mHat]
    &
    \\
    &
    \sum\nolimits_{\ell'\in[\ell:\mHat]}
    \plus{\muHat(\ell) - \distanceHat(\ell')}
    \leq \costHat\qquad
    &
    \ell \in[\mHat]
    \\
      &
     \costHat 
     +
     \sum\nolimits_{\ell\in[\mHat]}
     \distanceHat(\ell)
     \leq 1
     &
    \end{array}
\end{align}
\end{theorem}
Compared to program~\ref{eq:weakly factor-revealing program}
in the \TWOLFLP,
program~\ref{eq:weakly factor-revealing program MFLP}
has an additional constraint of 
the monotonicity of $\muHat(\ell)$,
and thus term $\plus{\min\{\muHat(\ell),\muHat(\ell')\} - \distanceHat(\ell')}$
in
constraint~\WFRCcontribution \ 
of program~\ref{eq:weakly factor-revealing program}
becomes $\plus{\muHat(\ell) - \distanceHat(\ell')}$,
which  is relatively easier to handle.
As a warm-up exercise, here we convert program~\ref{eq:weakly factor-revealing program MFLP}
into a strongly factor-revealing program~\ref{eq:strongly factor-revealing program MFLP}.
\begin{lemma}
\label{lem:from weakly to strongly factor revealing program MFLP}
For any 
$\nHat\in\naturals$,
$\mHat\in\naturals$,
\begin{align*}
\Obj{\text{\ref{eq:weakly factor-revealing program MFLP}}}
\leq \Obj{\text{\ref{eq:strongly factor-revealing program MFLP}}}
\end{align*}
where program~\ref{eq:strongly factor-revealing program MFLP} is defined as follows:
\begin{align}
\tag{$\SFRPMFLP{\nHat}$}
\label{eq:strongly factor-revealing program MFLP}
&\arraycolsep=1.4pt\def\arraystretch{2.2}
    \begin{array}{llll}
    \max\limits_{
    \substack{\costHat\geq 0,
    \muHatBf,
    \distanceHatBf \geq 
    \zerobf
    }} &
     \sum\nolimits_{\ell\in[\nHat]}
    \muHat(\ell)
    & \text{s.t.} & \\
    &
    \muHat(\ell) \leq \muHat(\ell')
    &
    \ell,\ell'\in[\nHat]
    \\
     &
    \muHat(\ell') \leq 
    \muHat(\ell) + \distanceHat(\ell)
    +
    \distanceHat(\ell')
    &
    \ell,\ell'\in[2:\nHat]
    &
    \\
    &
    \sum\nolimits_{\ell'\in[\ell + 1:\nHat]}
    \plus{\muHat(\ell) - \distanceHat(\ell')}
    \leq \costHat\qquad
    &
    \ell \in[\nHat]
    \\
      &
     \costHat 
     +
     \sum\nolimits_{\ell\in[\nHat]}
     \distanceHat(\ell)
     \leq 1
     &
    \end{array}
\end{align}
\end{lemma}
Both program~\ref{eq:weakly factor-revealing program MFLP} and 
program~\ref{eq:strongly factor-revealing program MFLP} admit
similar structures. The only difference is the index range 
in the second and third constraints. 
In fact, for any $\mHat\in\naturals$,
program~\SFRPMFLP{\mHat} 
by construction 
is 
a relaxation of program~\ref{eq:weakly factor-revealing program MFLP},
and thus $\Obj{\WFRPMFLP{\mHat}} \leq \Obj{\SFRPMFLP{\mHat}}$.
\Cref{lem:from weakly to strongly factor revealing program MFLP}
is a stronger statement which establishes that $\Obj{\WFRPMFLP{\mHat}} \leq \Obj{\SFRPMFLP{\nHat}}$ \emph{for every $\mHat,\nHat\in\naturals$}.
Its proof is based on a simple batching argument.
\begin{proof}[Proof of \Cref{lem:from weakly to strongly factor revealing program MFLP}]
Fix arbitrary $\nHat,\mHat\in \naturals$. 
Without loss of generality,
we assume that $\mHat$ is sufficiently large\footnote{Given any feasible solution $(\costHat, \{\muHat(\ell), \distanceHat(\ell)\}_{\ell\in[\mHat]})$
in program~\ref{eq:weakly factor-revealing program MFLP},
consider a new solution 
$(\costHat\primed, \{\muHat\primed(\ell), \distanceHat\primed(\ell)\}_{\ell\in[2\mHat]})$
in $\WFRPMFLP{2\mHat}$
where $\costHat\primed = \costHat$,
$\muHat\primed(\ell) = \muHat(\lceil\sfrac{\ell}{2}\rceil)$,
and $\distanceHat\primed(\ell) = \distanceHat(\lceil\sfrac{\ell}{2}\rceil)$.
The latter solution is also feasible and has 
the same objective value as the former solution.
Hence,
$\Obj{\WFRPMFLP{\mHat}}
\leq \Obj{\WFRPMFLP{2\mHat}}$ for 
every $\mHat\in\naturals$.
Therefore, it is without loss of generality to consider sufficiently large $\mHat\in\naturals$.}
so that
$\lceil\frac{\mHat}{\nHat}\rceil\cdot (\nHat - 1) \leq \mHat$.
Let $k = \lceil\frac{\mHat}{\nHat}\rceil$.
Define sequence $L = \{\ell_1, \ell_2, \dots, \ell_{\nHat}, \ell_{\nHat + 1}\}$
where 
\begin{align*}
\ell_a = \left\{
\begin{array}{ll}
 1    &  \text{$\;~a = 1$}\\
 1 + \mHat - k \cdot (\nHat + 1 - a)    & 
 \text{$\forall a \in [2:\nHat + 1]$}
\end{array}
\right.
\end{align*}
By definition, $\ell_2 - \ell_ 1 = \mHat - k\cdot (\nHat - 1) \leq k$,
and
$\ell_{a + 1} - \ell_a  = k$ for each $a \in[2:\nHat]$.
Given an arbitrary feasible solution $(\costHat,\{\muHat(\ell),\distanceHat(\ell)\})$ of 
program~\ref{eq:weakly factor-revealing program MFLP},
we construct a feasible solution $(\costHat\primed,\{\muHat\primed(\ell),\distanceHat\primed(\ell)\})$
through a batching procedure as follows:\footnote{Here 
we use superscript $\dagger$ to denote the solution in
program~\ref{eq:strongly factor-revealing program MFLP}.}
\begin{align*}
    &\qquad\costHat\primed \gets \costHat \\
    a\in[\nHat]:&\qquad
    \muHat\primed(a) \gets 
    \displaystyle\sum\nolimits_{\ell \in [\ell_a:\ell_{a + 1} - 1]}
    \muHat(\ell),
    \quad
    \distanceHat\primed(a) \gets 
    \displaystyle\sum\nolimits_{\ell \in [\ell_a:\ell_{a + 1} - 1]}
    \distanceHat(\ell)
\end{align*}
See a graphical illustration of batching in \Cref{fig:batching MFLP}.

It is straightforward to verify that the objective value remains unchanged
after the batching procedure, and all constraints are satisfied
in program~\ref{eq:strongly factor-revealing program MFLP}.
In particular, the third constraint in program~\ref{eq:strongly factor-revealing program MFLP} for each $a\in[\nHat]$ 
is 
implied by the third constraint in program~\ref{eq:weakly factor-revealing program MFLP} 
at $\ell = \ell_{a + 1} - 1$
and the convexity of $\plus{\cdot}$.
\end{proof}

\begin{figure}
    \centering
    \begin{tikzpicture}[scale=0.5, transform shape]
\begin{axis}[
axis line style=gray,
axis lines=middle,
        ytick = \empty,
        yticklabels = \empty,
        xtick = {0.7, 3.7, 7.7},
        xticklabels = {$\ell_1$, $\ell_2$, $\ell_3$},
x label style={at={(axis description cs:1.0,-0.01)},anchor=north},
y label style={at={(axis description cs:0.05,0.95)},anchor=south},
xlabel = {$\ell$},
ylabel = {$\muHat(\ell)$},
label style={font=\LARGE},
xmin=0,xmax=11.5,ymin=0,ymax=6.5,
width=0.9\textwidth,
height=0.5\textwidth,
samples=50]

\addplot[fill=white!90!black, postaction={
        pattern=crosshatch,
        pattern color=white!100!black
    }] coordinates {
(0.4, 0.)(0.4, 1.0)(1.0, 1.0)(1.0, 0.0)(0.4, 0.)
};

\addplot[fill=white!90!black, postaction={
        pattern=crosshatch,
        pattern color=white!100!black
    }] coordinates {
(1.4, 0.)(1.4, 2.0)(2.0, 2.0)(2.0, 0.0)(1.4, 0.)
};

\addplot[fill=white!90!black, postaction={
        pattern=crosshatch,
        pattern color=white!100!black
    }] coordinates {
(2.4, 0.)(2.4, 2.5)(3.0, 2.5)(3.0, 0.0)(2.4, 0.)
};

\addplot[dashed, line width=0.8mm] coordinates {
(0.3, 0)(0.3, 2.6)(3.1, 2.6) (3.1, 0)(0.3, 0)
};

\addplot[fill=white!50!black, postaction={
        pattern=north east lines,
        pattern color=white!100!black,
    }]coordinates {
(3.4, 0.)(3.4, 3)(4.0, 3)(4.0, 0.0)(3.4, 0.)
};

\addplot[fill=white!50!black, postaction={
        pattern=north east lines,
        pattern color=white!100!black,
    }] coordinates {
(4.4, 0.)(4.4, 4.0)(5.0, 4.0)(5.0, 0.0)(4.4, 0.)
};

\addplot[fill=white!50!black, postaction={
        pattern=north east lines,
        pattern color=white!100!black,
    }] coordinates {
(5.4, 0.)(5.4, 4.5)(6.0, 4.5)(6.0, 0.0)(5.4, 0.)
};

\addplot[fill=white!50!black, postaction={
        pattern=north east lines,
        pattern color=white!100!black,
    }] coordinates {
(6.4, 0.)(6.4, 4.5)(7.0, 4.5)(7.0, 0.0)(6.4, 0.)
};

\addplot[dashed, line width=0.8mm] coordinates {
(3.3, 0)(3.3, 4.6)(7.1, 4.6) (7.1, 0)(3.3, 0)
};

\addplot[fill=white!0!black] coordinates {
(7.4, 0.)(7.4, 5.0)(8.0, 5.0)(8.0, 0.0)(7.4, 0.)
};

\addplot[fill=white!0!black] coordinates {
(8.4, 0.)(8.4, 5.0)(9.0, 5.0)(9.0, 0.0)(8.4, 0.)
};

\addplot[fill=white!0!black] coordinates {
(9.4, 0.)(9.4, 6.0)(10.0, 6.0)(10.0, 0.0)(9.4, 0.)
};

\addplot[fill=white!0!black] coordinates {
(10.4, 0.)(10.4, 6.25)(11.0, 6.25)(11.0, 0.0)(10.4, 0.)
};

\addplot[dashed, line width=0.8mm] coordinates {
(7.3, 0)(7.3, 6.35)(11.1, 6.35) (11.1, 0)(7.3, 0)
};


\end{axis}

\end{tikzpicture}
    \caption{Graphical illustration of
    batching in \Cref{lem:from weakly to strongly factor revealing program MFLP}:
    converting a feasible solution in program~\WFRPMFLP{11}
    into a feasible solution in program~\SFRPMFLP{3}.
    }
    \label{fig:batching MFLP}
\end{figure}

Note that the strongly factor-revealing program~\ref{eq:strongly factor-revealing program MFLP}
can be converted into a linear program by introducing additional auxiliary variables and inequalities. 
By numerically computing it using Gurobi, 
we obtain the following approximation
ratio for the 2-Chance Greedy Algorithm in the MFLP.
\begin{proposition}
\label{prop:approx ratio general gamma MFLP}
In the MFLP, the approximation ratio of 
the 2-Chance Greedy Algorithm 
with discount factor $\discountfactor\in[0, 1]$ and opening cost scalar $\costscalar = 1$
is at most equal to $\WFRMFLPapproxratio \leq 
\Obj{\SFRPMFLP{500}} \leq \greedyapproxMFLP$.
\end{proposition}

\subsubsection{Construction in the \TWOLFLP}
\label{sec:strongly factor-revealing program 2-LFLP}

In this part, we discuss the main technical ingredients for the proof of \Cref{lem:from weakly to strongly factor revealing program}.
We start by highlighting the additional difficulty 
in the analysis  of \TWOLFLP,
and explaining the failure of the naive batching argument used in
the  proof of \Cref{lem:from weakly to strongly factor revealing program MFLP}
in the MFLP.
Then we  focus on obtaining 
the strongly factor-revealing quadratic program~\ref{eq:strongly factor-revealing quadratic program}, and provide one of our main technical contributions.

\paragraph{Failure of the naive batching argument.} 
In the proof of \Cref{lem:from weakly to strongly factor revealing program MFLP},
we use a naive batching argument 
that groups an arbitrary feasible solution of 
the factor-revealing program~\ref{eq:weakly factor-revealing program MFLP}
into a feasible solution in 
strongly factor-revealing program~\ref{eq:strongly factor-revealing program MFLP}
with $\nHat \leq \mHat$.
In particular, it divides the index set $[\mHat]$ 
into $\nHat$ \emph{consecutive} subsets with almost \emph{uniform} size, and
groups (i.e., sums up) solutions in each consecutive subsets separately.
The objective value remains unchanged,
and constraints remain satisfied due to the linearity or convexity 
(for the third constraint in program~\ref{eq:strongly factor-revealing program MFLP}).

As we highlighted at the end of \Cref{sec:primal dual framework},
in the \TWOLFLP,
the \candidatecosttext\ $\greedycounter(\edge)$
of
the 2-Chance Greedy Algorithm 
is not monotone in the time stamp when each edge $\edge$
is connected, since edge $\edge$ can be connected twice, and its 
\candidatecosttext\ $\greedycounter(\edge)$ stops increasing after its 
first connection.
Therefore, 
the factor-revealing program
program~\ref{eq:weakly factor-revealing program}
does not admit the monotonicity
property of $\muHat(\ell)$,
and constraint~\WFRCcontribution\ 
features the non-convex term $\min\{\muHat(\ell), \muHat(\ell')\}$.
Consequently, the naive batching argument 
fails.\footnote{To the best of our knowledge,
this naive batching argument 
is used in 
all previous works with strongly factor-revealing program
\citep[e.g.,][]{MY-11,GT-12}.}
Here we provide an example
to illustrate the failure of 
the naive batching argument for program~\ref{eq:weakly factor-revealing program}.

\begin{example}
\label{example: failure of naive batching}
Let $\permu(\ell) = \ell$ for each $\ell$.
Consider a feasible solution $(\costHat, \{\muHat(\ell), \distanceHat(\ell), \distanceHatStar(\ell)\}_{\ell\in[4]})$
of program~$\WFRP{4, \permu, 1, 1}$ as follows:
$\costHat = 3\normalizefactor$;
$\muHat(\ell) = 9\normalizefactor,
9\normalizefactor, 4\normalizefactor, 14\normalizefactor$;
$\distanceHat(\ell) = 9\normalizefactor, 9\normalizefactor, 4\normalizefactor, 6\normalizefactor$;
and 
$\distanceHatStar(\ell) = 9\normalizefactor, 9\normalizefactor, 4\normalizefactor, 14\normalizefactor$
for $\ell= 1, 2, 3, 4$, respectively.
Here $\normalizefactor$ is the normalization factor such that constraint~\WFRCtotalcost\ is satisfied with equality.
It is straightforward to verify that all other constraints are also satisfied. 
In particular, 
constraint~\WFRCcontribution\ at $\ell = 2$ is 
\begin{align*}
    \displaystyle\sum\nolimits_{\ell'\in[2:4]}
    \plus{ 
    \min\left\{\muHat(2), \muHat(\ell')\right\} - \distanceHat(\ell')}
    =
    \left(9\normalizefactor - 9\normalizefactor\right) + 
    \left(4\normalizefactor - 4\normalizefactor\right) + 
    \left(9\normalizefactor - 6\normalizefactor\right) 
    =
    3\normalizefactor = \costHat
\end{align*}
Now, suppose we use the naive batching modification (defined in the proof of \Cref{lem:from weakly to strongly factor revealing program MFLP})
to $\nHat\primed = 2$, the solution $(\costHat\primed, \{\muHat\primed(\ell), \distanceHat\primed(\ell), \distanceHatStar\primed(\ell)\}_{\ell\in[2]})$
after the batching modification
is $\costHat\primed \gets \costHat$;
$\muHat\primed(1) \gets \muHat(1) + \muHat(2) = 18\normalizefactor$;
$\muHat\primed(2) \gets \muHat(3) + \muHat(4) = 18\normalizefactor$;
$\distanceHat\primed(1) \gets \distanceHat(1) + \distanceHat(2) = 18\normalizefactor$;
$\distanceHat\primed(1) \gets \distanceHat(1) + \distanceHat(2) = 18\normalizefactor$;
$\distanceHatStar\primed(1) \gets \distanceHatStar(1) + \distanceHatStar(2) = 18\normalizefactor$;
and 
$\distanceHatStar\primed(2) \gets \distanceHatStar(3) + \distanceHatStar(4) = 18\normalizefactor$.
Notably, constraint~\WFRCcontribution\ at $\ell = 1$ is violated (even if 
we consider the similar relaxed on indexes as the one in program~\ref{eq:strongly factor-revealing program MFLP}),
\begin{align*}
    \displaystyle\sum\nolimits_{\ell'\in[2:2]}
    \plus{ 
    \min\left\{\muHat\primed(1), \muHat\primed(\ell')\right\} - \distanceHat\primed(\ell')}
    =
    18\normalizefactor - 10\normalizefactor
    =
    8\normalizefactor > \costHat\primed
\end{align*}
\end{example}

\paragraph{A solution-dependent batching argument.}

To prove \Cref{lem:from weakly to strongly factor revealing program},
we introduce a new batching argument that enables deriving
the strongly factor-revealing program.
Similar to the naive batching argument,  
given a feasible solution in 
the original factor-revealing program~\ref{eq:weakly factor-revealing program},
the batching procedure 
partitions the index set
and groups the variables of the original feasible solution
to construct a feasible solution in the strongly factor-revealing program~\ref{eq:strongly factor-revealing quadratic program}.
Then we show that its objective value weakly increases and all constraints 
are satisfied, which ensures that the objective value of~\ref{eq:strongly factor-revealing quadratic program} upper bounds that of ~\ref{eq:weakly factor-revealing program}.
Our new batching argument is \emph{solution-dependent}, i.e.,
it partitions the index set 
into \emph{non-consecutive} index subsets with \emph{non-uniform} size
\emph{based on 
the original feasible solution}.
Below we sketch the four major steps in our batching argument
and provide intuition on our approach.
All missing details and the formal proof of \Cref{lem:from weakly to strongly factor revealing program}
is in \Cref{apx:from weakly to strongly factor revealing program}.

\begin{figure}
    \centering
    \subfloat[
    ]{\begin{tikzpicture}[scale=0.35, transform shape]
\begin{axis}[
axis line style=gray,
axis lines=middle,
        ytick = \empty,
        yticklabels = \empty,
        xtick = {0.7, 4.7, 9.7},
        xticklabels = {$\ell_1$, $\ell_2$, $\ell_3$},
x label style={at={(axis description cs:1.0,-0.01)},anchor=north},
y label style={at={(axis description cs:0.05,0.95)},anchor=south},
xlabel = {$\ell$},
ylabel = {$\muHat(\ell)$},
label style={font=\LARGE},
xmin=0,xmax=12.5,ymin=0,ymax=6.5,
width=0.9\textwidth,
height=0.5\textwidth,
samples=50]

\addplot[fill=white!90!black, postaction={
        pattern=crosshatch,
        pattern color=white!100!black
    }] coordinates {
(0.4, 0.)(0.4, 3.0)(1.0, 3.0)(1.0, 0.0)(0.4, 0.)
};

\addplot[fill=white!90!black, postaction={
        pattern=crosshatch,
        pattern color=white!100!black
    }] coordinates {
(1.4, 0.)(1.4, 2.0)(2.0, 2.0)(2.0, 0.0)(1.4, 0.)
};

\addplot[fill=white!90!black, postaction={
        pattern=crosshatch,
        pattern color=white!100!black
    }] coordinates {
(2.4, 0.)(2.4, 1.0)(3.0, 1.0)(3.0, 0.0)(2.4, 0.)
};

\addplot[fill=white!90!black, postaction={
        pattern=crosshatch,
        pattern color=white!100!black
    }]coordinates {
(3.4, 0.)(3.4, 2.5)(4.0, 2.5)(4.0, 0.0)(3.4, 0.)
};

\addplot[fill=white!50!black, postaction={
        pattern=north east lines,
        pattern color=white!100!black,
    }] coordinates {
(4.4, 0.)(4.4, 5.0)(5.0, 5.0)(5.0, 0.0)(4.4, 0.)
};

\addplot[fill=white!50!black, postaction={
        pattern=crosshatch,
        pattern color=white!100!black,
    }] coordinates {
(5.4, 0.)(5.4, 1.5)(6.0, 1.5)(6.0, 0.0)(5.4, 0.)
};

\addplot[fill=white!50!black, postaction={
        pattern=north east lines,
        pattern color=white!100!black,
    }] coordinates {
(6.4, 0.)(6.4, 4.0)(7.0, 4.0)(7.0, 0.0)(6.4, 0.)
};

\addplot[fill=white!50!black, postaction={
        pattern=crosshatch,
        pattern color=white!100!black
    }] coordinates {
(7.4, 0.)(7.4, 3.0)(8.0, 3.0)(8.0, 0.0)(7.4, 0.)
};

\addplot[fill=white!50!black, postaction={
        pattern=north east lines,
        pattern color=white!100!black,
    }] coordinates {
(8.4, 0.)(8.4, 5.0)(9.0, 5.0)(9.0, 0.0)(8.4, 0.)
};

\addplot[fill=white!0!black] coordinates {
(9.4, 0.)(9.4, 6.0)(10.0, 6.0)(10.0, 0.0)(9.4, 0.)
};

\addplot[fill=white!0!black, postaction={
        pattern=crosshatch,
        pattern color=white!100!black
    }] coordinates {
(10.4, 0.)(10.4, 1.25)(11.0, 1.25)(11.0, 0.0)(10.4, 0.)
};

\addplot[fill=white!0!black, postaction={
        pattern=north east lines,
        pattern color=white!100!black,
    }] coordinates {
(11.4, 0.)(11.4, 4.25)(12.0, 4.25)(12.0, 0.0)(11.4, 0.)
};

\addplot[dotted, white!10!black] coordinates {
(0, 3) (12, 3)};
\addplot[dotted, white!10!black] coordinates {
(0, 5) (12, 5)};
\addplot[dotted, white!10!black] coordinates {
(0, 6) (12, 6)};

\end{axis}

\end{tikzpicture}}
    \;\;\;\;
    \subfloat[
    ]{\begin{tikzpicture}[scale=0.35, transform shape]
\begin{axis}[
axis line style=gray,
axis lines=middle,
        ytick = \empty,
        yticklabels = \empty,
        xtick = {0.7, 4.7, 9.7},
        xticklabels = {$\ell_1$, $\ell_2$, $\ell_3$},
x label style={at={(axis description cs:1.0,-0.01)},anchor=north},
y label style={at={(axis description cs:0.05,0.95)},anchor=south},
xlabel = {$\ell$},
ylabel = {$\muHat(\ell)$},
label style={font=\LARGE},
xmin=0,xmax=12.5,ymin=0,ymax=6.5,
width=0.9\textwidth,
height=0.5\textwidth,
samples=50]

\addplot[fill=white!90!black, postaction={
        pattern=crosshatch,
        pattern color=white!100!black
    }] coordinates {
(0.4, 0.)(0.4, 2.125)(1.0, 2.125)(1.0, 0.0)(0.4, 0.)
};

\addplot[fill=white!90!black, postaction={
        pattern=crosshatch,
        pattern color=white!100!black
    }] coordinates {
(1.4, 0.)(1.4, 2.125)(2.0, 2.125)(2.0, 0.0)(1.4, 0.)
};

\addplot[fill=white!90!black, postaction={
        pattern=crosshatch,
        pattern color=white!100!black
    }] coordinates {
(2.4, 0.)(2.4, 2.125)(3.0, 2.125)(3.0, 0.0)(2.4, 0.)
};

\addplot[fill=white!90!black, postaction={
        pattern=crosshatch,
        pattern color=white!100!black
    }]coordinates {
(3.4, 0.)(3.4, 2.125)(4.0, 2.125)(4.0, 0.0)(3.4, 0.)
};

\addplot[fill=white!50!black, postaction={
        pattern=north east lines,
        pattern color=white!100!black,
    }] coordinates {
(4.4, 0.)(4.4, 4.67)(5.0, 4.67)(5.0, 0.0)(4.4, 0.)
};

\addplot[fill=white!50!black, postaction={
        pattern=crosshatch,
        pattern color=white!100!black,
    }] coordinates {
(5.4, 0.)(5.4, 2.25)(6.0, 2.25)(6.0, 0.0)(5.4, 0.)
};

\addplot[fill=white!50!black, postaction={
        pattern=north east lines,
        pattern color=white!100!black,
    }] coordinates {
(6.4, 0.)(6.4, 4.67)(7.0, 4.67)(7.0, 0.0)(6.4, 0.)
};

\addplot[fill=white!50!black, postaction={
        pattern=crosshatch,
        pattern color=white!100!black
    }] coordinates {
(7.4, 0.)(7.4, 2.25)(8.0, 2.25)(8.0, 0.0)(7.4, 0.)
};

\addplot[fill=white!50!black, postaction={
        pattern=north east lines,
        pattern color=white!100!black,
    }] coordinates {
(8.4, 0.)(8.4, 4.67)(9.0, 4.67)(9.0, 0.0)(8.4, 0.)
};

\addplot[fill=white!0!black] coordinates {
(9.4, 0.)(9.4, 6.0)(10.0, 6.0)(10.0, 0.0)(9.4, 0.)
};

\addplot[fill=white!0!black, postaction={
        pattern=crosshatch,
        pattern color=white!100!black
    }] coordinates {
(10.4, 0.)(10.4, 1.25)(11.0, 1.25)(11.0, 0.0)(10.4, 0.)
};

\addplot[fill=white!0!black, postaction={
        pattern=north east lines,
        pattern color=white!100!black,
    }] coordinates {
(11.4, 0.)(11.4, 4.25)(12.0, 4.25)(12.0, 0.0)(11.4, 0.)
};

\addplot[dotted, white!10!black] coordinates {
(0, 2.125) (12, 2.125)};
\addplot[dotted, white!10!black] coordinates {
(0, 4.67) (12, 4.67)};
\addplot[dotted, white!10!black] coordinates {
(0, 6) (12, 6)};

\end{axis}

\end{tikzpicture}}
    \;\;\;\;
    \subfloat[
    ]
    {\begin{tikzpicture}[scale=0.19444444444, transform shape]
\begin{axis}[
axis line style=gray,
axis lines=middle,
        ytick = \empty,
        yticklabels = \empty,
        xtick = \empty,
        xticklabels = \empty,
x label style={at={(axis description cs:0.95,-0.01)},anchor=north},
y label style={at={(axis description cs:0.03,0.95)},anchor=south},
xlabel = {$a$},
ylabel = {$b$},
label style={font=\LARGE},
xmin=0.2,xmax=3.5,ymin=0.2,ymax=3.2,
width=0.9\textwidth,
height=0.9\textwidth,
samples=50]

\addplot[fill=white!90!black, postaction={
        pattern=crosshatch,
        pattern color=white!100!black
    }] coordinates {
(0.4, 0.4)(0.4, 1.2)(1.2, 1.2)(1.2, 0.4)(0.4, 0.4)
};
\addplot[fill=white!50!black, postaction={
        pattern=crosshatch,
        pattern color=white!100!black,
    }] coordinates {
(1.4, 0.4)(1.4, 1.2)(2.2, 1.2)(2.2, 0.4)(1.4, 0.4)
};
\addplot[fill=white!50!black, postaction={
        pattern=north east lines,
        pattern color=white!100!black,
    }] coordinates {
(1.4, 1.4)(1.4, 2.2)(2.2, 2.2)(2.2, 1.4)(1.4, 1.4)
};
\addplot[fill=white!0!black, postaction={
        pattern=crosshatch,
        pattern color=white!100!black,
    }] coordinates {
(2.4, 0.4)(2.4, 1.2)(3.2, 1.2)(3.2, 0.4)(2.4, 0.4)
};
\addplot[fill=white!0!black, postaction={
        pattern=north east lines,
        pattern color=white!100!black,
    }] coordinates {
(2.4, 1.4)(2.4, 2.2)(3.2, 2.2)(3.2, 1.4)(2.4, 1.4)
};
\addplot[fill=white!0!black] coordinates {
(2.4, 2.4)(2.4, 3.2)(3.2, 3.2)(3.2, 2.4)(2.4, 2.4)
};
\end{axis}

\end{tikzpicture} }
    \caption{Graphical example illustration of steps~1, 2 and 3
    in the solution-dependent batching argument for \Cref{lem:from weakly to strongly factor revealing program}.
    Each combination of colors (white, gray, black) and patterns (solid, stripe, crosshatch)
    corresponds to an index subset $L(a, b)$.
    Subplot (a): pivot index subset $L=\{\ell_1,\ell_2,\ell_3\}$ 
    and
    index partition $\{L(a, b)\}_{a\in[3],b\in[a]}$.
    Subplot (b): averaging variables in each index partition.
    Subplot (c): batching variables in each index partition,
    from one-dimensional index set \text{$\{\ell\in[\mHat]\}$
    }
    to two-dimensional index set \text{$\{(a,b):a\in[3],b\in[a]\}$}.}
    \label{fig:pivot ell}
\end{figure}

\paragraph{Step 1- identifying pivotal index subset $L$ with monotone $\muHat(\ell)$.}

Fix an arbitrary $\discountfactor\in[0, 1]$, $\costscalar\in[1, 1 + \discountfactor]$, and
$\mHat\in\naturals$.
For ease of presentation, in this part we only consider $\permu(\cdot)$
such that $\permu(\ell) = \ell$ for each $\ell \in[\mHat]$.
The formal argument which covers general $\permu\in\positivereals^{\mHat}$
is defined in \Cref{apx:from weakly to strongly factor revealing program}.
Suppose $(\costHat, \{\muHat(\ell), \distanceHat(\ell),\distanceHatStar(\ell)\})$
is an arbitrary feasible solution of program~\ref{eq:weakly factor-revealing program}.
Consider a sequence of $k$ indexes 
$1 = \ell_1 < \ell_2 < \dots < \ell_k \leq \mHat$ 
for some $k\in\naturals$
such that 
\begin{align*}
\forall a\in[k-1]:
\qquad 
\muHat(\ell_a) < \muHat(\ell_{a + 1})
        \\
     \forall a \in [k],~
     \forall\ell \in [\ell_a:\ell_{a + 1} - 1]: 
    \qquad \muHat(\ell) \leq 
    \muHat(\ell_a)
\end{align*}
where $\ell_{k + 1} = m + 1$.
We define pivotal index subset $L \triangleq 
\{\ell_a\}_{a\in[k]}$.
See \Cref{fig:pivot ell}
for a graphical illustration.
It is easy to show both 
the existence and uniqueness of 
pivotal index subset~$L$.

\paragraph{Step 2- partitioning index set $[\mHat]$ 
based on pivotal index subset $L$.}

For each $a\in[k]$, $b\in[a]$, let
\begin{align*}
    L(a,b) \triangleq \{
\ell \in [\ell_a: \ell_{a + 1} - 1]:
\muHat(\ell_{b - 1}) < \muHat(\ell) \leq \muHat(\ell_b)
\}
\end{align*}
where $\muHat(\ell_0) = 0$.
By definition of pivot index subset $L$, 
$\{L(a, b)\}_{b\in[a]}$ is a partition of
$[\ell_a: \ell_{a+1} - 1]$ for each $a\in[k]$,
and $\{L(a, b)\}_{a\in[k],b\in[a]}$
is a partition of index set $[\mHat]$.
See \Cref{fig:pivot ell}
for a graphical illustration.

Though $\{L(a, b)\}_{a\in[k],b\in[a]}$
is a partition
with non-uniform size 
and each $L(a, b)$ 
may not contain a consecutive subset of indexes,
it ensures 
the following desired monotonicity property on $\muHat(\ell)$:
Fix  $a,a'\in[k],~ b\in[a],~ b'\in[a']$ such that
  $b < b'$. We have
\begin{align*}
\muHat(\ell) \leq \muHat(\ell') \qquad
    &\forall \ell\in L(a, b),~
    \ell'\in L(a', b')
\intertext{
Moreover, constraint~\WFRCtriangleineq\ in program~\ref{eq:weakly factor-revealing program}
is satisfied across index subsets $L(a, b)$ and $L(a', b')$
such that $a < a'$. 
Namely, for every $a,a'\in[k], ~b\in[a],~ b'\in[a']$,
if $a < a'$ then}
\discountfactor\cdot \muHat(\ell') 
    \leq \distanceHatStar(\ell) + \distanceHat(\ell) + \distanceHatStar(\ell')
    \qquad
    &\forall \ell\in L(a, b),~
    \ell'\in L(a', b'):&
\end{align*}
Finally, 
the monotonicity of $\muHat(\ell)$
across partitions $\{L(a, b)\}_{a\in[k],b\in[a]}$ 
also enables
us to
remove non-convex term $\min\{\muHat(\ell),\muHat(\ell')\}$
in
constraint~\WFRCcontribution\ in program~\ref{eq:weakly factor-revealing program}
at $\ell = \ell_a$.
Namely, for every $a\in[k - 1]$, the aforementioned constraint can be rewritten as follows:
\begin{align*}
    \sum_{a'\in[a:k]}\sum_{b'\in[a]}\sum_{\ell'\in L(a', b')}
    \plus{
    \discountfactor\cdot \muHat(\ell') - \distanceHat(\ell')
    }
    +
    \sum_{a'\in[a:k]}\sum_{b'\in[a + 1: a']}\sum_{\ell'\in L(a', b')}
    \plus{
    \discountfactor\cdot \muHat(\ell_a) - \distanceHat(\ell')
    }
    \leq \costHat.
\end{align*}

\paragraph{Step 3- batching variables based on partitions
$\{L(a, b)\}_{a\in[k],b\in[a]}$}

Partitions $\{L(a, b)\}_{a\in[k],b\in[a]}$
nicely preserve the feasibility
of all constraints (for most indexes) in program~\ref{eq:weakly factor-revealing program} (as illustrated in the previous step). This suggests a  natural batching procedure for constructing a solution to
\SFRP{k, \discountfactor, \costscalar}
that involves   grouping variables 
in each $L(a, b)$ separately. 
Since the size of each $L(a, b)$ is not identical, we 
introduce additional variables $\{\qHat(a, b)\}_{a\in[k],b\in[a]}$
to keep track of the size.
Consequently, we construct a solution $(\costHat\primed, \{\qHat\primed(a, b), \muHat\primed(a, b),\distanceHat\primed(a, b),
\distanceHatStar\primed(a, b)\})$
for program~\SFRP{k, \discountfactor, \costscalar} as follows,\footnote{For ease of presentation, in 
this part we assume $L(a,b) \not=\emptyset$
for every $a\in[k],b\in[a]$.
The formal argument which covers general cases
is defined in \Cref{apx:from weakly to strongly factor revealing program}.}
\begin{align*}
    & 
    \costHat\primed\gets \costHat
    \\
    a\in[k], b\in[a]:
    \qquad 
    &\muHat\primed(a, b) \gets 
    \sum_{\ell\in L(a, b)} 
    \frac{
    \muHat(\ell)
    }{
    |L(a,b)|
    },
    \quad
    \distanceHat\primed(a, b) \gets 
    \sum_{\ell\in L(a, b)} 
    \frac{ \distanceHat(\ell)
    }{
    |L(a,b)|
    },
    \\
    &\distanceHatStar\primed(a, b) \gets 
    \sum_{\ell\in L(a, b)} 
    \frac{\distanceHatStar(\ell)
    }{
    |L(a,b)|
    },
    \quad
    \qHat\primed(a, b) \gets |L(a,b)|
\end{align*}
See \Cref{fig:pivot ell} for a graphical illustration.

It can be verified that 
the objective value remains unchanged,
and
all constraints of program~\SFRP{k,\discountfactor,\costscalar}
(except constraints~\SFRCdistance\ and \SFRCdensity)
are satisfied.
In the formal proof in \Cref{apx:from weakly to strongly factor revealing program}, 
we show that imposing 
constraint~\SFRCdistance,
does not change the optimal objective value.
We discuss how to handle constraint~\SFRCdensity, in the next step.\footnote{If we remove \SFRCdistance\ or 
\SFRCdensity, \ref{eq:strongly factor-revealing quadratic program}
becomes degenerate and has unbounded optimal objective value.}

\paragraph{Step 4- converting into a feasible solution of program~\ref{eq:strongly factor-revealing quadratic program}.}
In the last step, 
we further convert the solution obtained in step 3 for 
program~\SFRP{k,\discountfactor,\costscalar}
into a feasible solution for program~\ref{eq:strongly factor-revealing quadratic program}
for an arbitrary $\nHat\in\naturals$.
The formal argument 
is 
not difficult but quite detailed.
All technical details can be found in 
\Cref{apx:from weakly to strongly factor revealing program}.

The main technical difficulty in this step
is to design a batching procedure that guarantees
the feasibility of constraint~\SFRCdensity,
i.e.,
    $\forall b \in[\nHat]:
    \sum_{a \in [b:\nHat]}
     \qHat(a, b) = 1$.
To overcome this difficulty, we first apply a 
\emph{decomposition 
procedure} (see \Cref{fig:stretch}  in the appendix)
which converts the solution in program~\SFRP{k,\discountfactor,\costscalar}
into a solution in program~\SFRP{k\primed,\discountfactor,\costscalar}
with $k\primed \geq k$ such that the total mass $\sum_{a\in[b:k\primed]}\qHat(a, b)$ is sufficiently small for each $b\in[k\primed]$.
This enables us to identify a sequence 
$0 = b_0 < b_1 < b_2 < \dots <b_{\nHat} = k\primed$
such that for every $t \in[\nHat]$, $\sum_{b\in [b_{t - 1} + 1: b_t]}
\sum_{a\in[b:n]}\qHat(a, b) = 1$. 
Then, we apply another \emph{batching procedure} (see \Cref{fig:compression} in the appendix)
based on sequence $\{b_t\}_{t\in[\nHat]}$,
and obtain
a feasible solution in program~\ref{eq:strongly factor-revealing quadratic program},
whose objective value is weakly higher than the the objective value 
of the original solution in step 1, as desired.
It is worthwhile highlighting that 
both the decomposition and the batching procedure 
in this step
heavily rely on the 
the monotonicity property imposed in constraint~\SFRCmonotonicity,
which is established in Steps 1 and 2 
based on pivot index subset $L$ and its
induced solution-dependent index partitions.

\section{Numerical Experiments}
\label{sec:numerical}

To 
provide numerical justifications
for the performance of our proposed algorithm,
we
performed numerical experiments on
both synthetic data (\Cref{sec:numerical synthetic}) and
US census data (\Cref{sec:numerical us cities}).

\subsection{Experiments over synthetic data}
\label{sec:numerical synthetic}
We first discuss the numerical experiment
over randomly-generated synthetic data.

\paragraph{Experimental setup.}
In our test problem, there are $n = 30$ locations. 
Each location $i\in[n]$ is 
associated with a two-dimensional coordinate 
$(x_i, y_i)$, 
population $\pop_i$,
and 
facility opening cost $\opencost_i$.
Here, coordinates $x_i, y_i$ are drawn i.i.d.\ 
from the normal distribution $\texttt{Normal}(0, 1)$,
population $\pop_i$
is drawn i.i.d.\ from the exponential distribution 
$\texttt{Exponential}(\sfrac{1}{100})$,
and facility opening cost $\opencost_i$ is drawn 
i.i.d.\ from the exponential distribution 
$\texttt{Exponential}(\sfrac{1}{\costconstructionfactor})$ with $\costconstructionfactor\in\{20, 100\}$.
We consider the Euclidean distance based on 
coordinates $\{(x_i, y_i)\}$
as the
distance function $\distance:[n]\times[n]\rightarrow
\reals_+$.

The number of individuals $\flowi$
for each edge $(i, j)\in\Edge$ is generated as follows.
Each location~$i\in[n]$ is associated with an 
employee attractiveness $\empattract_i$
drawn i.i.d.\ from the exponential distribution 
$\texttt{Exponential}(1)$.
For each edge $(i, j) \in \Edge$,
we set
\begin{align*}
    \flowi \triangleq 
    \displaystyle
    \popi \cdot \frac{
    \empattract_j\cdot 
    \exp\left(-\flowconstructionfactor
    \cdot \distance(i, j)\right)
    }{\sum_{k\in[n]}
    \empattract_k\cdot 
    \exp\left(-\flowconstructionfactor
    \cdot \distance(i, k)\right)
     }
\end{align*}
In this construction, $\flowi$
is inline with the standard MNL model 
by interpreting 
$\log(\empattract_j) - \flowconstructionfactor\cdot \distance(i, j)$
as the value of working at location $j$ for individuals who reside in location $i$.
Consequently, $\flowi$
increases as the employee attractiveness $\empattract_j$
of location $j$
increases.
Parameter~$\flowconstructionfactor$ 
controls the impact of distance between 
locations $i$ and $j$ on $\flowi$.
Here we present results   
for $\flowconstructionfactor = \sfrac{1}{5}$.\footnote{We also
ran our experiments by varying all parameters in our synthetic instances. 
We obtained similar results 
and verified the robustness of our numerical findings.}

\newcommand{\CFAlgStar}{\CFAlg^*}
\newcommand{\CFAlgPStar}{\CFAlgP^*}
\newcommand{\Param}{\texttt{Param}}

\paragraph{Policies.}
In this numerical experiment, 
we compare three different classes of policies:
\begin{enumerate}
    \item \textsl{The 2-Chance Greedy Algorithm}:
    this policy is \Cref{alg:greedy modified}
    parameterized by 
    discount factor~$\discountfactor$ and opening cost scalar $\costscalar$.
    Its approximation ratio is upperbounded by 
    $\SFRapproxratio$.
    We implement this policy with discount factor 
    $\discountfactor\in\{0, 0.2, 0.4, 0.6, 0.8, 1\}$ and opening cost scalar $\costscalar\in\{1, 1 + 0.5 \discountfactor, 1 + \discountfactor\}$,
    and refer it as $\CFAlg(\discountfactor,\costscalar)$.
    Recall that the JMMSV algorithm is a special case, i.e., $\CFAlg(0, 1)$ (\Cref{ob:JMMSV equivalence TWOLFLP}).
    Hereafter, we use $\Param\triangleq \{(\discountfactor,\costscalar):\discountfactor\in\{0, 0.2, 0.4, 0.6, 0.8, 1\}, \costscalar\in\{1, 1 + 0.5 \discountfactor, 1 + \discountfactor\}\}$ to denote the space of discretized $\discountfactor$ and $\costscalar$ in our experiments.
    
    Moreover, for each randomly generated instance, we compute the total cost of $\CFAlg(\discountfactor,\costscalar)$ for all discretized $(\discountfactor,\costscalar)\in\Param$ described above, and then refer the best one as $\CFAlgStar$.
    Namely, for each randomly generated instance $\instance$, 
    $\CFAlgStar \triangleq \CFAlg(\discountfactor^*,\costscalar^*)$ with
    \begin{align*}
        (\discountfactor^*,\costscalar^*) = \argmin_{(\discountfactor,\costscalar)\in\Param} \Cost[\instance]{\CFAlg(\discountfactor,\costscalar)}
    \end{align*}
    where $\Cost[\instance]{\CFAlg(\discountfactor,\costscalar)}$ is the total cost of algorithm $\CFAlg(\discountfactor,\costscalar)$ on instance $\instance$.
    \item \textsl{The 2-Chance Greedy Algorithm
    with Myopic Pruning}:
    this policy combines \Cref{alg:greedy modified} 
    with an additional post-processing step 
    (i.e., myopic pruning) as follows.
    Given $\SOL$ returned by \Cref{alg:greedy modified},
    this policy iteratively checks whether 
    the total cost can be reduced by removing 
    (a.k.a., pruning) 
    a facility 
    from current solution $\SOL$. 
    If such a facility exists, 
    it \emph{greedily prunes} the one with the highest cost reduction, and repeats.
    By construction, the performance
    of this policy 
    is weakly better than the original 2-Chance Greedy Algorithm,
    and its approximation ratio 
    is upperbouned by $\SFRapproxratio$ as well.
    We implement this policy with discretized discount factor 
    $\discountfactor$ and opening cost scalar $\costscalar$ for all $(\discountfactor,\costscalar)\in\Param$,
    and refer it as $\CFAlgP(\discountfactor,\costscalar)$.
    Similar to $\CFAlgStar$, we introduce $\CFAlgPStar$ to denote algorithm $\CFAlgP(\discountfactor,\costscalar)$ with the best discretized $(\discountfactor, \costscalar)$ for each randomly generated instance.
    \item \textsl{The Greedy Algorithm 
    with Home (resp.\ Work) Location}:
    this policy is the JMMSV algorithm 
    (\citealp{JMMSV-02}, \Cref{alg:JMMSV})
    assuming that each individual can only be connected 
    through her home (resp.\ work) location.
    This policy requires the knowledge of 
    population $\{\popi\}$
    (resp.\ employment) for each location,
    and the knowledge of $\{\flowi\}$
    is unnecessary.
    Its approximation ratio in the \TWOLFLP\ 
    is unbounded.
    We refer this policy as $\GDH$ (resp.\ $\GDW$).
\end{enumerate}

\paragraph{Results.}
In order to compare different policies,
we sample 100 randomly generated synthetic instances.
For each instance and each policy, 
we normalize its performance by 
computing the ratio between
the cost of this policy on
this instance 
and 
the cost of policy $\CFAlgPStar$.
Below we discuss our numerical results in detail.

\begin{table}[t]
    \centering
    \caption{Average costs for 
    different policies in \Cref{sec:numerical synthetic}.}
    \label{table:synthetic performance discount}
    \begin{tabular}[t]{ccccc}
    \toprule
    & 
    $\CFAlgStar$ & $\CFAlgPStar$ & $\GDH$ & $\GDW$
     \\
     \midrule
     $\costconstructionfactor = 20$ & 236.23 & 235.59 & 302.03 & 298.63\\
     $\costconstructionfactor = 100$ & 599.27 & 597.94 & 682.28 & 680.71\\
    \bottomrule
    \end{tabular}
\end{table}

\begin{figure}[ht]
  \centering
       \subfloat[$\costconstructionfactor = 20$]
      {\includegraphics[width=0.45\textwidth]{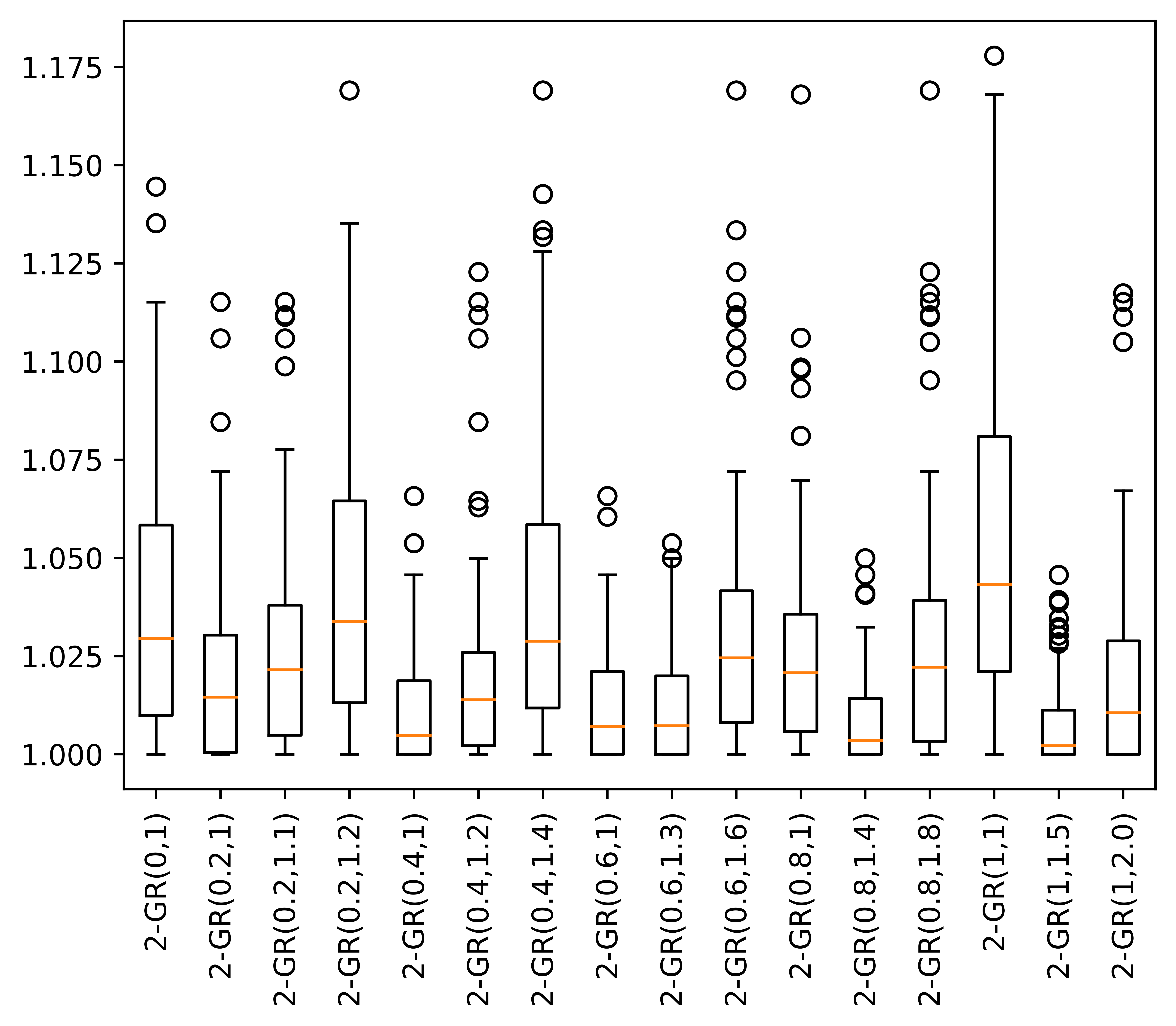}}
      ~~~~
      ~~~~
      \subfloat[$\costconstructionfactor = 100$]
      {\includegraphics[width=0.45\textwidth]{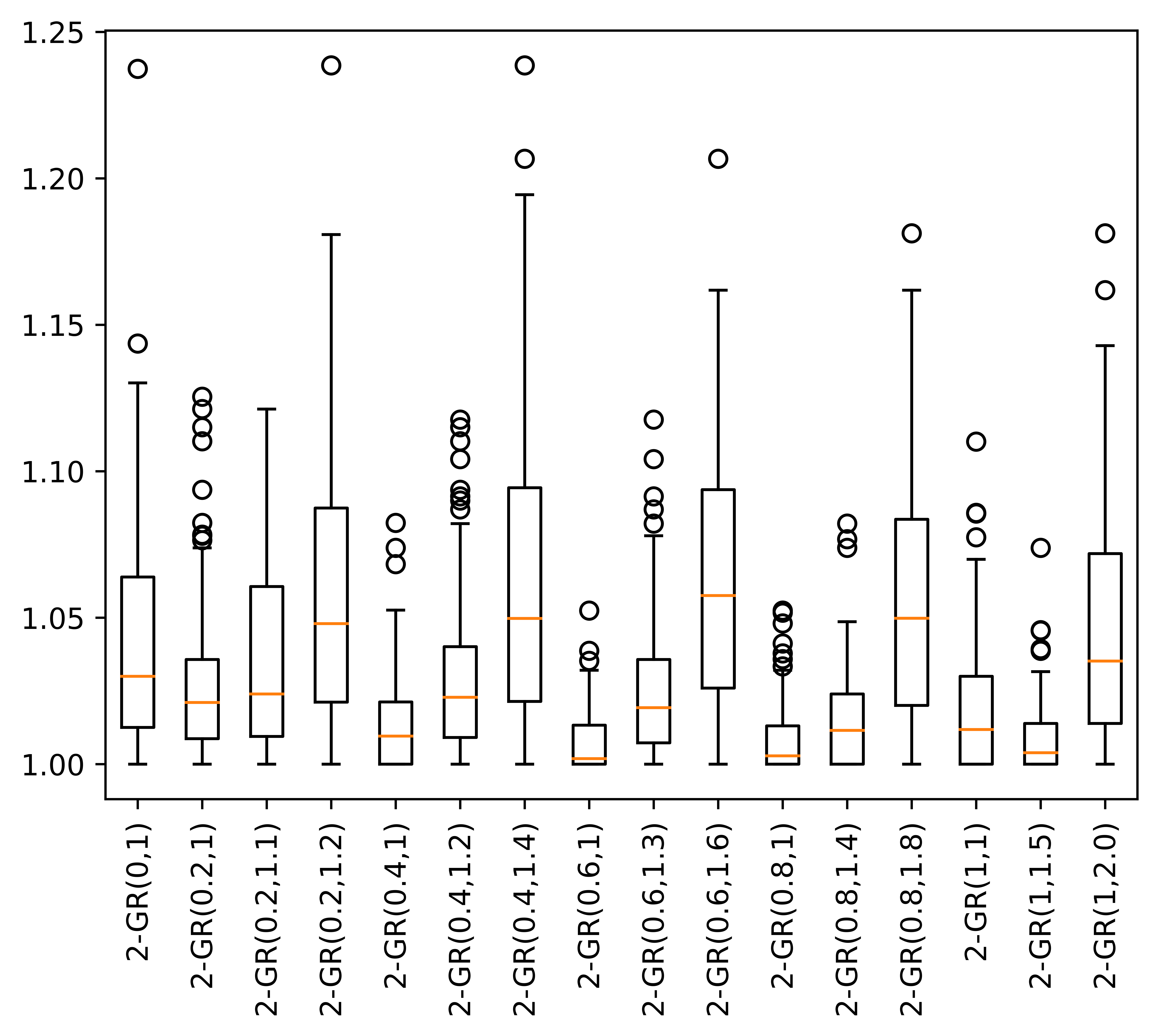}}
  \caption{``Box and whisker'' plots to compare $\CFAlg(\discountfactor,\costscalar)$ with different $(\discountfactor,\costscalar)\in\Param$ in terms of the normalized performance in \Cref{sec:numerical synthetic}.
  }
   \label{fig:numerical synthetic ratio discount}
\end{figure}

\begin{figure}[ht]
  \centering
       \subfloat[$\costconstructionfactor = 20$]
      {\includegraphics[width=0.45\textwidth]{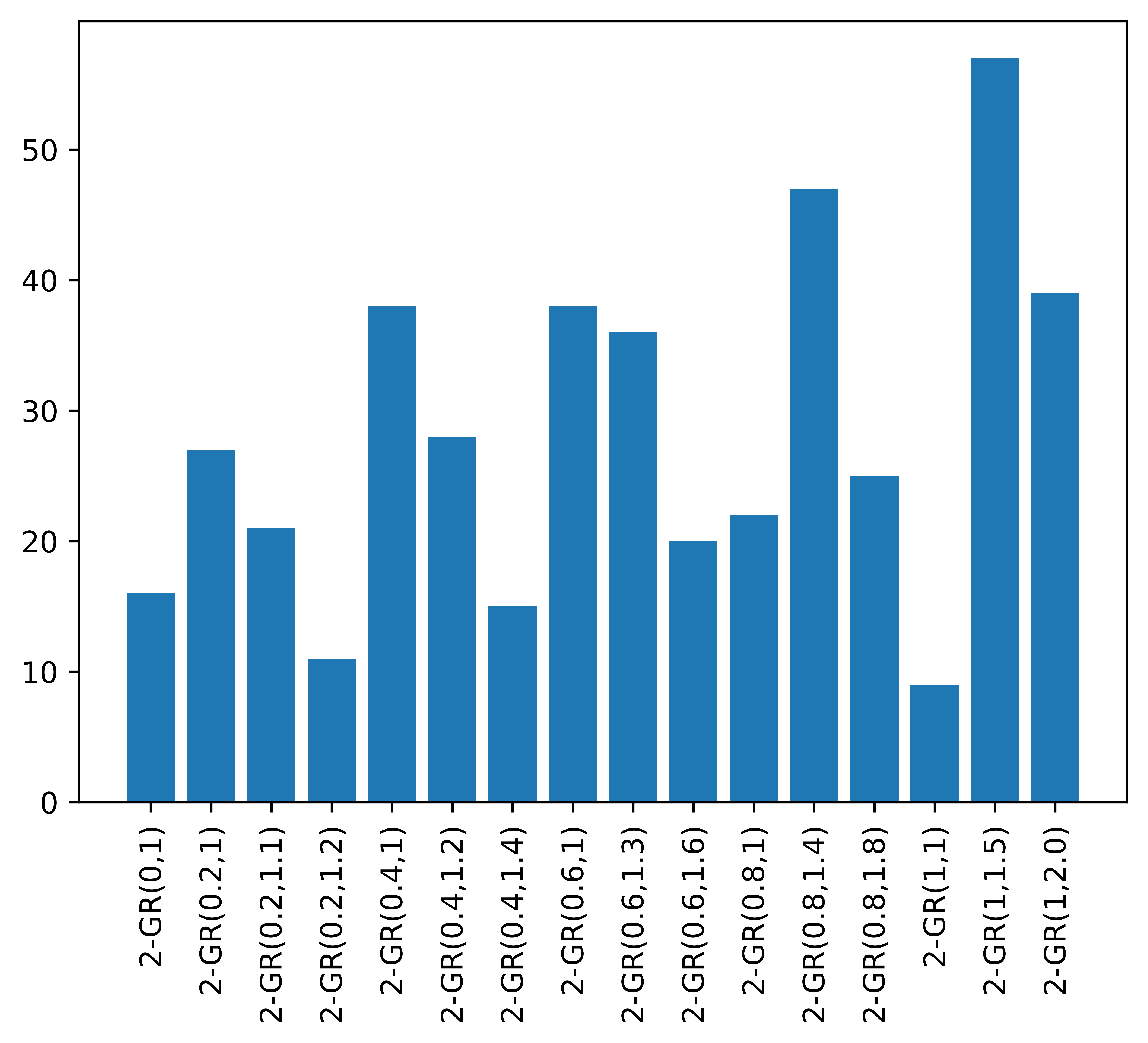}}
      ~~~~
      ~~~~
      \subfloat[$\costconstructionfactor = 100$]
      {\includegraphics[width=0.45\textwidth]{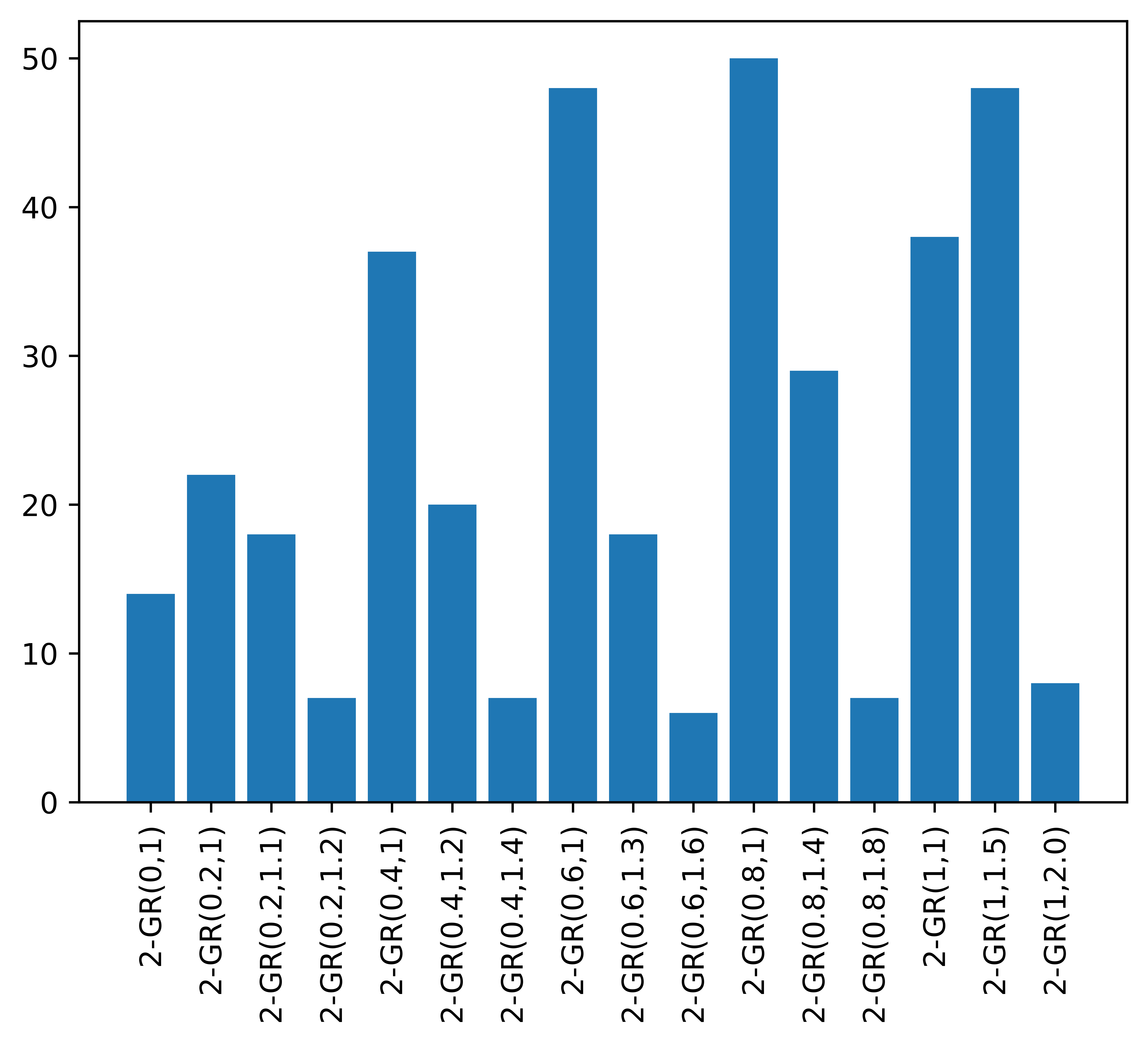}}
  \caption{The frequency that each $(\discountfactor,\costscalar)\in\Param$ achieves minimum cost among all $\CFAlg(\discountfactor,\costscalar)$) in \Cref{sec:numerical synthetic}.
  }
   \label{fig:numerical synthetic freq discount}
\end{figure}

\noindent\textsl{\underline{The performance of the 2-Chance Greedy $\CFAlg(\discountfactor,\costscalar)$ with different parameters $(\discountfactor,\costscalar)$}} 
In \Cref{sec:greedy two location analysis}, we obtain the main theoretical result, i.e., the approximation ratio of $\greedyapprox$ for the 2-Chance Greedy Algorithm, by considering discount factor $\discountfactor = 1$ and opening cost scalar $\costscalar = 2$. 
However, for a particular instance, using $\discountfactor = 1$ and $\costscalar = 2$ may not minimize the cost among all $(\discountfactor,\costscalar)$ assignments. In practice, the social planner can implement the 2-Chance Greedy Algorithm by varying $(\discountfactor,\costscalar)$ and then select the solution with the minimum cost. 
From \Cref{fig:numerical synthetic ratio discount}, we observe that the best average normalized performance among all $\CFAlg(\discountfactor,\costscalar)$ is obtained under $(\discountfactor = 1,\costscalar = 1.5)$ and $(\discountfactor = 0.6, \costscalar = 1)$ for different experiment setups, respectively. In \Cref{fig:numerical synthetic freq discount}, we also plot the frequency that each $(\discountfactor,\costscalar)\in\Param$ achieves minimum cost among all $\CFAlg(\discountfactor,\costscalar)$.

\noindent\textsl{\underline{The impact of discount factor $\discountfactor$ and opening cost scalar $\costscalar$ in $\CFAlg(\discountfactor,\costscalar)$}}
Recall that
the approximation ratio of the 2-Chance Greedy Algorithm
becomes $\Omega(\log n)$
when discount factor $\discountfactor$ approaches zero
(\Cref{example:JMMSV failure}),
while the approximation is constant 
with positive constant discount factor $\discountfactor$,
e.g., a $\greedyapprox$-approximation for $\discountfactor = 1,\costscalar = 2$ 
(\Cref{prop:approx ratio gamma one,prop:approx ratio general gamma}).
From \Cref{fig:numerical synthetic ratio discount},
it can be observed that 
when we fix $\costscalar = 1$, 
the average normalized performance of $\CFAlg(\discountfactor, 1)$ 
is first decreasing and then increasing
for $\discountfactor\in\{0, 0.2, 0.4, 0.6, 0.8, 1\}$.
To understand this behavior,
note that the discount factor $\discountfactor$
controls the impacts 
of partially connected edges 
for opening a new facility.
Loosely speaking, 
$\CFAlg(\discountfactor,\costscalar)$ 
with larger discount factor $\discountfactor$
opens facilities more aggressively
(see \Cref{fig:numerical synthetic size discount}),
and some of those facilities are unnecessary
(i.e., removing some facilities from the final 
solution
reduces the total cost).
On the other direction, since opening cost scalar $\costscalar$ controls the tradeoff between the facility opening cost and individual connection cost, the number of opened facilities decreases (and thus less facility opening cost occurs) when we increase opening cost scalar $\costscalar$ while holding discount factor $\discountfactor$ fixed.

\begin{figure}[ht]
  \centering
       \subfloat[$\costconstructionfactor = 20$]
      {\includegraphics[width=0.45\textwidth]{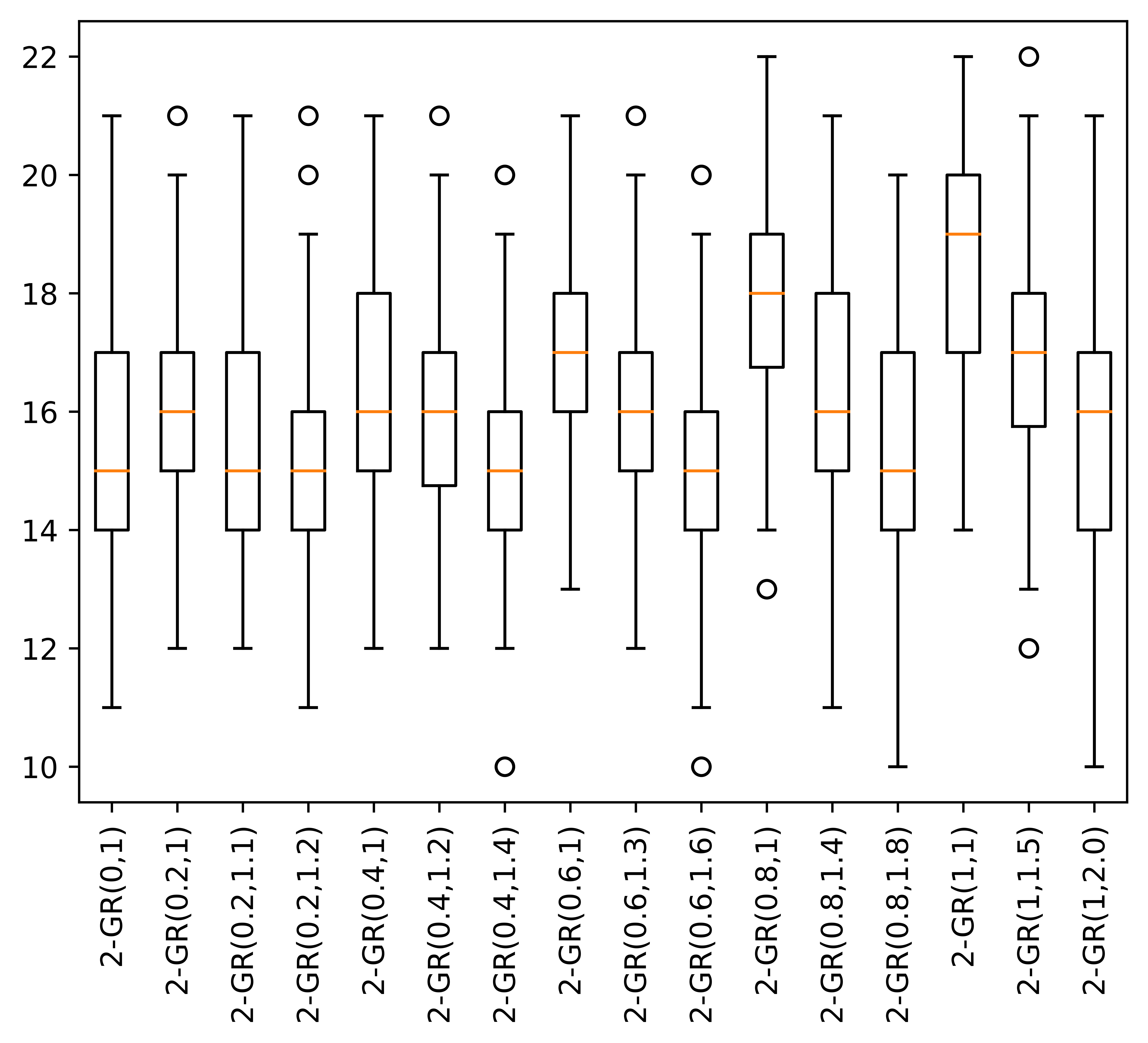}}
      ~~~~
      ~~~~
       \subfloat[$\costconstructionfactor = 100$]
      {\includegraphics[width=0.45\textwidth]{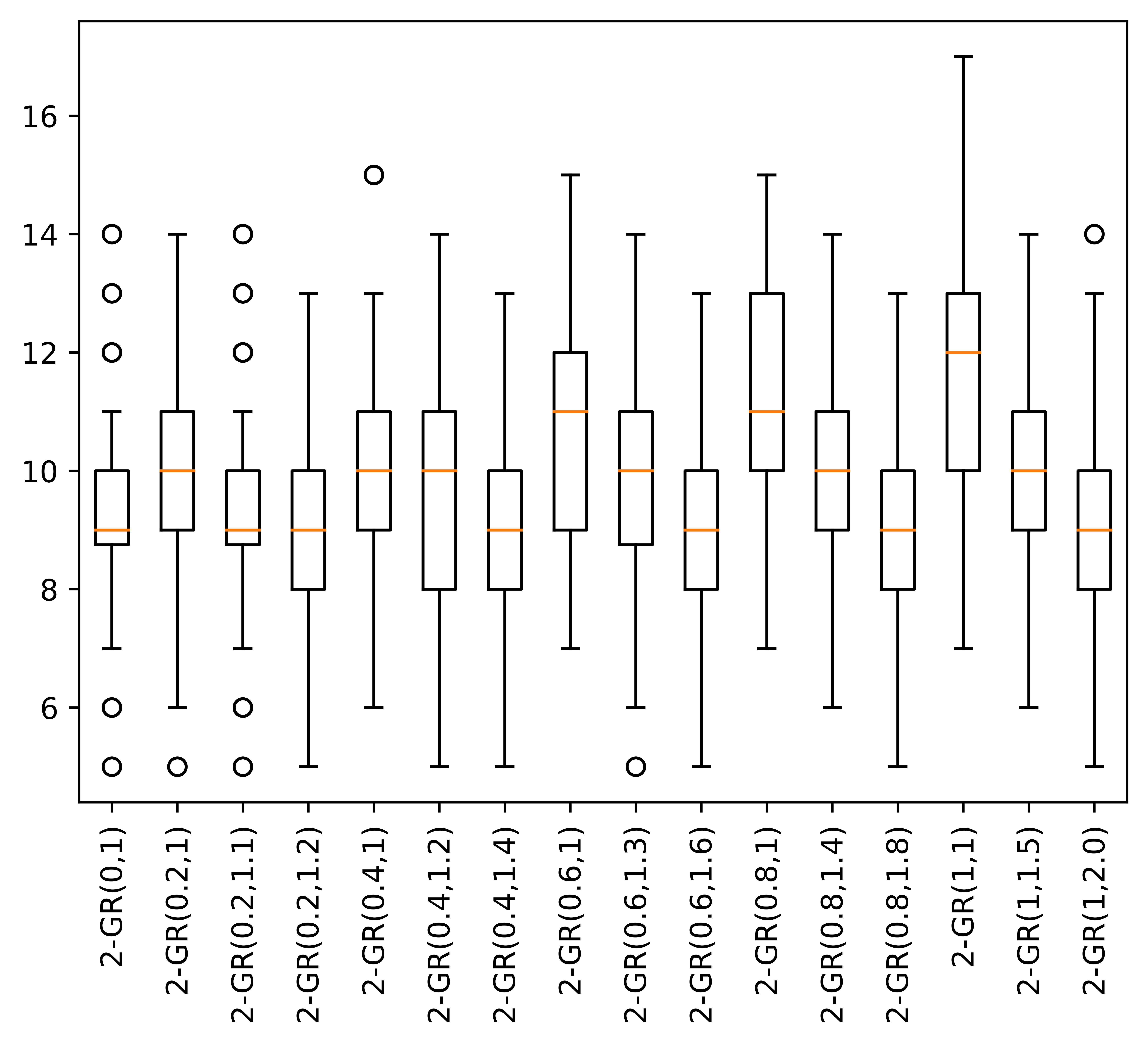}}
  \caption{``Box and whisker'' plots to compare $\CFAlg(\discountfactor,\costscalar)$ with different $(\discountfactor,\costscalar)\in\Param$ in terms of the size of solutions in \Cref{sec:numerical synthetic}.
  }
   \label{fig:numerical synthetic size discount}
\end{figure}

The observation that the 2-Chance Greedy Algorithm might open some unnecessary facilities in fact motivates
the 2-Chance Greedy Algorithm 
with Myopic Pruning ($\CFAlgP(\discountfactor, \costscalar)$)
described above.
From \Cref{table:synthetic performance discount},
we observe that 
the average normalized performance of $\CFAlgPStar$ 
slightly improves. 
Moreover, $\CFAlgPStar$ not only 
outperforms all other policies
on average, but also beats them in 66 (resp.\ 73) out of 100 instances under the experiment setup $\costconstructionfactor = 20$ (resp.\ $\costconstructionfactor = 100$).


\noindent\textsl{\underline{The value of mobility data.}}
When the mobility data $\{\flowi\}$
is not available and the 
social planner only has the 
population (resp.\ employment) 
information
in each location, 
one natural heuristic policy 
is pretending that
each individual can only be connected 
through her home (resp.\ work) location
and then implement $\GDH$ (resp.\ $\GDW$).
It is clear from \Cref{table:synthetic performance discount}
that both $\CFAlgStar$
and $\CFAlgPStar$
perform noticeably better than $\GDH$ and $\GDW$.
Specifically, the average performance gap between
$\CFAlgP(1)$ and $\GDH$ (resp.\ $\GDW$)
is 28.4\%, 14.1\% (resp.\ 27.7\%, 14.1\%), respectively.\footnote{In the 
single-location facility problem,
the solution outputted from JMMSV algorithm 
does not allow local improvement, i.e., 
it is impossible to reduce the total cost 
by removing a facility from this solution.
Therefore, the comparison between 
$\CFAlgPStar$
and $\GDH$, $\GDW$ is fair.}
Such
performance gap highlights the value of mobility 
data $\{\flowi\}$
in this numerical experiment over 
synthetic data.

\subsection{Experiments over US census data}
\label{sec:numerical us cities}
Our second numerical experiment is constructed 
through US census data.

\paragraph{Experimental setup.}
We construct four sets of
\TWOLFLP\ 
instances for 
four cities
in the US, including  
New York City (NYC), Los Angeles metropolitan area (greater LA),
Washington metropolitan area (greater DC),
and Raleigh-Durham-Cary CSA (Research Triangle).

In each instance, 
a location corresponds to a 
Zip Code Tabulation Area (ZCTA).\footnote{ZCTAs are closely 
related to zip codes. The ZCTA code of a block is the most common zip code contained in it; see \citet{ZCTA}.}
For each pair of locations, 
we define their distance as 
the Euclidean distance 
between the centroids of their corresponding 
ZCTAs. 
We set the facility opening cost $\opencost_i$
for each location $i$ as 
\begin{align*}
    \opencost_i \triangleq 
    \costconstructionfactor \cdot 
    \zhvi_i
\end{align*}
where $\zhvi_i$ is the Zillow Home Value Index\footnote{In each instance, 
there are less than 5\% of locations whose 
Zillow Home Value Index (ZHVI) is missing. We set their 
$\zhvi_i$'s as the average ZHVI of other locations.
We also explore other specifications and verify the robustness of our findings.}
recorded in 
\citet{ZHVI},
and $\costconstructionfactor$ is a cost normalization parameter 
which controls the tradeoff between facility opening costs and 
individual connection cost
in our experiments.
Our results are presented for 
different values of $\costconstructionfactor$.

Finally, \citet{LODES} records the total number of individuals 
who reside in one census block (CB) and work in another.
Aggregating it at the ZCTA level, we complete the construction of 
$\{\flowi\}$.
We summarize the number of locations and 
individuals in \Cref{table:instance construction}.

\begin{table}[t]
    \centering
    \caption{The number of locations and individuals
    in \Cref{sec:numerical us cities}.}
    \label{table:instance construction}
    \begin{tabular}[t]{ccccc}
    \toprule
    & 
    NYC & greater LA & greater DC & Research Tri.
     \\
     \midrule
     total locations & 177 & 386 & 313 & 102\\
     total population & 3166075 & 5338918 & 1934411 & 771249\\
    \bottomrule
    \end{tabular}
\end{table}

\paragraph{Policies.}
In this numerical experiment, we consider the same 
set of policies as in \Cref{sec:numerical synthetic}:
$\CFAlg(0, 1)$ (i.e., JMMSV algorithm)
$\CFAlgStar$,
$\CFAlgPStar$,
$\GDH$, and $\GDW$.

\paragraph{Results.}
For each city, we 
construct the \TWOLFLP\ instances
under different values of 
cost normalization parameter $\costconstructionfactor$.
Loosely speaking, 
as parameter $\costconstructionfactor$ increases, 
the magnitude of facility opening cost increases and
consequently the sizes of solutions
in all policies decreases.
We consider two regimes:
\emph{low $\costconstructionfactor$ regime} 
(resp.\ 
\emph{high $\costconstructionfactor$ regime})
such that the facilities are opened 
in 40\% to 80\%
(resp.\ 5\% to 30\%) 
of total locations 
in policy $\CFAlgP(1, 1)$.
For each city and each regime, our results
are qualitatively similar.
\Cref{table:numerical us cities ratio low}
and 
\Cref{table:numerical us cities ratio high}
illustrate the average normalized performance
of different policies 
for each city in the low $\costconstructionfactor$ regime
and the high $\costconstructionfactor$ regime, respectively.
Similar to the experiments over synthetic data in 
\Cref{sec:numerical synthetic},
the normalized performance 
of a policy is defined as 
the ratio between the cost of this policy
and the cost of policy $\CFAlgPStar$.
 Below we discuss our numerical results in detail.

\begin{table}[t]
    \centering
    \caption{The average normalized (by $\CFAlgPStar$) performance 
    in the low $\costconstructionfactor$ regime.}
    \label{table:numerical us cities ratio low}
    \begin{tabular}[t]{ccccc}
    \toprule
    & 
    $\CFAlg(0, 1)$ & $\CFAlgStar$ & $\GDH$ & $\GDW$
     \\
     \midrule
     NYC 
     & 1.041 & 1.005 & 1.23 & 1.1 \\
     greater LA
     & 1.018 & 1.009 & 1.244 & 1.096  \\
     greater DC
     & 1.003 & 1.002 & 1.37 & 1.094 \\
     Research Tri.
     & 1.006 & 1.001 & 1.278 & 1.126  \\
    \bottomrule
    \end{tabular}
\end{table}

\begin{table}[t]
    \centering
    \caption{The average normalized (by $\CFAlgPStar$) performance 
    in the high $\costconstructionfactor$ regime.}
    \label{table:numerical us cities ratio high}
    \begin{tabular}[t]{ccccc}
    \toprule
    & 
    $\CFAlg(0, 1)$ & $\CFAlgStar$ & $\GDH$ & $\GDW$
     \\
     \midrule
     NYC 
     & 1.071 & 1.003 & 1.147 & 1.031  \\
     greater LA
     & 1.045 & 1.01 & 1.131 & 1.035  \\
     greater DC
     & 1.031 & 1.008 & 1.258 & 1.042  \\
     Research Tri.
     & 1.053 & 1.008 & 1.223 & 1.065  \\
    \bottomrule
    \end{tabular}
\end{table}

\noindent\textsl{\underline{JMMSV algorithm vs.\ 2-Chance Greedy Algorithm.}}
Similar to the observation in \Cref{sec:numerical synthetic},
the 2-Chance Greedy Algorithm $\CFAlg(\discountfactor,\costscalar)$ may open some unnecessary facilities and thus $\CFAlgP(\discountfactor,\costscalar)$ using myopic pruning as a post-processing step reduces the cost. 
Nonetheless, it can be observed in \Cref{table:numerical us cities ratio low,table:numerical us cities ratio high} that such cost reduction is  marginal (i.e., less than 1\%).
On the other hand, there is a performance gap between $\CFAlgStar$ (and thus $\CFAlgPStar$) with the JMMSV algorithm (i.e., $\CFAlg(0, 1)$). For NYC, the gap is 4.1\% (7.1\%) in the low (high) $\costconstructionfactor$ regime. For the other three cities, in the high $\costconstructionfactor$ regime, the gap is 4.5\%, 3.1\% and 5.3\%, respectively. 

\noindent\textsl{\underline{The value of mobility data.}}
Similar to the observation in \Cref{sec:numerical us cities},
$\CFAlgStar$ (and thus $\CFAlgPStar$) outperforms 
$\GDH$ and $\GDW$
for all four cities in both regimes.\footnote{Recall that {low $\costconstructionfactor$ regime} 
(resp.\ 
{high $\costconstructionfactor$ regime})
is defined 
such that the facilities are opened 
in 40\% to 80\%
(resp.\ 5\% to 30\%) 
of total locations 
in policy $\CFAlgP(1)$.
In the two extremes (i.e., 
$\costconstructionfactor$
goes to zero
and infinite),
the mobility data 
has no value. 
However,
theoretically speaking,
the value could be substantial
for $\costconstructionfactor$
in between.}
In the low $\costconstructionfactor$ regime,
the performance gap between $\CFAlgPStar$
and $\GDH,\GDW$ is more than 23\%, 9.4\%,
respectively 
(\Cref{table:numerical us cities ratio low}).
In the high $\costconstructionfactor$ regime,
the performance gap between $\CFAlgPStar$
and $\GDH$ is still more than 13.1\%,
but the performance gap between $\CFAlgPStar$
and $\GDW$ becomes 3.1\% to 6.5\%.
(\Cref{table:numerical us cities ratio high}).
Overall our results indicate that in the absence of mobility data, facility placements based on work locations may make sense. However, the value of mobility data is nontrivial in practical settings, and there is significant gain (3-23\% based on assumed parameters) for firms for acquiring and leveraging mobility data in their facility location decisions.

\section{Conclusion and Future Directions}
\label{sec:conclusion}

Motivated by practical applications that utilize mobility data, in this paper, we introduced the 2-location facility location problem. 
We illustrate the shortcomings of the classic greedy algorithm for facility location in our setup and the APX-hardness of computing the optimal solution.
As the main algorithmic contribution of the paper, we propose the 2-Chance Greedy Algorithm. 
By first conducting a primal-dual analysis and then introducing the strongly factor-revealing quadratic program, we prove that the approximation ratio of the 2-Chance Greedy Algorithm is between $\greedyapproxLB$ and $\greedyapprox$. We believe that our new analysis framework, in particular the solution-dependent batching argument for obtaining the strongly factor-revealing program, might be of independent interest. We also present extensive numerical studies that justify the performance of our proposed algorithm in practically relevant settings and highlight the value of mobility data. Finally, we extend our model, algorithm, and analysis to the $K$-location facility location problem.

Several open questions arise from this work. 
The first question is to identify the use of mobility data in other operational problems, such as optimizing transportation systems and scheduling on ride-hailing platforms. Second, inspired by our numerical results, one might want to study (and bound) the value of mobility data theoretically. Since collecting mobility data might be costly, is it possible for a social planner to characterize how much mobility data can improve her decision in isolation or in competitive environments?
What are the distinguishing features of settings where this value is expected to be large versus small?
Third, there is still a gap (i.e., $[2, \greedyapprox]$) for the optimal approximation ratio among polynomial time algorithms for the \TWOLFLP. One can investigate if other algorithmic techniques (e.g., myopic pruning introduced in the numerical section) can be used to further reduce or even close the gap. Fourth, the \TWOLFLP\ in this paper assumes that the individual's connection cost is the \emph{minimum} distance between any of her locations to the closest facility. One can consider other connection cost models\footnote{One example is the ``maximum model'' where the value users derive from facilities depend on the maximum of the distance between a facility and home/work locations. This may be a more appropriate model to consider in the optimization of transit systems, where the users value the proximity of the transit stops both to their origin and to their destination.}
and adapt the machinery and framework
we developed
to these variants.




\bibliography{refs}





\appendix

\section{Connection to  
Vertex Cover Problem}
\label{apx:vertex cover}

In the (weighted) vertex cover problem,
given a (vertex-weighted) undirected graph,
the goal is to find a subset of vertices 
with minimum size (total weights)
such that every edge intersects this subset
(i.e., at least one end point is in this subset).
\citet{KR-08}
establish the following hardness result
for the vertex cover problem.

\begin{theorem}[\citealp{KR-08}]
In the vertex cover problem, there exists no polynomial-time 
with a $2-\epsilon$ approximation guarantee 
under the unique game conjecture.
\end{theorem}

It is straightforward to observe that the \TWOLFLP\ 
generalizes the weighted vertex cover problem.
Given an arbitrary instance in weighted vertex cover problem,
we can construct an instance in the \TWOLFLP\ as follows.
Each vertex is a location, whose facility opening cost
equals to the weight of the vertex.
The number of individuals between each pair of locations is one 
if there is an edge between their corresponding vertices,
and zero otherwise.
The distance between every pair of distinct locations is infinite.
Due to this distance construction, in the \TWOLFLP\ instance,
it is necessary for any solution with bounded total cost
to open a facility at each individual's home or work location.
Therefore, such solution is a valid vertex cover.

In the 
weighted vertex cover problem,
the 2-Chance Greedy Algorithm
with discount factor $\discountfactor = 1$
recovers the classic 
primal-dual algorithm and is a 2-approximation \citep{vaz-01}. 
We formalize this in the theorem below and present 
a proof for completeness.
\begin{theorem}
    In the weighted vertex cover problem,
    the approximation ratio 
    of 
    the 2-Chance Greedy Algorithm 
    with discount factor $\discountfactor = 1$ 
    is 2.
\end{theorem}
\begin{proof}
    For each location $i$, let $\Edge_i$ denote
    the subset of edges whose home or work location is $i$,
    i.e., $\Edge_i = \{\edge\in\Edge: \edgehome = i \lor 
        \edgework = i\}$.
    Now consider the same primal-dual analysis framework
    as the one in \Cref{sec:primal dual framework}.
    Specifically, 
    given \Cref{lem:lp relaxation},
    \Cref{lem:dual assignment objective value} 
    and dual assignment construction~\eqref{eq:dual assignment},
    it is sufficient to show that 
    for each location $i\in[n]$,
        $\sum_{\edge\in\Edge_i} \greedycounter(\edge) \leq \opencost_i$.
    We prove this by contradiction. Suppose 
    there exists location $i\in[n]$ such that 
    the inequality is violated.
    Consider the time stamp $\timeindex$
    when $\sum_{\edge\in\Edge_1} \greedycounter_t(\edge) = \opencost_i$.
    Here $\greedycounter_t(\edge)$ is the value of $\greedycounter(\edge)$
    at time stamp $t$.
    Since the distance in the vertex cover problem is infinite,
    no edge $\edge\in\Edge_i$ has been connected through location $i$.
    Therefore, the condition of Event~(b) for opening location $i$ 
    is satisfied in the algorithm, and thus $\greedycounter(\edge)$
    stops increasing after time stamp $t$ for every edge $\edge\in\Edge_1$.
    This leads to a contradiction.
\end{proof}

\section{The JMMSV Algorithm}
\label{apx:single location FLP}
See \Cref{alg:JMMSV} 
for a formal description of the JMMSV algorithm. 
\begin{algorithm}
    \SetKwInOut{Input}{input}
    \SetKwInOut{Output}{output}

    \Input{
    single location facility location problem instance $\instance = 
    (n, \distance, 
\{\flow_j\}_{j\in[n]}, \opencosts_{i\in[n]})$}
    \Output{subset $\stores \subseteq[n]$
    as the locations of opened facilities.}
 	\caption{The JMMSV algorithm \citep{JMMSV-02}}
 	\label{alg:JMMSV}
 	
 	initialize $\stores \gets \emptyset$
 	
 	initialize $\greedycounter(j) \gets 0$ for each location $j\in[n]$
 	
 	
 	initialize $\firstchanceset \gets [n]$,

    initialize $\connectmapping(j) \gets \nullsymbol$
    for each location $j \in[n]$,
 	
 	\While{$\firstchanceset \not= \emptyset$}{ 
 	    increase $\greedycounter(j)$ by 
 	    $d\greedycounter$
 	    for every location $j \in \firstchanceset$

 	    \tcc{Event (a)}
 	    
 	    \While{there exists location $j \in \firstchanceset$
 	    and location $i\in\stores$ 
 	    s.t.\ $\greedycounter(j) = \distance(j,i)$ 
 	    }
 	    {
 	    
 	  
 	    remove location $j$ from $\firstchanceset$,
 	    i.e., $\firstchanceset\gets \firstchanceset\backslash\{j\}$

 	    connect location $j$ and facility $i$, i.e., 
 	    $\connectmapping(j) \gets i$
 	    }

 	    \tcc{Event (b)}
 	    
 	    \While{there exists location $i\not\in\stores$ s.t.\
 	    $\sum_{j\in \firstchanceset} 
 	    \flow_j\cdot 
 	    \plus{\greedycounter(j) - \distance(j, i)}
 	    = \opencost_i$
 	    }{
 	    
 	    \vspace{1mm}
 	    
 	    add location $i$ into $\stores$, i.e.,
 	    $\stores \gets \stores \cup \{i\}$

 	    \For{each location $j\in \firstchanceset$}{
 	    
 	    \If{$\greedycounter(j) - \distance(j, i)\geq 0$}{
 	    
 	    remove location $j$ from $\firstchanceset$,
 	    i.e., $\firstchanceset\gets \firstchanceset\backslash\{j\}$
 	    
 	    connect location $j$ and facility $i$, i.e., 
 	    $\connectmapping(j) \gets i$
 	    
 	    }
 	    }
 	    }
 	   
 	}
 	\textbf{return} $\stores$
\end{algorithm}

\section{Extensions to $K$-Location 
Facility Location Problem}
\label{apx:k location}

The 2-location facility location problem
(\TWOLFLP)
admits a natural extension,
which we refer as the 
\emph{$K$-location facility location problem
(\KLFLP)}.
The \KLFLP\ is the same as the \TWOLFLP\ 
same except that each individual is
associated with $K$ locations.
In this model, let $\Edge = [n]^K$.
We  
use notation 
$\edge = (\edge_1, \dots, \edge_K) \in \Edge$ to denote 
individuals associated with $K$ locations $e_1, \dots, e_K$,
and refer to $\edge$ as a hyperedge.
The distance between a hyperedge $\edge$ and location $i$ 
is defined as $\distance(\edge, i) = \min_{\edgeindex \in[K]}
\distance(e_\edgeindex, i)$.

The 2-Chance Greedy Algorithm also admits a natural extension,
i.e., $K$-Chance Greedy Algorithm where each hyperedge $e$ 
can be connected $K$ times through each of its $K$ 
associated locations. 
Similar to the 2-Chance Greedy Algorithm,
the $K$-Chance Greedy Algorithm 
is also parameterized by a vector of discount factors 
$1 = \discountfactor_0 \geq \discountfactor_1 \geq \dots 
\geq \discountfactor_K = 0$ and opening cost scalar $\costscalar\in\reals_+$.
We say a hyperedge $\edge$ is $k$-partially connected 
if it has been connected through $k$ associated locations.
The discount factor $\discountfactor_k$ control 
the impact of $k$-partially connected hyperedge 
for opening a new facility.
See \Cref{alg:K greedy modified} for a formal description.

\begin{algorithm}
    \SetKwInOut{Input}{input}
    \SetKwInOut{Output}{output}

    \Input{discount factor $1 = 
    \discountfactor_0
    \geq \discountfactor_1
    \geq \dots \geq 
    \discountfactor_K = 0$, opening cost scalar $\costscalar\in\reals_+$,
    \KLFLP\ instance $\instance = 
    (n, \distance, 
    \{\flowe\}, \opencosts)$}
    \Output{subset $\stores \subseteq[n]$
    as the locations of opened facilities.}
 	\caption{$K$-Chance Greedy Algorithm}
 	\label{alg:K greedy modified}
 	
 	initialize $\stores \gets \emptyset$
 	
 	initialize $\greedycounter(\edge) \gets 0$ for each 
  hyperedge $\edge\in\Edge$
 	
 	
 	initialize $\firstchanceset \gets \Edge$,

    initialize $\connectmapping(\edge,\edgeindex) \gets \nullsymbol$
    for each hyperedge $\edge \in\Edge$,
    each $\edgeindex\in[K]$
 	
 	\While{$\firstchanceset \not= \emptyset$}{ 
 	    increase $\greedycounter(\edge)$ by 
 	    $d\greedycounter$
 	    for every hyperedge $\edge \in \firstchanceset$

 	    \tcc{Event (a)}
 	    
 	    \While{there exists hyperedge $\edge \in \firstchanceset$
 	    and location $i\in\stores$ 
 	    s.t.\ $\greedycounter(\edge) = \distance(\edge, i)$ 
 	    }
 	    {
 	    
 	  
 	    remove hyperedge $\edge$ from $\firstchanceset$,
 	    i.e., $\firstchanceset\gets \firstchanceset\backslash\{\edge\}$
 	    
 	     \For{each $\edgeindex \in [K]$}{
 	    \If{$\greedycounter(\edge) = \distance(\edge_\edgeindex,i)$}{
 	    
 	    connect hyperedge $\edge$ and facility $i$ through location $\edge_\edgeindex$, i.e., 
 	    $\connectmapping(\edge,\edgeindex) \gets i$
 	    }
 	    }

 	    }

 	    \tcc{Event (b)}
 	    
 	    \While{there exists location $i\not\in\stores$ s.t.\
 	    $\sum\limits_{\edge\in \firstchanceset} 
 	    \flowe\cdot 
 	    \plus{\greedycounter(\edge) - \distance(\edge, i)}
 	    +
 	    \sum\limits_{
 	    \subalign{
 	    &\edge\not\in \firstchanceset
 	    }}
 	    \flowe\cdot 
 	    \plus{
        \discountfactor_{\|\connectmapping(\edge,\cdot)\|_0}
        \cdot \greedycounter(\edge) - 
        \min\limits_{\edgeindex\in[K]:
        \connectmapping(\edge,\edgeindex) = \nullsymbol}
        \distance(\edge_\edgeindex,i)}
 	    = \costscalar\cdot \opencost_i$
 	    }{
 	    \vspace{1mm}

      \tcc{subscript $\|\connectmapping(\edge,\cdot)\|_0$
      is defined as $\|\connectmapping(\edge,\cdot)\|_0 \triangleq
    |\{\edgeindex \in[K]: \connectmapping(\edge, \edgeindex) = \nullsymbol\}|$}
    
 	    \vspace{1mm}
 	    
 	    add location $i$ into $\stores$, i.e.,
 	    $\stores \gets \stores \cup \{i\}$

 	    \For{each hyperedge $\edge\not\in \firstchanceset$}{
 	    \For {each $\edgeindex \in [K]$}{
 	  
 	    \If{$\connectmapping(\edge,\edgeindex) = \nullsymbol$ 
 	    and $\discountfactor\cdot \greedycounter(\edge) \geq \distance(\edge_\edgeindex, i) $}{
 	  
 	    connect hyperedge $\edge$ and facility $i$ through location $\edge_\edgeindex$, i.e., 
 	    $\connectmapping(\edge,\edgeindex) \gets i$
 	   
 	    }
 	    }
 	    }
 	    
 	    \For{each hyperedge $\edge\in \firstchanceset$}{
 	    
 	    \If{$\greedycounter(\edge) - \distance(\edge, i)\geq 0$}{
 	    
 	    remove hyperedge $\edge$ from $\firstchanceset$,
 	    i.e., $\firstchanceset\gets \firstchanceset\backslash\{\edge\}$
 	    
 	    \For{each $\edgeindex \in [K]$}{
 	    \If{$\greedycounter(\edge) \geq \distance(\edge_\edgeindex,i) $}{
 	    
 	    connect hyperedge $\edge$ and facility $i$ through location $\edge_\edgeindex$, i.e., 
 	    $\connectmapping(\edge,\edgeindex) \gets i$
 	    }
 	    }
 	    }
 	    }
 	    }
 	   
 	}
 	\textbf{return} $\stores$
\end{algorithm}

Inspired by the parameter choice of the 2-Chance Greedy Algorithm, in the remaining of this section, we focus on
the $K$-Chance Greedy Algorithm with discount factors 
$\discountfactor_0 = \discountfactor_1 = \dots 
\discountfactor_{K - 1} = 1$ and opening cost scalar $\costscalar = K$,
and refer it as the \emph{Canonical $K$-Chance Greedy Algorithm}
for simplicity. 
The main theoretical guarantee of this algorithm is as follows.

\begin{theorem}
\label{thm:K greedy modified}
In the \KLFLP,
the approximation ratio of 
the Canonical $K$-Chance Greedy Algorithm
is 
at most equal to $\RelaxSFRapproxratio(K)$,
where 
\begin{align*}
    \RelaxSFRapproxratio(K) = 
    \inf\limits_{\nHat\in\naturals}
    \Obj{\text{\ref{eq:relax strongly factor-revealing quadratic program}}}
\end{align*}
Here \ref{eq:relax strongly factor-revealing quadratic program}
is the maximization program parameterized by $\nHat\in\naturals$
defined as follows:
\begin{align}
\tag{$\RelaxSFRP{\nHat}$}
\label{eq:relax strongly factor-revealing quadratic program}
    \begin{array}{clll}
    \max\limits_{
    \substack{\costHat\geq 0,
    \\
    \qHatBf,
    \muHatBf,
    \distanceHatBf,
    \distanceHatStarBf \geq 
    \zerobf
    }} &
     \displaystyle\sum\nolimits_{a\in[\nHat],b\in[a]}
     \qHat(a,b)\cdot 
    \muHat(a,b)  
    &  & \\
    \text{s.t.} &\SFRCmonotonicity \text{--}
    \SFRCdensity ~~
    \text{with $\discountfactor = 1$ and $\costscalar = K$}
    &
    \end{array}
\end{align}
\end{theorem}

\begin{remark}
The constraints 
in program~\ref{eq:relax strongly factor-revealing quadratic program}
are the same as the constraints 
in
program~$\SFRP{\nHat, 1, K}$.
Furthermore, the objective function 
in program~\ref{eq:relax strongly factor-revealing quadratic program}
is a natural extension of the objective function of 
program~$\SFRP{\nHat, 1, 2}$.
Specifically, program~\hyperref[eq:relax strongly factor-revealing quadratic program]{$\mathcal{P}_{\texttt{SFR-2}}(\nHat)$} for $K=2$
recovers program~$\SFRP{\nHat, 1, 2}$.
\end{remark}

\begin{proof}[Proof of \Cref{thm:K greedy modified}]
The analysis follows the same argument as \Cref{thm:greedy modified} with two modification on
(a) the dual assignment construction~\eqref{eq:dual assignment};
and (b) the solution construction for \Cref{lem:dual assignment approx feasible}:
Let $\kcountedgeset$ be the hyperedge subset where each hyperedge $\edge\in\kcountedgeset$ are connected to $k$ different facilities 
in the Canonical $K$-Chance Greedy Algorithm.
We assume for every hyperedge $\edge \in \kcountedgeset$,
$\connectmapping(\edge, \edgeindex)\in\stores$ for every $\edgeindex\in[k]$,
and $\distance(\edge_{\edgeindex}, \connectmapping(\edge,\edgeindex)) \leq 
\distance(\edge_{\edgeindex + 1}, \connectmapping(\edge,\edgeindex + 1))$
for every $\edgeindex \in[k - 1]$.
This is without loss of generality, since the role of $K$ locations of a fixed hyperedge are ex ante symmetric in our model. 
For each $k\in[K]$ and each hyperedge $\edge\in\kcountedgeset$, we construct the dual assignment as
\begin{align*}
    \duale \gets \flowe\cdot 
    \left(
    \greedycounter(\edge) 
    - 
    \frac{1}{k}
    \left(
    \sum_{\edgeindex\in[k]}
    \distance(\edge_{\edgeindex}, \connectmapping(\edge,\edgeindex))
    \right)
    +
    \distance(\edge_{\texttt{1}}, \connectmapping(\edge,\texttt{1}))
    \right)
\end{align*}
Regarding the solution construction for \Cref{lem:dual assignment approx feasible} in the \KLFLP, we let variable $\distanceHatStar\primed(\edgetoellmapping(\edge))\gets \normalizefactor\cdot \greedycounter(\edge)$ for all hyperedge $\edge\in\EdgeTilde$ and leave all other variables' construction remaining the same.
Since then, all remaining analysis can be extended straightforwardly. We omit them to avoid
repetition.
\end{proof}

Note that $\Obj{\text{\hyperref[eq:relax strongly factor-revealing quadratic program]{$\mathcal{P}_{\texttt{SFR-K}}(\nHat)$}}} \leq \frac{K}{K'}
\Obj{\text{\hyperref[eq:relax strongly factor-revealing quadratic program]{$\mathcal{P}_{\texttt{SFR-K}'}(\nHat)$}}}$
for every $K \geq K'$, since any feasible solution $(\costHat,
\qHatBf,
\muHatBf,
\distanceHatBf,
\distanceHatStarBf)$ of program~\hyperref[eq:relax strongly factor-revealing quadratic program]{$\mathcal{P}_{\texttt{SFR-K}}(\nHat)$} can be converted into a feasible solution $(\costHat' \gets \costHat,
\qHatBf' \gets\qHatBf,
\muHatBf'\gets \frac{K'}{K}\muHatBf,
\distanceHatBf'\gets\frac{K'}{K}\distanceHatBf,
\distanceHatStarBf'\gets\frac{K'}{K}\distanceHatStarBf)$
of the program~\hyperref[eq:relax strongly factor-revealing quadratic program]{$\mathcal{P}_{\texttt{SFR-K}'}(\nHat)$}.
Therefore, it suffices to evaluate program~\ref{eq:relax strongly factor-revealing quadratic program} for small $K$'s which also provides upperbounds for larger $K$.
In particular, we numerically compute $\Obj{\RelaxSFRP{25}}$ for $K = 1, \dots, 20$
with software Gurobi,
and upperbound the approximation ratios of the Canonical $K$-Chance Greedy Algorithm in \Cref{table:approx ratio KFLP}. 

The linear dependence on $K$ in the approximation guarantee
is unavoidable 
for all polynomial-time algorithms.
Similar to our discussion in \Cref{apx:vertex cover},
the \KLFLP\ generalizes the vertex cover problem 
for $K$-uniform hypergraphs,
which is hard to 
approximate within any constant factor better than $K$
under the unique game conjecture \citep{KR-08}.
Observed that $\frac{1}{K}\cdot \Obj{\text{\hyperref[eq:relax strongly factor-revealing quadratic program]{$\mathcal{P}_{\texttt{SFR-K}}(25)$}}}$ is monotone decreasing in $K$ as expected in \Cref{table:approx ratio KFLP}.
In fact, we conjecture that $\lim\limits_{K\rightarrow \infty}\lim\limits_{\nHat\rightarrow\infty} \frac{1}{K}\cdot \Obj{\text{\ref{eq:relax strongly factor-revealing quadratic program}}} = 1$ and thus the Canonical $K$-Chance Greedy Algorithm attains asymptotic optimal approximation ratio of $K + o(1)$ as $K$ becomes large.

\begin{table}
    \centering
    \caption{The approximation ratio upper bounds (UB) of the Canonical $K$-Chance Greedy Algorithm.}
    \label{table:approx ratio KFLP}
    \begin{tabular}[t]{cccccccccccc}
    \toprule
    $K$    &  
     1 & 2 & 3 & 4 & 5 & 6 & 7 & 8 & 9 & 10 & 11
     \\
     \midrule
     UB & 1.864 & 2.497 & 3.538 & 4.58 & 5.611 & 6.659 & 7.685 & 8.714 & 9.769 & 10.816 & 11.855
     \\
    \bottomrule
    \\
     \addlinespace[-\aboverulesep] 
    \cmidrule[\heavyrulewidth]{1-11}
    $K$     
    & 12 & 13 & 14 & 15 & 16 & 17 & 18 & 19 & 20 & $\geq 21$ & 
     \\
    \cmidrule{1-11}
     UB & 12.887 & 13.912 & 14.93 & 15.941 & 16.944 & 18.0 & 19.059 & 20.118 & 21.176 
     &  $1.059K$
     &
     \\
    \cmidrule[\heavyrulewidth]{1-11}
     \addlinespace[-\belowrulesep] \\
    \end{tabular}
\end{table}
 
\section{Missing Proofs}
\label{apx:missing proofs}

\subsection{Proof of \cref{thm:greedy modified}}
\label{apx:greedy modified}
\greedymodified*

\begin{proof}
Combining \Cref{lem:weakly factor-revealing program}
and \Cref{lem:from weakly to strongly factor revealing program}
finishes the proof.
\end{proof}

\subsection{Proof of \Cref{prop:approx ratio general gamma}}
\label{apx:approx ratio general gamma}
\approxratiogeneralgamma*

\begin{proof}
Consider constraint~\SFRCcontribution\ 
at $a = 1$,
\begin{align*}
    \costscalar \cdot \costHat &\geq 
    \qHat(2, 1)
    \cdot \plus{\discountfactor \cdot \muHat(2, 1) - \distanceHat(2, 1)}
    +
    \qHat(2, 2)
    \cdot \plus{\discountfactor \cdot \muHat(1, 1) - \distanceHat(2, 2)}
    \\
    &\geq 
    \qHat(2, 2)
    \cdot \discountfactor \cdot \muHat(1, 1) - \distanceHat(2, 2)
    \overset{(a)}{\geq}
    \discountfactor \cdot \muHat(1, 1) - \distanceHat(2, 2)\
    \intertext{
where inequality~(a) holds due to constraint~\SFRCdensity\ at $b = 2$.
Hence,}
    \muHat(1, 1) &\leq \frac{1}{\discountfactor}\cdot \left(
    \costscalar\cdot \costHat + \distanceHat(2, 2)
    \right)
    \overset{(a)}{\leq} \frac{\costscalar}{\discountfactor}
    \intertext{
    where inequality~(a) holds due to constraints~\SFRCtotalcost\ and \SFRCdensity\ at $b = 2$.
    Invoking constraint~\SFRCtriangleineq\ at $a = 1, a' = 2, b = 1, b' = 2$,}
    \discountfactor\cdot \muHat(2, 2) 
    &\leq \distanceHatStar(1, 1) + \distanceHat(1, 1)
    + \distanceHat(2, 2)
    \overset{(a)}{\leq}
    \muHat(1, 1) + \muHat(1, 1) 
    + \distanceHat(2, 2)
    \overset{(b)}{\leq}
    \frac{2\costscalar}{\discountfactor} + 1
    \leq \frac{3\costscalar}{\discountfactor}
    \intertext{
where inequality~(a) holds due to constraints~\SFRCconnectcost\ and 
\SFRCdistance\ at $a = 1, b = 1$;
and inequality~(b) holds since $\muHat(1, 1) \leq \frac{1}{\discountfactor}$
and $\distanceHat(2, 2) \leq 1$ implied by constraints~\SFRCtotalcost\ and \SFRCdensity\ at $b = 2$.
Putting all pieces together,}
&~~~~\Obj{\SFRP{2,\discountfactor,\costscalar}}
\\
&=
\qHat(1, 1)\cdot 
\left(\frac{1+\discountfactor}{\costscalar}\cdot \muHat(1, 1) - 
    \left(
    \frac{1 + \discountfactor}{\costscalar} - 1\right)\cdot \distanceHatStar(1, 1)\right)
    \\
&\qquad +
\qHat(2, 1)\cdot 
\left(\frac{1+\discountfactor}{\costscalar}\cdot \muHat(2, 1) - 
    \left(
    \frac{1 + \discountfactor}{\costscalar} - 1\right)\cdot \distanceHatStar(2, 1)\right)
\\
&\qquad
+
\qHat(2, 2)\cdot 
\left(\frac{1+\discountfactor}{\costscalar}\cdot \muHat(2, 2) - 
    \left(
    \frac{1 + \discountfactor}{\costscalar} - 1\right)\cdot \distanceHatStar(2, 2)
    \right)
\\
&\leq 
\frac{1+\discountfactor}{\costscalar}
\cdot 
\left(
\qHat(1, 1)\cdot\muHat(1, 1) + 
\qHat(2, 1)\cdot\muHat(2, 1) + 
\qHat(2, 2)\cdot\muHat(2, 2)  
\right)
\\
&\overset{(a)}{\leq}
\frac{1+\discountfactor}{\costscalar}
\cdot 
\left(
\qHat(1, 1)+ 
\qHat(2, 1) + 
\qHat(2, 2)
\right)\cdot\muHat(2, 2)  
\\
&\overset{(b)}{\leq}
\frac{1+\discountfactor}{\costscalar}
\cdot 2 
\cdot 
\frac{3\costscalar}{\discountfactor^2}
=
\frac{6(1 + \discountfactor)}{\discountfactor^2}
\end{align*}
where inequality~(a) holds due to constraint~\SFRCmonotonicity,
and inequality~(b) holds since $\muHat(2, 2) \leq \frac{3\costscalar}{\discountfactor^2}$ and 
constraint~\SFRCdensity.
\end{proof}

\subsection{Proof of \Cref{prop:approx ratio gamma one lower bound}}
\label{apx:approx ratio gamma one lower bound}
\approxratiogammaonelowerbound*

To prove \Cref{prop:approx ratio gamma one lower bound},
we introduce the following lemma.
\begin{lemma}
\label{lem:LB LP to instance}
    There exists a \TWOLFLP\ instance such that 
    the approximation of the 2-Chance Greedy Algorithm
    with discount factor $\discountfactor = 1$ and opening cost scalar $\costscalar = 2$ is at least 
    equal to 
    $\sup_{\mHat\in\naturals} \Obj{\LBLP{\mHat}}$,
    where program~\ref{eq:lower bound LP}
    is the following maximization program parameterized by $\mHat\in\naturals$,
    \begin{align}
\tag{$\LBLP{\mHat}$}
\label{eq:lower bound LP}
&\arraycolsep=1.4pt\def\arraystretch{2.2}
    \begin{array}{llll}
    \max\limits_{
    \substack{\costHat\geq 0,
    \\
    \muHatBf,
    \distanceHatBf,
    \distanceHatStarBf \geq 
    \zerobf
    }} &
     \displaystyle\sum\nolimits_{\ell\in[\mHat]}
    \muHat(\ell) 
    & \text{s.t.} & \\
    &
    \muHat(\ell) \leq \muHat(\ell')
    &
    \ell,\ell'\in[\mHat],
    \ell \leq \ell'
    \\
     &
    \muHat(\ell') \leq 
    \distanceHatStar(\ell) + \distanceHat(\ell)
    +
    \distanceHat(\ell')
    &
    \ell,\ell'\in[\mHat],
    \ell \leq \ell'
    &
    \\
    &   
    \distanceHatStar(\ell) \leq \muHat(\ell)
    &
    \ell\in[\mHat]
    \\
    &
    \displaystyle\sum\nolimits_{\ell'\in[\ell:\mHat]}
    \plus{\muHat(\ell) - \distanceHat(\ell')}
    \leq 2\costHat\qquad
    &
    \ell \in[\mHat]
    \\
      &
     \costHat 
     +
     \displaystyle\sum\nolimits_{\ell\in[\mHat]}
     \distanceHat(\ell)
     = 1
     &
    \end{array}
\end{align}
\end{lemma}
It is worth highlighting that 
program~$\SFRP{\mHat, 1, 2}$
can be considered as a relaxation of 
program~\ref{eq:lower bound LP}.
In particular, from every 
feasible solution $(\costHat, \{\muHat(\ell), \distanceHat(\ell),\distanceHatStar(\ell)\}_{\ell\in[\mHat]})$
of program~\ref{eq:lower bound LP},
we can straightforwardly construct a feasible solution 
$(\costHat\primed,\{\muHat\primed(a, b), 
\distanceHat\primed(a, b), \distanceHatStar\primed(a, b), \qHat\primed(a, b)\}_{a\in[\mHat], b\in[a]})$
of program~$\SFRP{\mHat, 1, 2}$
as follows,
\begin{align*}
&\costHat\primed \gets \costHat,
\\
a\in[\mHat], b\in[a]:\qquad &
\muHat\primed(a, b) \gets \muHat(a),
\quad 
\distanceHat\primed(a, b) \gets \distanceHat(a),
\\
\qquad &
\distanceHatStar\primed(a, b) \gets \distanceHatStar(a),
\quad
\qHat\primed(a, b) \gets \indicator{a = b}
\end{align*}
and both solutions have the same objective value.

Given \Cref{lem:LB LP to instance},
the proof of \Cref{prop:approx ratio gamma one lower bound}
is based on the lowerbound of program~$\LBLP{500}$
numerically computed with software Gurobi.

\begin{proof}[Proof of \Cref{lem:LB LP to instance}]
Given any feasible solution 
$(\costHat\primed, \{\muHat\primed(\ell), \distanceHat\primed(\ell),
\distanceHatStar\primed(\ell)\})$
in program~\ref{eq:lower bound LP},
we can construct a \TWOLFLP\ instance such that the 
approximation ratio of the 2-Chance Greedy Algorithm with $\discountfactor = 1$
and $\costscalar = 2$
equals to the objective value of this solution.

Fix an arbitrary small $\epsilon > 0$.
In this \TWOLFLP\ instance, there are $n \triangleq 4 \mHat + 1$ locations
with facility opening cost as follows,
\begin{align*}
    \opencost_i = \left\{
    \begin{array}{ll}
     \infty    &\quad \forall i \in[2\mHat]  \\
     \muHat\primed(i - 2\mHat) - \distanceHatStar\primed(i - 2 \mHat) 
     & \quad \forall i \in[2\mHat + 1: 3\mHat] \\
     \muHat\primed(i - 3\mHat) - \distanceHatStar\primed(i - 3 \mHat)
     & \quad \forall i \in[3\mHat + 1: 4\mHat] \\
     \costHat\primed + \epsilon
     & \quad~~ i = 4\mHat + 1 
    \end{array}
    \right.
\end{align*}
The distance function $\distance$ satisfies that
\begin{align*}
    \forall i \in [\mHat]:
    &\qquad 
    \distance(i, i + 2\mHat) = \distanceHatStar\primed(i),
    \quad 
    \distance(i + \mHat, i + 3\mHat) = \distanceHatStar\primed(i),
    \\
    &\qquad 
    \distance(i, n) = \distanceHat\primed(i),
    \quad\quad\quad
    \distance(i + \mHat, n) = \infty
\end{align*}
and triangle inequality holds with equality for all other pairs
of locations.

There are $\mHat$ individuals where each individual $i\in[\mHat]$
resides at location $i$ and works at location $i + \mHat$. 
Namely, the number of individuals $\{\flowe\}$ are 
\begin{align*}
    \flowe = \left\{
    \begin{array}{ll}
       1  &\quad \text{if $\edgehome \in [\mHat]$ 
       and $\edgework = \edgehome + \mHat$}  \\
       0  & \quad \text{o.w.}
    \end{array}
    \right.
\end{align*}

The optimal solution opens a facility at location $4\mHat + 1$
(i.e., $\OPT = \{4\mHat + 1\}$
with optimal total cost $\Cost{\OPT} = \opencost\primed + \epsilon
+ \sum_{i\in[\mHat]}\distanceHat\primed(i) = 1 + \epsilon$.

In contrast, consider the 2-Chance Greedy Algorithm with $\discountfactor = 1$ and $\costscalar = 2$.
At time stamp $\muHat\primed(1)$, 
the conditions of Event~(b) for opening facility $2\mHat + 1$ and facility $3\mHat + 1$
are satisfied due to individual $1$ who resides at location $i$
and works at location $\mHat + 1$.
Thus, the algorithm opens both facility $2\mHat + 1$ and facility $3\mHat + 1$,
and fully connects individual 1 to those two facilities 
through her home $1$ and work location $1 + \mHat$, respectively.
Then at time stamp $\muHat\primed(2)$,
facilities at location $2\mHat + 2$ and $3\mHat + 2$ are opened due to 
individual $2$.
Proceeding similarly, it can be seen that when the algorithm terminates, 
it outputs $\SOL = \{2\mHat + 1, \dots, 4\mHat\}$ with total cost
$\Cost{\SOL} = \sum_{i\in[\mHat]}2\muHat\primed(i) 
- \distanceHatStar\primed(i)\geq \sum_{i\in[\mHat]}\muHat\primed(i) $.

Putting all pieces together, as $\epsilon$ goes to zero,
the approximation of the 2-Chance Greedy Algorithm with $\discountfactor = 1$ and $\costscalar = 2$
converges to the objective value of solution 
$(\costHat\primed, \{\muHat\primed(\ell), \distanceHat\primed(\ell),
\distanceHatStar\primed(\ell)\})$ as desired.
\end{proof}

\subsection{Proof of \Cref{lem:structural lemma}}
\label{apx:structural lemma}
\structural*

\begin{proof}
We first show property~(i). By definition of $\permuTGD$,
if $\permuTGD(\edge,\edgeindex) < \permuTGD(\edge',\edgeindex')$,
edge $\edge$ is connected through home/work location $\edge_\edgeindex$
before the termination of the algorithm,
and thus
$\connectmapping(\edge,\edgeindex) \not=\nullsymbol$.
At time stamp $\permuTGD(\edge,\edgeindex)$,
facility~$\connectmapping(\edge,\edgeindex)$ has been opened, and 
edge $\edge'$ has not been connected through its home/work location
$\edge'_{\edgeindex'}$ yet.
This implies that
$\discountfactor \cdot \greedycounter(\edge') < 
\distance(\edge'_{\edgeindex'}, \connectmapping(\edge,\edgeindex))$.
Otherwise,  
the condition of Event~(a) is satisfied, i.e.,
edge $\edge'$ should be connected to
facility~$\connectmapping(\edge,\edgeindex)$ through 
its home/work location
$\edge'_{\edgeindex'}$ weakly before time stamp $\permuTGD(\edge,\edgeindex)$, which
is a contradiction.
Therefore,
\begin{align*}
    \discountfactor \cdot \greedycounter(\edge') 
    < 
\distance(\edge'_{\edgeindex'}, \connectmapping(\edge,\edgeindex))
\overset{}{\leq}
\distance(\edge_\edgeindex, \connectmapping(\edge,\edgeindex))
+
\distance(\edge_\edgeindex, i)
+
\distance(\edge'_{\edgeindex'}, i)
\end{align*}
where the second inequality holds due to the triangle inequality.

Next we prove property~(ii) by contradiction. 
Suppose there exists a location $i\in[n]$
and edge $\edge\in\Edge$
such that 
\begin{align*}
    \sum_{\edge'\in\Edge:\permuTGD(\edge',\config_i(\edge')) \geq \permuTGD(\edge,\config_i(\edge))}
    \plus{
    \discountfactor\cdot 
    \min\left\{
    \greedycounter(\edge),
    \greedycounter(\edge')
    \right\}
    -
    \distance(\edge'_{\config_i(\edge')}, i)
    }
    >
    \costscalar \cdot \opencost_i
\end{align*}
Now consider the algorithm at time stamp $\timeindex \triangleq \permuTGD(\edge,\config_i(\edge)) - \epsilon$ for sufficiently small $\epsilon$.
By definition, none of edges $\edge'\in\Edge$ 
with 
$\permuTGD(\edge',\config_i(\edge'))
\geq \permuTGD(\edge,\config_i(\edge))$
has been connected to any facility through its 
home/work location $\edge'_{\config_i(\edge')}$.
In a slight abuse of notation, 
let $\greedycounter_{\timeindex}(\edge')$
be the value of variable $\greedycounter(\edge')$
in the algorithm at time stamp $t$.
By construction, we know that 
$\greedycounter_{\timeindex}(\edge')
\geq \min\{\timeindex,\greedycounter(\edge')\}$,
and $\timeindex \geq \greedycounter(\edge)$.
Thus, we claim that facility $i$ is opened 
weakly before time stamp $\timeindex$, 
since
the condition of Event~(b) for 
opening facility $i$
is satisfied at time stamp $\timeindex$.
Moreover, there exists at least an edge $\edge'\in\Edge$ 
such that  
$\permuTGD(\edge',\config_i(\edge'))\geq \timeindex$
and $\discountfactor\cdot  
\greedycounter_{\timeindex}(\edge') \geq \distance(\edge'_{\config_i(\edge')}, i)$.
This implies that edge $\edge'$ 
is connected to facility $i$ 
through its $\edge'_{\config_i(\edge')}$ 
weakly before time stamp $\timeindex$,
which is a contradiction.

Finally, property~(iii) is guaranteed by the construction of the algorithm,
and the assumption that $\discountfactor \in[0, 1]$.
\end{proof}


\subsection{Proof of \Cref{lem:lp relaxation}}
\label{apx:lp relaxation}
\lprelaxation*

\begin{proof}
Consider the following integer solution of program~\ref{eq:lp relaxtion}.
For each service configuration~$\starinstance =  (i, \EdgeTilde)$,
set $\allocstar = 1$
if the optimal solution $\OPT$
opens facility $i$ and serves all edges in $\EdgeTilde$
through facility $i$;
and $\allocstar = 0$ otherwise.
It can be verified that this solution is feasible 
and its objective value equals 
the optimal cost $\Cost{\OPT}$.
Therefore, 
$\Cost{\OPT} \geq \Obj{\text{\ref{eq:lp relaxtion}}}$.
\end{proof}

\subsection{Proof of \Cref{lem:dual assignment objective value}}
\label{apx:dual assignment objective value}
\dualassignmentobjectivevalue*
\begin{proof}
By the dual assignment~\eqref{eq:dual assignment}, we have
\begin{align*}
    \sum_{\edge\in\Edge}\duale
    =&
    \sum_{\edge \in \singlecountedgeset}
    \flowe\cdot 
    \left(
    \frac{1+\discountfactor}{\costscalar}
    \cdot \greedycounter(\edge)
    -
    \left(
    \frac{1 + \discountfactor}{\costscalar} - 1
    \right)
    \cdot 
    \distance(\edgehome, \connectmapping(\edge,\Home))
    \right)
    \\
    &\quad +
    \sum_{\edge \in \doublecountedgeset}
    \flowe\cdot 
    \left(
    \frac{1 + \discountfactor}{\costscalar} \cdot \greedycounter(\edge) 
    - 
    \frac{1}{\costscalar}
    \left(
    \distance(\edgehome, \connectmapping(\edge,\Home))
    +
    \distance(\edgework, \connectmapping(\edge,\Work))
    \right)
    +
    \distance(\edgehome, \connectmapping(\edge,\Home))
    \right)
\end{align*}
We analyze the term of each individual on each edge $\edge$ on the right-hand side separately. 
For each individual on edge $\edge \in \singlecountedgeset$,
we have
\begin{align*}
    &~\frac{1+\discountfactor}{\costscalar}
    \cdot \greedycounter(\edge)
    -
    \left(
    \frac{1 + \discountfactor}{\costscalar} - 1
    \right)
    \cdot 
    \distance(\edgehome, \connectmapping(\edge,\Home))
    \\
    =&~
    \frac{1+\discountfactor}{\costscalar}
    \cdot 
    \left(
    \greedycounter(\edge)
    -
    \distance(\edgehome, \connectmapping(\edge,\Home))
    \right)
    +
    \distance(\edgehome, \connectmapping(\edge,\Home))
    \\
    \geq&~
    \frac{1}{\costscalar}
    \cdot 
    \left(
    \greedycounter(\edge)
    -
    \distance(\edgehome, \connectmapping(\edge,\Home))
    \right)
    +
    \distance(\edgehome, \connectmapping(\edge,\Home))
\end{align*}
where the last inequality holds since $\greedycounter(\edge)
    \geq 
    \distance(\edgehome, \connectmapping(\edge,\Home))$,
    which is implied by the construction of the 2-Chance Greedy Algorithm.

For each individual on edge $\edge \in \doublecountedgeset$, we have
\begin{align*}
    &~\frac{1 + \discountfactor}{\costscalar} \cdot \greedycounter(\edge) 
    - 
    \frac{1}{\costscalar}
    \left(
    \distance(\edgehome, \connectmapping(\edge,\Home))
    +
    \distance(\edgework, \connectmapping(\edge,\Work))
    \right)
    +
    \distance(\edgehome, \connectmapping(\edge,\Home))
    \\
    =&~
    \frac{1}{\costscalar}
    \cdot 
    \left(
    (1+\discountfactor)
    \greedycounter(\edge)
    -
    \distance(\edgehome, \connectmapping(\edge,\Home))
    -
    \distance(\edgework, \connectmapping(\edge,\Work))
    \right)
    +
    \distance(\edgehome, \connectmapping(\edge,\Home))
\end{align*}
Combining two pieces together,
we obtain
\begin{align*}
    \sum_{\edge\in\Edge}\duale &\geq 
    \frac{1}{\costscalar}
    \cdot 
    \left(
    \sum_{\edge\in\singlecountedgeset} 
    \flowe \cdot 
    \left(
    \greedycounter(\edge)
    -
    \distance(\edgehome, \connectmapping(\edge,\Home))
    \right)
    +
    \sum_{\edge\in\doublecountedgeset}
    \flowe \cdot  \left(
    (1+\discountfactor)
    \greedycounter(\edge)
    -
    \distance(\edgehome, \connectmapping(\edge,\Home))
    -
    \distance(\edgework, \connectmapping(\edge,\Work))
    \right)
    \right)
    \\
    &\qquad +
    \sum_{\edge\in\Edge} \flowe\cdot \distance(\edgehome, \connectmapping(\edge,\Home))
\end{align*}
To finish the proof, note that the construction (i.e., the condition of Event~(b)) of the 2-Chance Greedy Algorithm implies that 
the first term on the right-hand side is equal to the facility opening cost of solution $\SOL$, and the second term on the right-hand side is at least the connection cost of solution $\SOL$.
\end{proof}

\subsection{Proof of \Cref{lem:dual assignment approx feasible}}
\label{apx:dual assignment approx feasible}

\dualassignmentapproxfeasible*

\begin{proof}
Fix an arbitrary 
\TWOLFLP\ instance, and 
an arbitrary service region $\starinstance = (i, \EdgeTilde)\in\starspace$
such that $\opencost_i+\sum_{\edge\in\EdgeTilde}\distance(\edge, i) 
> 0$.\footnote{Otherwise, 
$\greedycounter(\edge) = 0$
for every edge $\edge\in\EdgeTilde$ and thus the dual constraint is satisfied 
with equality trivially.}
We construct a solution
$\{\costHat\primed,\muHat\primed(\ell),\distanceHat\primed(\ell),
\distanceHatStar\primed(\ell)\}$ for program~\ref{eq:weakly factor-revealing program} as follows:\footnote{Here we use superscript $\dagger$
to denote the solution in program~\ref{eq:weakly factor-revealing program}.}
\begin{align*}
    &\qquad\costHat\primed \gets \normalizefactor \cdot {\opencost_i},
    \\
    \edge\in\EdgeTilde:&
    \qquad
    \muHat\primed(\edgetoellmapping(\edge)) \gets 
    \normalizefactor \cdot {\greedycounter(\edge)},
    \quad 
    \distanceHat\primed(\edgetoellmapping(\edge)) \gets 
    \normalizefactor \cdot 
    {\distance(\edge, i)},
    \\
    \edge\in\EdgeTilde\cap \singlecountedgeset:
    &\qquad
    \distanceHatStar\primed(\edgetoellmapping(\edge)) \gets 
    \distanceHat(\edgehome,\connectmapping(\edge,\Home)),
    \\
    \edge\in\EdgeTilde\cap \doublecountedgeset:
    &\qquad
    \distanceHatStar\primed(\edgetoellmapping(\edge)) \gets 
    \normalizefactor\cdot \distance(\edge_{\config_i(\edge)}, \connectmapping(\edge, \config_i(\edge)))
\end{align*}
where $\normalizefactor = \frac{1}{\opencost_i+\sum_{\edge\in\EdgeTilde}\distance(\edge, i)}$ is the normalization factor which guarantees the feasibility of 
constraint~\WFRCtotalcost.
Moreover, \Cref{lem:structural lemma} and the solution construction
implies the feasibility of constraints~\WFRCtriangleineq, \WFRCcontribution, and
\WFRCconnectcost.

Finally, we argue that $\frac{\sum_{\edge\in\EdgeTilde} \duale}{\coststar}$ is upperbounded by the objective value of the constructed solution in program~\ref{eq:weakly factor-revealing program}. It is sufficient to show that for every edge $\edge\in \EdgeTilde$, 
\begin{align*}
    \frac{\duale}{\coststar} \leq 
     \frac{1 + \discountfactor}{\costscalar}
     \cdot \muHat\primed(\edgetoellmapping(\edge)) - 
    \left(
    \frac{1 + \discountfactor}{\costscalar} - 1
    \right)
    \cdot \distanceHatStar\primed(\edgetoellmapping(\edge))
\end{align*}
We show this inequality for two different cases separately. 

\smallskip
For each edge $\edge \in \EdgeTilde\cap \singlecountedgeset$, we have
\begin{align*}
     \frac{\duale}{\coststar}
     &\overset{(a)}{=}
     \frac{\frac{1+\discountfactor}{\costscalar}
    \cdot \greedycounter(\edge)
    -
    \left(
    \frac{1 + \discountfactor}{\costscalar} - 1
    \right)
    \cdot 
    \distance(\edgehome, \connectmapping(\edge,\Home))
    }{{}\opencost_i+\sum_{\edge'\in\EdgeTilde}
\distance(\edge', i)}
    \\
     &\overset{(b)}{=}
     \frac{\frac{1+\discountfactor}{\costscalar}
    \cdot \muHat\primed(\edgetoellmapping(\edge))
    -
    \left(
    \frac{1 + \discountfactor}{\costscalar} - 1
    \right)
    \cdot 
    \distanceHatStar\primed(\edgetoellmapping(\edge))
    }{{}\opencost\primed+\sum_{\ell\in[m]}
\distance\primed(\ell)}
    \\
    &\overset{(c)}{\leq}
     \frac{1 + \discountfactor}{\costscalar}
     \cdot \muHat\primed(\edgetoellmapping(\edge)) - 
    \left(
    \frac{1 + \discountfactor}{\costscalar} - 1
    \right)
    \cdot \distanceHatStar\primed(\edgetoellmapping(\edge))
\end{align*}
where equality~(a) holds due to $\flowe = 1$ and the dual assignment \eqref{eq:dual assignment};
equality~(b) holds due to the solution construction in program~\ref{eq:weakly factor-revealing program};
and inequality~(c) holds due to constraint~\WFRCtotalcost\ in program~\ref{eq:weakly factor-revealing program}.

\smallskip
For each edge $\edge \in \EdgeTilde\cap \doublecountedgeset$, we have
\begin{align*}
     \frac{\duale}{\coststar}
     &\overset{(a)}{=}
     \frac{    
     \frac{1 + \discountfactor}{\costscalar} \cdot \greedycounter(\edge) 
    - 
    \frac{1}{\costscalar}
    \left(
    \distance(\edgehome, \connectmapping(\edge,\Home))
    +
    \distance(\edgework, \connectmapping(\edge,\Work))
    \right)
    +
    \distance(\edgehome, \connectmapping(\edge,\Home))
    }{\opencost_i+\sum_{\edge'\in\EdgeTilde}
\distance(\edge', i)}
    \\
     &\overset{}{=}
     \frac{
     \frac{1 + \discountfactor}{\costscalar} \cdot \greedycounter(\edge) 
    - 
    \frac{1}{\costscalar}
    \distance(\edgework, \connectmapping(\edge,\Work))
    +
    \left(1 - \frac{1}{\costscalar}\right)
    \distance(\edgehome, \connectmapping(\edge,\Home))
     }{\opencost_i+\sum_{\edge'\in\EdgeTilde}
\distance(\edge', i)}
    \\
     &\overset{(b)}{\leq}
     \frac{
     \frac{1 + \discountfactor}{\costscalar} \cdot \greedycounter(\edge) 
    - 
    \frac{1}{\costscalar}
    \distance(\edgework, \connectmapping(\edge,\Work))
    +
    \left(1 - \frac{1}{\costscalar}\right)
    \distance(\edgework, \connectmapping(\edge,\Work))
     }{\opencost_i+\sum_{\edge'\in\EdgeTilde}
\distance(\edge', i)}
    \\
     &\overset{(c)}{\leq}
     \frac{\left(\frac{1+\discountfactor}{\costscalar}
    \cdot \muHat\primed(\edgetoellmapping(\edge))
    -
    \left(
    \frac{1 + \discountfactor}{\costscalar} - 1
    \right)
    \cdot 
    \distanceHatStar\primed(\edgetoellmapping(\edge))
    \right)}{\opencost\primed+\sum_{\ell\in[m]}
\distance\primed(\ell)}
    \\
    &\overset{(d)}{\leq}
     \frac{1 + \discountfactor}{\costscalar}
     \cdot \muHat\primed(\edgetoellmapping(\edge)) - 
    \left(
    \frac{1 + \discountfactor}{\costscalar} - 1
    \right)
    \cdot \distanceHatStar\primed(\edgetoellmapping(\edge))
\end{align*}
where equality~(a) holds due to $\flowe = 1$ and the dual assignment \eqref{eq:dual assignment};
inequality~(b) holds due to the assumption that $\distance(\edgehome, \connectmapping(\edge,\Home)) \leq 
 \distance(\edgework, \connectmapping(\edge,\Work))$ introduced in step 2 and $(1 - \sfrac{1}{\costscalar}) \geq 0$;
equality~(c) holds due to the solution construction in program~\ref{eq:weakly factor-revealing program}, the assumption that $\distance(\edgehome, \connectmapping(\edge,\Home)) \leq 
 \distance(\edgework, \connectmapping(\edge,\Work))$ introduced in step 2 and $\sfrac{2}{\costscalar} - 1 
 \geq \sfrac{(1 + \discountfactor)}{\costscalar} - 1 \geq 0$;
and inequality~(d) holds due to constraint~\WFRCtotalcost\ in program~\ref{eq:weakly factor-revealing program}.
\end{proof}

\subsection{Proof of \Cref{lem:from weakly to strongly factor revealing program}}
\label{apx:from weakly to strongly factor revealing program}

\weaklytostronglyFRP*

Before the proof of \Cref{lem:from weakly to strongly factor revealing program},
we first present a technical lemma as follows.

\begin{lemma}
\label{lem:distance upperbound}
For any $\discountfactor\in[0,1],~\mHat\in\naturals,~
\permu:[\mHat]\rightarrow\positivereals$
and any feasible solution $(\costHat, \{\muHat(\ell), \distanceHat(\ell),\distanceHatStar(\ell)\})$
of program~\ref{eq:weakly factor-revealing program},
there exists $\mHat\primed \in\naturals,~ \permu\primed:[\mHat\primed]\rightarrow\positivereals$
and 
a feasible solution $(\costHat\primed, \{\muHat\primed(\ell), \distanceHat\primed(\ell),\distanceHatStar\primed(\ell)\})$
of program~\WFRP{\mHat\primed,\permu\primed,\discountfactor}
such that $\distanceHat\primed(\ell)\leq \muHat\primed(\ell)$
for every $\ell\in[\mHat\primed]$ and 
has weakly higher objective value.
\end{lemma}
\begin{proof}
Let $\ell^*$ be the index with largest $\muHat(\ell)$,
i.e., $\ell^*= \argmax_{\ell\in[\mHat]} \muHat(\ell)$.
If $\distanceHat(\ell^*) > \muHat(\ell^*)$,
then we can update $\muHat(\ell^*) \gets \distanceHat(\ell^*)$.
It can be verified that the updated solution is
still feasible in 
program~\ref{eq:weakly factor-revealing program}
(in particular, constraints \WFRCtriangleineq 
and \WFRCcontribution\
are still satisfied), and has 
higher objective value.

Now, we partition index set $[\mHat]$ into two subsets:
$\Delta = \{\ell\in[\muHat]: \distanceHat(\ell) 
\leq \muHat(\ell)\}$ and 
$\bar \Delta = [\mHat]\backslash \Delta$.
As we discussed in the previous paragraph,
$\ell^* \in \Delta$.
Let $g$ be an arbitrary bijection from $\Delta$
to $[|\Delta|]$.
Consider the following construction of parameters $\mHat\primed, \permu\primed$:
\begin{align*}
    &\qquad \mHat\primed \gets |\Delta|,
    \\
    \ell\in\Delta:
    &\qquad \permu\primed(g(\ell)) \gets \permu(\ell),
\intertext{and solution $(\costHat\primed, \{\muHat\primed(\ell), \distanceHat\primed(\ell),\distanceHatStar\primed(\ell)\})$:}
    &\qquad \costHat\primed \gets \costHat, \\
    &\qquad 
    \muHat\primed(g(\ell^*)) \gets 
    \muHat(\ell^*) + \sum_{\ell\in\bar\Delta}\distanceHat(\ell),
    \quad
    \distanceHat\primed(g(\ell^*)) \gets 
    \distanceHat(\ell^*) + \sum_{\ell\in\bar\Delta}\distanceHat(\ell),
    \quad
    \distanceHatStar\primed(g(\ell^*)) \gets 
    \distanceHatStar(\ell^*) 
    \\
    \ell\in\Delta\backslash\{\ell^*\}:
    &\qquad
    \muHat\primed(g(\ell)) \gets 
    \muHat(\ell),
    \quad
    \distanceHat\primed(g(\ell)) \gets 
    \distanceHat(\ell),
    \quad
    \distanceHatStar\primed(g(\ell)) \gets 
    \distanceHatStar(\ell) 
\end{align*}
It is straightforward to verify that 
the constructed solution is feasible in 
program~\WFRP{\mHat\primed,\permu\primed,\discountfactor},
and has weakly higher objective value.
\end{proof}

Now we are ready to present the proof of \Cref{lem:from weakly to strongly factor revealing program}.
\begin{proof}[Proof of \Cref{lem:from weakly to strongly factor revealing program}]
Fix an arbitrary $\discountfactor\in[0, 1],~
\nHat\in\naturals,~
\mHat\in\naturals$ and $\permu:[\mHat]\rightarrow$.
Consider an arbitrary feasible solution feasible solution $(\costHat, \{\muHat(\ell), \distanceHat(\ell),\distanceHatStar(\ell)\})$
of program~\ref{eq:weakly factor-revealing program}
such that $\distanceHat(\ell) \leq \muHat(\ell)$ for each $\ell\in[\mHat]$.
It is sufficient for us to construct a feasible solution $(\costHat,\{\qHat(a,b ),
\muHat(a,b), \distanceHat(a,b),\distanceHatStar(a,b)\})$
of program~\ref{eq:strongly factor-revealing quadratic program}
with weakly higher objective value. 
Our argument proceeds in four steps as we sketched 
in \Cref{sec:strongly factor-revealing program 2-LFLP}.

\paragraph{Step 1- identifying pivotal index subset $L$ with monotone $\muHat(\ell)$.}
Consider a sequence of $k$ indexes 
$\ell_1 < \ell_2 < \dots < \ell_k$ 
for some $k\in\naturals$
such that 
\begin{align*}
\forall \ell\in[\mHat]:
&\qquad \permu(\ell_1) \leq \permu(\ell)
\\
\forall a\in[k - 1]:
&\qquad \permu(\ell_a) < \permu(\ell_{a + 1})
\\
\forall a\in[k-1]:
&\qquad 
\muHat(\ell_a) < \muHat(\ell_{a + 1})
        \\
     \forall a \in [k],~
     \forall\ell \in \{\ell'\in[\mHat]:
     \permu(\ell_a) \leq 
     \permu(\ell') <
     \permu(\ell_{a + 1})
     \}: 
    &\qquad \muHat(\ell) \leq 
    \muHat(\ell_a)
\end{align*}
where $\permu(\ell_{k + 1}) = \infty$.
We define pivotal index subset $L \triangleq 
\{\ell_a\}_{a\in[k]}$.
To see the existence of 
pivotal index subset~$L$,
consider the following iterative procedure which generates $L$:
(i) initialize $L =\emptyset$,
(ii) add index $\ell_k = \argmax_{\ell\in[\mHat]}\muHat(\ell)$
into $L$,
(iii) add index $\ell_{k - 1} = \argmax_{\ell\in[\mHat]: \permu(\ell) < \permu(\ell_k)}
\muHat(\ell)$ into $L$ (break tie in favor of smaller $\permu(\ell)$),
(iv) add index $\ell_{k - 2} = \argmax_{\ell\in[\mHat]: \permu(\ell) < \permu(\ell_{k - 1})}
\muHat(\ell)$ into $L$ (break tie in favor of smaller $\permu(\ell)$),
and (v) so on so forth.

\paragraph{Step 2- partitioning index set $[\mHat]$ 
based on pivotal index subset $L$.}

For each $a\in[k]$, $b\in[a]$, define set 
$L(a, b)$
as follows:
\begin{align*}
    L(a,b) \triangleq \{
\ell \in [\muHat]:
\permu(\ell_a) \leq \permu(\ell) < \permu(\ell_{a + 1})
\land
\muHat(\ell_{b - 1}) < \muHat(\ell) \leq \muHat(\ell_b)
\}
\end{align*}
where $\muHat(\ell_0) = 0$.
By definition of pivot index subset $L$, 
$\{L(a, b)\}_{b\in[a]}$ is a partition of
$\{\ell\in[\mHat]: \permu(\ell_a) \leq \permu(\ell) < \permu(\ell_{a + 1})\}$
for each $a\in[k]$,
and $\{L(a, b)\}_{a\in[k],b\in[a]}$
is a partition of index set $[\mHat]$.

\paragraph{Step 3- batching variables based on partitions
$\{L(a, b)\}_{a\in[k],b\in[a]}$}
In this step, we construct 
a solution 
$(\costHat, \{\qHat(a, b), \muHat(a, b),\distanceHat(a, b),
\distanceHatStar(a, b)\}_{a\in[k],b\in[a]})$
which satisfies 
constraints \SFRCmonotonicity, \SFRCtriangleineq, 
\SFRCdistance, \SFRCconnectcost, \SFRCtotalcost\
of program~\SFRP{k,\discountfactor,\costscalar}
as well as constraint~\SFRCcontributionstronger\ (parameterized by $k,\discountfactor,\costscalar$)
defined as follows
\begin{align*}
    \SFRCcontributionstronger \qquad
    &\sum_{a' \in [a:k]}
    \sum_{b' \in [a]}
    \qHat(a',b')\cdot \plus{\discountfactor\cdot \muHat(a', b') - \distanceHat(a',b')}
    \\
    &\qquad 
    +\sum_{a' \in [a:k]}
    \sum_{b' \in [a + 1:a']}
    \qHat(a',b')\cdot \plus{\discountfactor\cdot \muHat(a, a) - \distanceHat(a',b')}
    \leq \costscalar\cdot \costHat
    \qquad 
    a\in[k]
\end{align*}
It is straightforward to see that 
constraint~\SFRCcontribution\ in program~\SFRP{k,\discountfactor,\costscalar}
is a relaxation of \SFRCcontributionstronger.

Briefly speaking, we set $\muHat(a, b)$
(resp.\ $\distanceHat(a, b)$, $\distanceHatStar(a, b)$) 
for each 
two-dimensional index 
$(a, b)$
as the average of $\muHat(\ell)$
(resp.\ $\distanceHat(\ell)$, $\distanceHatStar(\ell)$) 
over single-dimensional index $\ell\in L(a, b)$.
The formal construction is as follows,
\begin{align*}
    & 
    \costHat\primed\gets \costHat
    \\
    a\in[k], b\in[a],
    |L(a,b)|\geq 1:
    \qquad 
    &\muHat\primed(a, b) \gets 
    \sum_{\ell\in L(a, b)} 
    \frac{
    \muHat(\ell)
    }{
    |L(a,b)|
    },
    \quad
    \distanceHat\primed(a, b) \gets 
    \sum_{\ell\in L(a, b)} 
    \frac{ \distanceHat(\ell)
    }{
    |L(a,b)|
    },
    \\
    &\distanceHatStar\primed(a, b) \gets 
    \sum_{\ell\in L(a, b)} 
    \frac{\distanceHatStar(\ell)
    }{
    |L(a,b)|
    },
    \quad
    \qHat\primed(a, b) \gets |L(a,b)|
    \\
    a\in[k], b\in[a],
    |L(a,b)| = 0:
    \qquad 
    &\muHat\primed(a, b) \gets 
    \muHat\primed(a - 1, b),
    \quad
    \distanceHat\primed(a, b) \gets \distanceHat\primed(a - 1, b),
    \\
    &\distanceHatStar\primed(a, b) \gets \distanceHatStar\primed(a - 1, b),
    \quad
    \qHat\primed(a, b) \gets |L(a,b)|
\end{align*}
Note that $|L(a, a)| \geq 1$ for each $a\in[k]$,
since $\ell_a\in L(a, a)$ by definition.
Therefore, the above construction (in particular, the part to handle $|L(a, b)| = 0$) is well-defined. 

It is easy to verify that the constructed solution has the same objective value 
as the original solution.
Now, we verify the feasibility of 
constraints \SFRCmonotonicity, \SFRCtriangleineq, \SFRCcontributionstronger,
\SFRCdistance, \SFRCconnectcost, \SFRCtotalcost, respectively.
\begin{enumerate}
    \item[\SFRCmonotonicity] It is implied by the construction of solution
    and the definition of partitions $\{L(a, b)\}$.
    \item[\SFRCtriangleineq] It is implied by \WFRCtriangleineq\ in
    program~\ref{eq:weakly factor-revealing program}
    and the construction of solution.
    Specifically,
    for any $a, a'\in[k], b\in[a],b'\in[a']$, 
    $a < a'$, 
    suppose\footnote{The arguments for other cases
    (i.e., 
    $|L(a, b)| = 0$ or $|L(a', b')| = 0$)
    are similar 
    and thus omitted here.} 
    $|L(a, b)| \geq 1$, $|L(a', b')| \geq 1$, then
    \begin{align*}
        \discountfactor\cdot \muHat\primed(a', b') 
        &\overset{(a)}{=} 
        \sum_{\ell'\in L(a', b')} \frac{\discountfactor\cdot \muHat(\ell')}{|L(a', b')|}
        \\
        &\overset{(b)}{\leq}
        \sum_{\ell'\in L(a', b')} \left(\frac{1}{|L(a', b')|}
        \sum_{\ell\in L(a, b)} \frac{1}{|L(a, b)|}
        \left(
        \distanceHatStar(\ell) + \distanceHat(\ell) +
        \distanceHat(\ell')
        \right)
        \right)
        \\
        &=
        \sum_{\ell\in L(a, b)} \frac{
        \distanceHatStar(\ell)}{|L(a, b)|}
        +
        \sum_{\ell\in L(a, b)} \frac{
        \distanceHat(\ell)}{|L(a, b)|}
        +
        \sum_{\ell'\in L(a', b')} \frac{\distanceHat(\ell')}{|L(a', b')|}
        \\
        &\overset{(c)}{=}
        \distanceHatStar\primed(a, b) + \distanceHat\primed(a, b) + \distanceHat\primed(a', b')
    \end{align*}
    where equalities~(a) (c) hold due to the solution construction;
    and inequality~(b) holds due to constraint \WFRCtriangleineq\ and 
    the construction of $L(a, b), L(a', b')$, which guarantees
    $\permu(\ell) < \permu(\ell')$ for every $\ell\in L(a, b)$, $\ell'\in L(a', b')$.
    \item[\SFRCcontributionstronger]
    It is implied by \WFRCcontribution\ 
    and the solution construction.
    Specifically, 
    for any $a\in[k]$,  
    \begin{align*}
        \costscalar\cdot \costHat\primed &\overset{(a)}{=} 
        \costscalar\cdot \costHat
        \\
        &\overset{(b)}{\geq} 
         \sum_{\ell'\in[\mHat]: \permu(\ell') \geq \permu(\ell_a)}
     \plus{\min\{\discountfactor\cdot \muHat(\ell_a),\muHat(\ell')\} - 
     \distanceHat(\ell')}
     \\
    &\overset{(c)}{=} 
     \sum_{a'\in[a:k]}\sum_{b'\in[a']}
     \sum_{\ell'\in L(a', b')}
     \plus{\min\{\discountfactor\cdot \muHat(\ell_a),\muHat(\ell')\} - 
     \distanceHat(\ell')}
     \\
     &\overset{(d)}{=}
     \sum_{a'\in[a:k]}
     \left(
     \sum_{b'\in[a]}
     \sum_{\ell'\in L(a', b')}
     \plus{\discountfactor\cdot \muHat(\ell') - 
     \distanceHat(\ell')}
     +
     \sum_{b'\in[a+1:a']}
     \sum_{\ell'\in L(a', b')}
     \plus{\discountfactor\cdot \muHat(\ell_a) - 
     \distanceHat(\ell')}
     \right)
     \\
    & \overset{(e)}{\geq}
    \sum_{a'\in[a:k]}
     \sum_{b'\in[a]}
     |L(a', b')|
     \cdot 
     \plus{
     \sum_{\ell'\in L(a', b')}
     \frac{\discountfactor\cdot \muHat(\ell') - \distanceHat(\ell')}
     {|L(a', b')|}
     }
     \\
     &\qquad\qquad 
    +\sum_{a'\in[a:k]}
     \sum_{b'\in[a+1:a']}
     |L(a', b')|
     \cdot 
     \plus{
     \sum_{\ell'\in L(a', b')}
     \frac{\discountfactor\cdot \muHat(\ell_a) - \distanceHat(\ell')}
     {|L(a', b')|}
     }
     \\
     &\overset{(f)}{=}
    \sum_{a'\in[a:k]}
     \sum_{b'\in[a]}
     \qHat\primed(a', b')
     \cdot 
     \plus{
     \discountfactor\cdot \muHat\primed(a', b') - \distanceHat\primed(a', b')
     }
     \\
     &\qquad\qquad 
    +\sum_{a'\in[a:k]}
     \sum_{b'\in[a+1:a']}
     \qHat\primed(a', b')
     \cdot 
     \plus{
     \discountfactor\cdot \muHat(\ell_a) - \distanceHat\primed(a', b')
     }
     \\
     &\overset{(g)}{\geq}
    \sum_{a'\in[a:k]}
     \sum_{b'\in[a]}
     \qHat\primed(a', b')
     \cdot 
     \plus{
     \discountfactor\cdot \muHat\primed(a', b') - \distanceHat\primed(a', b')
     }
     \\
     &\qquad\qquad 
    +\sum_{a'\in[a:k]}
     \sum_{b'\in[a+1:a']}
     \qHat\primed(a', b')
     \cdot 
     \plus{
     \discountfactor\cdot \muHat\primed(a, a) - \distanceHat\primed(a', b')
     }
    \end{align*}
    where inequalities~(a) (f) hold due to the solution construction;
    inequality~(b) holds due to \WFRCcontribution\ at $\ell = \ell_a$;
    equality~(c) holds due to the construction of sets
    $\{L(a', b')\}_{a'\in[a:n],b'\in[a']}$,
    which guarantees that $\{L(a', b')\}_{a'\in[a:n],b'\in[a']}$
    is a partition of $\{\ell'\in[\mHat]:\permu(\ell')\geq \permu(\ell_a)\}$;
    equality~(d) holds due to the construction of sets
    $\{L(a', b')\}$,
    which guarantees that 
    for every $a' \in [a:n]$
    and every $\ell'\in L(a', b')$,
    $\muHat(\ell') \leq \muHat(\ell_a)$ 
    if and only if $b' \in [a]$;
    inequality~(e) holds due to the convexity of $\plus{\cdot}$;
    and inequality~(g) holds since 
    the solution construction which guarantees 
    $\muHat(\ell_a) \geq \muHat\primed(a, a)$.

    \item[\SFRCdistance]
    It is implied by our initial assumption and the solution construction.
    
    \item[\SFRCconnectcost] 
    It is implied by \WFRCconnectcost\  and the solution construction.
    
    \item[\SFRCtotalcost] 
    It is implied by \WFRCtotalcost\  and the solution construction.
\end{enumerate}

\paragraph{Step 4- converting into a feasible solution of program~\ref{eq:strongly factor-revealing quadratic program}.}
In the last step, 
we further convert the solution obtained in step 3 for 
program~\SFRP{k,\discountfactor,\costscalar}
into a feasible solution for program~\ref{eq:strongly factor-revealing quadratic program}
for an arbitrary $\nHat\in\naturals$.

Note that by simultaneously
scaling $\{\qHat(a, b)\}$ down,
and $\{\muHat(a, b), \distanceHat(a, b), 
\distanceHatStar(a, b)\}$ up
by a 
multiplicative factor $\frac{k}{\nHat}$,\footnote{Namely, 
$\qHat(a, b) \gets \qHat(a, b) \cdot \frac{k}{\nHat}$,
and 
$\muHat(a, b) \gets \muHat(a, b) \cdot \frac{\nHat}{k}$,
$\distanceHat(a, b) \gets 
\distanceHat(a, b) \cdot \frac{\nHat}{k}$,
$\distanceHatStar(a, b) \gets 
\distanceHatStar(a, b) \cdot \frac{\nHat}{k}$.}
the solution still satisfies  
constraints \SFRCmonotonicity, \SFRCtriangleineq, \SFRCcontributionstronger,
\SFRCdistance, \SFRCconnectcost, \SFRCtotalcost.
Furthermore,
$\sum_{a\in[k]}\sum_{b\in[a]}\qHat(a, b) = k$,
and the objective value remains unchanged. 

Next, we first \emph{decompose} the current solution 
in program~\SFRP{k,\discountfactor,\costscalar}
into a solution 
in program~\SFRP{k\primed,\discountfactor,\costscalar}
for some larger $k\primed \in\naturals, k\primed\geq  k$ 
such that $\sum_{a\in[b:k\primed]} \qHat(a, b) \leq \epsilon$
for every $b\in[k\primed]$ and arbitrary small $\epsilon$.
Then, we \emph{batch} this solution in program~\SFRP{k\primed,\discountfactor,\costscalar}
into a feasible solution in program~\ref{eq:strongly factor-revealing quadratic program}
as desired. 

\paragraph{Step 4.a- decomposition.}

\begin{figure}
    \centering
    \subfloat[
    before stretch-modification]{\begin{tikzpicture}[scale=0.55, transform shape]
\begin{axis}[
axis line style=gray,
axis lines=middle,
        ytick = {4.7},
        yticklabels = {$b^*$},
        xtick = \empty,
        xticklabels = \empty,
x label style={at={(axis description cs:0.95,-0.01)},anchor=north},
y label style={at={(axis description cs:0.03,0.95)},anchor=south},
xlabel = {$a$},
ylabel = {$b$},
label style={font=\LARGE},
xmin=0,xmax=11.5,ymin=0,ymax=11.5,
width=0.9\textwidth,
height=0.9\textwidth,
samples=50]

\addplot[dotted, white!10!black] coordinates {
(0.05, 4.7) (4.35, 4.7)};

\addplot[fill=white!90!black] coordinates {
(0.4, 0.4)(0.4, 1.0)(1.0, 1.0)(1.0, 0.4)(0.4, 0.4)
};
\addplot[fill=white!90!black] coordinates {
(1.4, 0.4)(1.4, 1.0)(2.0, 1.0)(2.0, 0.4)(1.4, 0.4)
};
\addplot[fill=white!90!black] coordinates {
(1.4, 1.4)(1.4, 2.0)(2.0, 2.0)(2.0, 1.4)(1.4, 1.4)
};
\addplot[fill=white!90!black] coordinates {
(2.4, 0.4)(2.4, 1.0)(3.0, 1.0)(3.0, 0.4)(2.4, 0.4)
};
\addplot[fill=white!90!black] coordinates {
(2.4, 1.4)(2.4, 2.0)(3.0, 2.0)(3.0, 1.4)(2.4, 1.4)
};
\addplot[fill=white!90!black] coordinates {
(2.4, 2.4)(2.4, 3.0)(3.0, 3.0)(3.0, 2.4)(2.4, 2.4)
};
\addplot[fill=white!90!black] coordinates {
(3.4, 0.4)(3.4, 1.0)(4.0, 1.0)(4.0, 0.4)(3.4, 0.4)
};
\addplot[fill=white!90!black] coordinates {
(3.4, 1.4)(3.4, 2.0)(4.0, 2.0)(4.0, 1.4)(3.4, 1.4)
};
\addplot[fill=white!90!black] coordinates {
(3.4, 2.4)(3.4, 3.0)(4.0, 3.0)(4.0, 2.4)(3.4, 2.4)
};
\addplot[fill=white!90!black] coordinates {
(3.4, 3.4)(3.4, 4.0)(4.0, 4.0)(4.0, 3.4)(3.4, 3.4)
};
\addplot[fill=white!90!black] coordinates {
(4.4, 0.4)(4.4, 1.0)(5.0, 1.0)(5.0, 0.4)(4.4, 0.4)
};
\addplot[fill=white!90!black] coordinates {
(4.4, 1.4)(4.4, 2.0)(5.0, 2.0)(5.0, 1.4)(4.4, 1.4)
};
\addplot[fill=white!90!black] coordinates {
(4.4, 2.4)(4.4, 3.0)(5.0, 3.0)(5.0, 2.4)(4.4, 2.4)
};
\addplot[fill=white!90!black] coordinates {
(4.4, 3.4)(4.4, 4.0)(5.0, 4.0)(5.0, 3.4)(4.4, 3.4)
};
\addplot[fill=black] coordinates {
(4.4, 4.4)(4.4, 5.0)(5.0, 5.0)(5.0, 4.4)(4.4, 4.4)
};
\addplot[fill=white!90!black] coordinates {
(5.4, 0.4)(5.4, 1.0)(6.0, 1.0)(6.0, 0.4)(5.4, 0.4)
};
\addplot[fill=white!90!black] coordinates {
(5.4, 1.4)(5.4, 2.0)(6.0, 2.0)(6.0, 1.4)(5.4, 1.4)
};
\addplot[fill=white!90!black] coordinates {
(5.4, 2.4)(5.4, 3.0)(6.0, 3.0)(6.0, 2.4)(5.4, 2.4)
};
\addplot[fill=white!90!black] coordinates {
(5.4, 3.4)(5.4, 4.0)(6.0, 4.0)(6.0, 3.4)(5.4, 3.4)
};
\addplot[fill=black] coordinates {
(5.4, 4.4)(5.4, 5.0)(6.0, 5.0)(6.0, 4.4)(5.4, 4.4)
};
\addplot[fill=white!90!black] coordinates {
(5.4, 5.4)(5.4, 6.0)(6.0, 6.0)(6.0, 5.4)(5.4, 5.4)
};
\addplot[fill=white!90!black] coordinates {
(6.4, 0.4)(6.4, 1.0)(7.0, 1.0)(7.0, 0.4)(6.4, 0.4)
};
\addplot[fill=white!90!black] coordinates {
(6.4, 1.4)(6.4, 2.0)(7.0, 2.0)(7.0, 1.4)(6.4, 1.4)
};
\addplot[fill=white!90!black] coordinates {
(6.4, 2.4)(6.4, 3.0)(7.0, 3.0)(7.0, 2.4)(6.4, 2.4)
};
\addplot[fill=white!90!black] coordinates {
(6.4, 3.4)(6.4, 4.0)(7.0, 4.0)(7.0, 3.4)(6.4, 3.4)
};
\addplot[fill=black] coordinates {
(6.4, 4.4)(6.4, 5.0)(7.0, 5.0)(7.0, 4.4)(6.4, 4.4)
};
\addplot[fill=white!90!black] coordinates {
(6.4, 5.4)(6.4, 6.0)(7.0, 6.0)(7.0, 5.4)(6.4, 5.4)
};
\addplot[fill=white!90!black] coordinates {
(6.4, 6.4)(6.4, 7.0)(7.0, 7.0)(7.0, 6.4)(6.4, 6.4)
};
\addplot[fill=white!90!black] coordinates {
(7.4, 0.4)(7.4, 1.0)(8.0, 1.0)(8.0, 0.4)(7.4, 0.4)
};
\addplot[fill=white!90!black] coordinates {
(7.4, 1.4)(7.4, 2.0)(8.0, 2.0)(8.0, 1.4)(7.4, 1.4)
};
\addplot[fill=white!90!black] coordinates {
(7.4, 2.4)(7.4, 3.0)(8.0, 3.0)(8.0, 2.4)(7.4, 2.4)
};
\addplot[fill=white!90!black] coordinates {
(7.4, 3.4)(7.4, 4.0)(8.0, 4.0)(8.0, 3.4)(7.4, 3.4)
};
\addplot[fill=black] coordinates {
(7.4, 4.4)(7.4, 5.0)(8.0, 5.0)(8.0, 4.4)(7.4, 4.4)
};
\addplot[fill=white!90!black] coordinates {
(7.4, 5.4)(7.4, 6.0)(8.0, 6.0)(8.0, 5.4)(7.4, 5.4)
};
\addplot[fill=white!90!black] coordinates {
(7.4, 6.4)(7.4, 7.0)(8.0, 7.0)(8.0, 6.4)(7.4, 6.4)
};
\addplot[fill=white!90!black] coordinates {
(7.4, 7.4)(7.4, 8.0)(8.0, 8.0)(8.0, 7.4)(7.4, 7.4)
};
\addplot[fill=white!90!black] coordinates {
(8.4, 0.4)(8.4, 1.0)(9.0, 1.0)(9.0, 0.4)(8.4, 0.4)
};
\addplot[fill=white!90!black] coordinates {
(8.4, 1.4)(8.4, 2.0)(9.0, 2.0)(9.0, 1.4)(8.4, 1.4)
};
\addplot[fill=white!90!black] coordinates {
(8.4, 2.4)(8.4, 3.0)(9.0, 3.0)(9.0, 2.4)(8.4, 2.4)
};
\addplot[fill=white!90!black] coordinates {
(8.4, 3.4)(8.4, 4.0)(9.0, 4.0)(9.0, 3.4)(8.4, 3.4)
};
\addplot[fill=black] coordinates {
(8.4, 4.4)(8.4, 5.0)(9.0, 5.0)(9.0, 4.4)(8.4, 4.4)
};
\addplot[fill=white!90!black] coordinates {
(8.4, 5.4)(8.4, 6.0)(9.0, 6.0)(9.0, 5.4)(8.4, 5.4)
};
\addplot[fill=white!90!black] coordinates {
(8.4, 6.4)(8.4, 7.0)(9.0, 7.0)(9.0, 6.4)(8.4, 6.4)
};
\addplot[fill=white!90!black] coordinates {
(8.4, 7.4)(8.4, 8.0)(9.0, 8.0)(9.0, 7.4)(8.4, 7.4)
};
\addplot[fill=white!90!black] coordinates {
(8.4, 8.4)(8.4, 9.0)(9.0, 9.0)(9.0, 8.4)(8.4, 8.4)
};

\end{axis}

\end{tikzpicture}}
    ~~~~
    \subfloat[
    after stretch-modification]{\begin{tikzpicture}[scale=0.55, transform shape]
\begin{axis}[
axis line style=gray,
axis lines=middle,
        ytick = {4.7},
        yticklabels = {$b^*$},
        xtick = \empty,
        xticklabels = \empty,
x label style={at={(axis description cs:0.95,-0.01)},anchor=north},
y label style={at={(axis description cs:0.03,0.95)},anchor=south},
xlabel = {$a$},
ylabel = {$b$},
label style={font=\LARGE},
xmin=0,xmax=11.5,ymin=0,ymax=11.5,
width=0.9\textwidth,
height=0.9\textwidth,
samples=50]

\addplot[dotted, white!10!black] coordinates {
(0.05, 4.7) (4.35, 4.7)};

\addplot[fill=white!90!black] coordinates {
(0.4, 0.4)(0.4, 1.0)(1.0, 1.0)(1.0, 0.4)(0.4, 0.4)
};
\addplot[fill=white!90!black] coordinates {
(1.4, 0.4)(1.4, 1.0)(2.0, 1.0)(2.0, 0.4)(1.4, 0.4)
};
\addplot[fill=white!90!black] coordinates {
(1.4, 1.4)(1.4, 2.0)(2.0, 2.0)(2.0, 1.4)(1.4, 1.4)
};
\addplot[fill=white!90!black] coordinates {
(2.4, 0.4)(2.4, 1.0)(3.0, 1.0)(3.0, 0.4)(2.4, 0.4)
};
\addplot[fill=white!90!black] coordinates {
(2.4, 1.4)(2.4, 2.0)(3.0, 2.0)(3.0, 1.4)(2.4, 1.4)
};
\addplot[fill=white!90!black] coordinates {
(2.4, 2.4)(2.4, 3.0)(3.0, 3.0)(3.0, 2.4)(2.4, 2.4)
};
\addplot[fill=white!90!black] coordinates {
(3.4, 0.4)(3.4, 1.0)(4.0, 1.0)(4.0, 0.4)(3.4, 0.4)
};
\addplot[fill=white!90!black] coordinates {
(3.4, 1.4)(3.4, 2.0)(4.0, 2.0)(4.0, 1.4)(3.4, 1.4)
};
\addplot[fill=white!90!black] coordinates {
(3.4, 2.4)(3.4, 3.0)(4.0, 3.0)(4.0, 2.4)(3.4, 2.4)
};
\addplot[fill=white!90!black] coordinates {
(3.4, 3.4)(3.4, 4.0)(4.0, 4.0)(4.0, 3.4)(3.4, 3.4)
};
\addplot[fill=white!100!black] coordinates {
(4.4, 0.4)(4.4, 1.0)(5.0, 1.0)(5.0, 0.4)(4.4, 0.4)
};
\addplot[fill=white!100!black] coordinates {
(4.4, 1.4)(4.4, 2.0)(5.0, 2.0)(5.0, 1.4)(4.4, 1.4)
};
\addplot[fill=white!100!black] coordinates {
(4.4, 2.4)(4.4, 3.0)(5.0, 3.0)(5.0, 2.4)(4.4, 2.4)
};
\addplot[fill=white!100!black] coordinates {
(4.4, 3.4)(4.4, 4.0)(5.0, 4.0)(5.0, 3.4)(4.4, 3.4)
};
\addplot[fill=white!100!black] coordinates {
(4.4, 4.4)(4.4, 5.0)(5.0, 5.0)(5.0, 4.4)(4.4, 4.4)
};
\addplot[fill=white!90!black] coordinates {
(5.4, 0.4)(5.4, 1.0)(6.0, 1.0)(6.0, 0.4)(5.4, 0.4)
};
\addplot[fill=white!90!black] coordinates {
(5.4, 1.4)(5.4, 2.0)(6.0, 2.0)(6.0, 1.4)(5.4, 1.4)
};
\addplot[fill=white!90!black] coordinates {
(5.4, 2.4)(5.4, 3.0)(6.0, 3.0)(6.0, 2.4)(5.4, 2.4)
};
\addplot[fill=white!90!black] coordinates {
(5.4, 3.4)(5.4, 4.0)(6.0, 4.0)(6.0, 3.4)(5.4, 3.4)
};
\addplot[fill=white!50!black] coordinates {
(5.4, 4.4)(5.4, 5.0)(6.0, 5.0)(6.0, 4.4)(5.4, 4.4)
};
\addplot[fill=white!50!black] coordinates {
(5.4, 5.4)(5.4, 6.0)(6.0, 6.0)(6.0, 5.4)(5.4, 5.4)
};
\addplot[fill=white!90!black] coordinates {
(6.4, 0.4)(6.4, 1.0)(7.0, 1.0)(7.0, 0.4)(6.4, 0.4)
};
\addplot[fill=white!90!black] coordinates {
(6.4, 1.4)(6.4, 2.0)(7.0, 2.0)(7.0, 1.4)(6.4, 1.4)
};
\addplot[fill=white!90!black] coordinates {
(6.4, 2.4)(6.4, 3.0)(7.0, 3.0)(7.0, 2.4)(6.4, 2.4)
};
\addplot[fill=white!90!black] coordinates {
(6.4, 3.4)(6.4, 4.0)(7.0, 4.0)(7.0, 3.4)(6.4, 3.4)
};
\addplot[fill=white!50!black] coordinates {
(6.4, 4.4)(6.4, 5.0)(7.0, 5.0)(7.0, 4.4)(6.4, 4.4)
};
\addplot[fill=white!50!black] coordinates {
(6.4, 5.4)(6.4, 6.0)(7.0, 6.0)(7.0, 5.4)(6.4, 5.4)
};
\addplot[fill=white!90!black] coordinates {
(6.4, 6.4)(6.4, 7.0)(7.0, 7.0)(7.0, 6.4)(6.4, 6.4)
};
\addplot[fill=white!90!black] coordinates {
(7.4, 0.4)(7.4, 1.0)(8.0, 1.0)(8.0, 0.4)(7.4, 0.4)
};
\addplot[fill=white!90!black] coordinates {
(7.4, 1.4)(7.4, 2.0)(8.0, 2.0)(8.0, 1.4)(7.4, 1.4)
};
\addplot[fill=white!90!black] coordinates {
(7.4, 2.4)(7.4, 3.0)(8.0, 3.0)(8.0, 2.4)(7.4, 2.4)
};
\addplot[fill=white!90!black] coordinates {
(7.4, 3.4)(7.4, 4.0)(8.0, 4.0)(8.0, 3.4)(7.4, 3.4)
};
\addplot[fill=white!50!black] coordinates {
(7.4, 4.4)(7.4, 5.0)(8.0, 5.0)(8.0, 4.4)(7.4, 4.4)
};
\addplot[fill=white!50!black] coordinates {
(7.4, 5.4)(7.4, 6.0)(8.0, 6.0)(8.0, 5.4)(7.4, 5.4)
};
\addplot[fill=white!90!black] coordinates {
(7.4, 6.4)(7.4, 7.0)(8.0, 7.0)(8.0, 6.4)(7.4, 6.4)
};
\addplot[fill=white!90!black] coordinates {
(7.4, 7.4)(7.4, 8.0)(8.0, 8.0)(8.0, 7.4)(7.4, 7.4)
};
\addplot[fill=white!90!black] coordinates {
(8.4, 0.4)(8.4, 1.0)(9.0, 1.0)(9.0, 0.4)(8.4, 0.4)
};
\addplot[fill=white!90!black] coordinates {
(8.4, 1.4)(8.4, 2.0)(9.0, 2.0)(9.0, 1.4)(8.4, 1.4)
};
\addplot[fill=white!90!black] coordinates {
(8.4, 2.4)(8.4, 3.0)(9.0, 3.0)(9.0, 2.4)(8.4, 2.4)
};
\addplot[fill=white!90!black] coordinates {
(8.4, 3.4)(8.4, 4.0)(9.0, 4.0)(9.0, 3.4)(8.4, 3.4)
};
\addplot[fill=white!50!black] coordinates {
(8.4, 4.4)(8.4, 5.0)(9.0, 5.0)(9.0, 4.4)(8.4, 4.4)
};
\addplot[fill=white!50!black] coordinates {
(8.4, 5.4)(8.4, 6.0)(9.0, 6.0)(9.0, 5.4)(8.4, 5.4)
};
\addplot[fill=white!90!black] coordinates {
(8.4, 6.4)(8.4, 7.0)(9.0, 7.0)(9.0, 6.4)(8.4, 6.4)
};
\addplot[fill=white!90!black] coordinates {
(8.4, 7.4)(8.4, 8.0)(9.0, 8.0)(9.0, 7.4)(8.4, 7.4)
};
\addplot[fill=white!90!black] coordinates {
(8.4, 8.4)(8.4, 9.0)(9.0, 9.0)(9.0, 8.4)(8.4, 8.4)
};
\addplot[fill=white!90!black] coordinates {
(9.4, 0.4)(9.4, 1.0)(10.0, 1.0)(10.0, 0.4)(9.4, 0.4)
};
\addplot[fill=white!90!black] coordinates {
(9.4, 1.4)(9.4, 2.0)(10.0, 2.0)(10.0, 1.4)(9.4, 1.4)
};
\addplot[fill=white!90!black] coordinates {
(9.4, 2.4)(9.4, 3.0)(10.0, 3.0)(10.0, 2.4)(9.4, 2.4)
};
\addplot[fill=white!90!black] coordinates {
(9.4, 3.4)(9.4, 4.0)(10.0, 4.0)(10.0, 3.4)(9.4, 3.4)
};
\addplot[fill=white!50!black] coordinates {
(9.4, 4.4)(9.4, 5.0)(10.0, 5.0)(10.0, 4.4)(9.4, 4.4)
};
\addplot[fill=white!50!black] coordinates {
(9.4, 5.4)(9.4, 6.0)(10.0, 6.0)(10.0, 5.4)(9.4, 5.4)
};
\addplot[fill=white!90!black] coordinates {
(9.4, 6.4)(9.4, 7.0)(10.0, 7.0)(10.0, 6.4)(9.4, 6.4)
};
\addplot[fill=white!90!black] coordinates {
(9.4, 7.4)(9.4, 8.0)(10.0, 8.0)(10.0, 7.4)(9.4, 7.4)
};
\addplot[fill=white!90!black] coordinates {
(9.4, 8.4)(9.4, 9.0)(10.0, 9.0)(10.0, 8.4)(9.4, 8.4)
};
\addplot[fill=white!90!black] coordinates {
(9.4, 9.4)(9.4, 10.0)(10.0, 10.0)(10.0, 9.4)(9.4, 9.4)
};

\end{axis}

\end{tikzpicture}}
    \caption{Graphical example illustration of the decomposition procedure in 
    \textbf{step 4.a- decomposition}:
    given an initial solution with $k = 9$ obtained in \textbf{step 3},
    we construct a solution of program~\ref{eq:strongly factor-revealing quadratic program} with $k\primed = 10$.
    Here each solid box corresponds to a two-dimensional index $(a, b)$,
    and its color denotes the magnitude of $\qHat(a, b)$. 
    First, we identify index $b^* = 5$ which maximizes 
    $\sum_{a\in[b, k]}\qHat(a, b)$.
    Then, we let 
    $\muHat(\cdot, b^*), \distanceHat(\cdot,b^*), \distanceHatStar(\cdot, b^*)$ in 
    row $b^*$ be duplicated, and $\qHat(\cdot, b^*)$
    be halved (i.e., break the black boxes into 
    dark gray boxes in row $b^*$ and $b^* + 1$).
    All light gray boxes remain unchanged, but their indexes may be shifted (depending on their positions).
    Finally, we add dummy (i.e., white) boxes whose 
    $\qHat\primed(b^*, \cdot) = 0$ and 
    duplicate $\muHat(b^*, b), \distanceHat(b^*,b), \distanceHatStar(b^*,b)$ from 
    $\muHat(b^* - 1,b)$ if $b \in[b^* - 1]$,
    and $\muHat(b^*, b^*)$ if $b = b^*$.
    }
    \label{fig:stretch}
\end{figure}
Here we use the following iterative decomposition argument.
Let $b^*$ be the index which maximizes $\sum_{a\in[b:k]}
\qHat(a, b)$.
Suppose $\sum_{a\in[b^*:k]}
\qHat(a, b^*) \geq \epsilon$.
Consider the following decomposition procedure which increases $k$ by one.\footnote{In step 4.a, we use superscript $\dagger$ to denote the solution after the decomposition modification.}
\begin{align*}
    \muHat\primed(a, b) &\gets \left\{
    \begin{array}{ll}
    \muHat(a - 1, b - 1)     
    &  \forall b\in[b^* + 1:k + 1], a\in[b:k + 1]\\
    \muHat(a - 1, b)     
    &  \forall  b\in[1:b^*], a\in [b^* + 1: k + 1] \\
    \muHat(a, b) 
    & \forall b \in[1:b^* - 1], a \in [b, b^* - 1] \\
    \muHat(b^* - 1, b)
    & \forall b \in[1:b^* - 1], a = b^*\\
    \muHat(b^*, b^*)
    &\;~ b = b^*, a = b^*
    \end{array}
    \right.
    \\
    \distanceHat\primed(a, b) &\gets \left\{
    \begin{array}{ll}
    \distanceHat(a - 1, b - 1)     
    &  \forall b\in[b^* + 1:k + 1], a\in[b:k + 1]\\
    \distanceHat(a - 1, b)     
    &  \forall  b\in[1:b^*], a\in [b^* + 1: k + 1] \\
    \distanceHat(a, b) 
    & \forall b \in[1:b^* - 1], a \in [b, b^* - 1] \\
    \muHat(b^* - 1, b)
    & \forall b \in[1:b^* - 1], a = b^*\\
    \muHat(b^*, b^*)
    &\;~ b = b^*, a = b^*
    \end{array}
    \right.
    \\
    \distanceHatStar\primed(a, b) &\gets \left\{
    \begin{array}{ll}
    \distanceHatStar(a - 1, b - 1)     
    &  \forall b\in[b^* + 1:k + 1], a\in[b:k + 1]\\
    \distanceHatStar(a - 1, b)     
    &  \forall  b\in[1:b^*], a\in [b^* + 1: k + 1] \\
    \distanceHatStar(a, b) 
    & \forall b \in[1:b^* - 1], a \in [b, b^* - 1] \\
    \muHat(b^* - 1, b)
    & \forall b \in[1:b^* - 1], a = b^*\\
    \muHat(b^*, b^*)
    &\;~ b = b^*, a = b^*
    \end{array}
    \right.
    \\
    \qHat\primed(a, b) &\gets \left\{
    \begin{array}{ll}
    \qHat(a - 1, b - 1)     
    &  \forall b\in[b^* + 2:k + 1], a\in[b:k + 1]\\
    \frac{1}{2} \cdot \qHat(a - 1, b^*)
    &
    \forall b \in [b^*: b^* + 1], a \in[b^* + 1: k + 1] 
    \\
    \qHat(a, b) 
    &
    \forall b \in [1: b^* - 1], a\in[b: b^* - 1]
    \\
    \qHat(a - 1, b)
    &
    \forall b \in [1: b^* - 1], a \in [b^* + 1, k + 1] 
    \\
    0
    & 
    \forall b \in [1:b^*], a = b^*
    \end{array}
    \right.
    \\
    \costHat\primed &\gets \costHat,
    \qquad 
    k\primed \gets k + 1
\end{align*}
See a graphical illustration of the decomposition procedure in 
\Cref{fig:stretch}.

It is straightforward to verify that the objective value 
remains unchanged after the decomposition procedure,
and 
constraints \SFRCmonotonicity, \SFRCtriangleineq, \SFRCcontributionstronger,
\SFRCdistance, \SFRCconnectcost, \SFRCtotalcost, 
as well as
$\sum_{a\in[n\primed]}\sum_{b\in[a]}\qHat\primed(a, b) = \nHat$ 
are satisfied.
The only non-trivial verification is the feasibility of 
constraint is \SFRCtriangleineq.
Here, we verify two cases respectively.\footnote{The arguments for 
the other cases are similar and thus omitted here.}
\begin{itemize}
    \item[-] 
    For $a\in[1:b^* - 1]$, $a' = b^*$,
    $b\in[a], b'\in[a']$, note that 
    \begin{align*}
        \discountfactor\cdot \muHat\primed(a', b') 
        &\overset{(a)}{=}
        \discountfactor\cdot \distanceHat\primed(a', b')
        \leq 
        \distanceHatStar\primed(a, b)
        +
        \distanceHat\primed(a, b)
        +
        \distanceHat\primed(a', b')
    \end{align*}
    where equality~(a) holds by construction.
    
    \item[-] For $a = b^*$, $a'\in[b^* + 2:k\primed]$, $b\in[a - 1], 
    b'\in[a']$, note that 
    \begin{align*}
         \discountfactor\cdot \muHat\primed(a', b') 
        &
        \overset{(a)}{=}
        \left\{
        \begin{array}{ll}
         \discountfactor\cdot \muHat(a' - 1, b' - 1)  
         &  
         \quad \text{if $b' \in [b^* + 1:k\primed]$}\\
         \discountfactor\cdot \muHat(a' - 1, b')    & 
         \quad \text{if $b' \in [1:b^*]$}
        \end{array}
        \right.
        \\
        &\overset{(b)}{\leq}
        \left\{
        \begin{array}{ll}
         \distanceHatStar(b^* - 1, b) + 
         \distanceHat(b^* - 1, b) +
         \distanceHat(a' - 1, b' - 1)  
         &  
         \quad \text{if $b' \in [b^* + 1:k\primed]$}\\
         \distanceHatStar(b^* - 1, b) + 
         \distanceHat(b^* - 1, b) +
         \distanceHat(a' - 1, b')    & 
         \quad \text{if $b' \in [1:b^*]$}
        \end{array}
        \right.
        \\
        &\overset{(c)}{\leq}
        \left\{
        \begin{array}{ll}
         \muHat(b^* - 1, b) + 
         \muHat(b^* - 1, b) +
         \distanceHat(a' - 1, b' - 1)  
         &  
         \quad \text{if $b' \in [b^* + 1:k\primed]$}\\
         \muHat(b^* - 1, b) + 
         \muHat(b^* - 1, b) +
         \distanceHat(a' - 1, b')    & 
         \quad \text{if $b' \in [1:b^*]$}
        \end{array}
        \right.
    \\
    &\overset{(d)}{=}
    \distanceHatStar\primed(a, b) + \distanceHat\primed(a, b)
    +
    \distanceHat\primed(a', b')
    \end{align*}
    where equalities~(a) (d) hold by construction;
    inequality~(b) holds due to \SFRCtriangleineq\ in the original
    solution;
    and inequality~(c) holds due to \SFRCdistance\ \SFRCconnectcost\
    in the original solution.
    
\end{itemize}

\paragraph{Step 4.b- batching.}
Suppose we repeat \textbf{step 4.a- decomposition} 
until the following event happens:\footnote{It is guaranteed 
that the event happens in the limit.
For ease of presentation, we assume the event 
happens when $k$ is still finite. 
Our analysis extends straightforward when the event happens in the limit.} 
there exists sequence $0 = b_0 < b_1 < b_2 < \dots
< b_{\nHat} = k$
such that
for every $t\in[\nHat]$,
$\sum_{b\in [b_{t - 1} + 1: b_t]}
\sum_{a\in[b:k]}\qHat(a, b) = 1$.

To simplify the notation, 
for each $t\in[\nHat], \tau \in[t]$
define $C(t, \tau) = \{(a, b):
a\in[b_{t - 1} + 1:b_t],
b\in[b_{\tau - 1} + 1:\min\{a, b_\tau\}\}$.
By definition, $\{C(t, \tau)\}_{t\in[\nHat], \tau\in[t]}$
is a partition of index set $\{(a, b):a\in[k], b\in[a]\}$.

Now we construct a feasible solution of program~\ref{eq:strongly factor-revealing quadratic program} using a batching procedure defined as follows:\footnote{In step 4.b, we use superscript $\dagger$ to denote the solution after batching.}
\begin{align*}
    &\qquad \costHat\primed \gets \costHat
    \\
    t \in [\nHat],
    \tau \in [t]:&
    \qquad 
    \muHat\primed(t, \tau) \gets
    \frac{
    \sum_{(a, b)\in C(t, \tau)}
    \qHat(a, b) \cdot \muHat(a, b)
    }{
    \sum_{(a, b)\in C(t, \tau)}
    \qHat(a, b) 
    }\\&
    \qquad 
    \distanceHat\primed(t, \tau) \gets
    \frac{
    \sum_{(a, b)\in C(t, \tau)}
    \qHat(a, b) \cdot \distanceHat(a, b)
    }{
    \sum_{(a, b)\in C(t, \tau)}
    \qHat(a, b) 
    }\\&
    \qquad 
    \distanceHatStar\primed(t, \tau) \gets
    \frac{
    \sum_{(a, b)\in C(t, \tau)}
    \qHat(a, b) \cdot \distanceHatStar(a, b)
    }{
    \sum_{(a, b)\in C(t, \tau)}
    \qHat(a, b) 
    }\\&
    \qquad 
    \qHat\primed(t, \tau) \gets
    {
    \sum_{(a, b)\in C(t, \tau)}
    \qHat(a, b) 
    }
\end{align*}
When $
    \sum_{(a, b)\in C(t, \tau)} \qHat(a, b)= 0$,
we set $\muHat\primed(t,\tau),
\distanceHat\primed(t,\tau),
\distanceHatStar(t, \tau)$
to
be the unweighted average,
and $\qHat\primed(t, \tau) = 0$.
See a graphical illustration of the batching procedure
in \Cref{fig:compression}.

\begin{figure}
    \centering
    \begin{tikzpicture}[scale=0.55, transform shape]
\begin{axis}[
axis line style=gray,
axis lines=middle,
        ytick = {1.7,4.7, 5.7, 9.7},
        yticklabels = {$b_1$, $b_2$, $b_3$, $b_4$},
        xtick = \empty,
        xticklabels = \empty,
x label style={at={(axis description cs:0.95,-0.01)},anchor=north},
y label style={at={(axis description cs:0.03,0.95)},anchor=south},
xlabel = {$a$},
ylabel = {$b$},
label style={font=\LARGE},
xmin=0,xmax=11.5,ymin=0,ymax=11.5,
width=0.9\textwidth,
height=0.9\textwidth,
samples=50]

\addplot[dotted, white!10!black] coordinates {
(0.05, 1.7) (1.35, 1.7)};

\addplot[dotted, white!10!black] coordinates {
(0.05, 4.7) (4.35, 4.7)};

\addplot[dotted, white!10!black] coordinates {
(0.05, 5.7) (5.35, 5.7)};

\addplot[dotted, white!10!black] coordinates {
(0.05, 9.7) (9.35, 9.7)};

\addplot[fill=white!90!black] coordinates {
(0.4, 0.4)(0.4, 1.0)(1.0, 1.0)(1.0, 0.4)(0.4, 0.4)
};
\addplot[fill=white!90!black] coordinates {
(1.4, 0.4)(1.4, 1.0)(2.0, 1.0)(2.0, 0.4)(1.4, 0.4)
};
\addplot[fill=white!90!black] coordinates {
(1.4, 1.4)(1.4, 2.0)(2.0, 2.0)(2.0, 1.4)(1.4, 1.4)
};
\addplot[fill=white!90!black] coordinates {
(2.4, 0.4)(2.4, 1.0)(3.0, 1.0)(3.0, 0.4)(2.4, 0.4)
};
\addplot[fill=white!90!black] coordinates {
(2.4, 1.4)(2.4, 2.0)(3.0, 2.0)(3.0, 1.4)(2.4, 1.4)
};
\addplot[fill=white!90!black] coordinates {
(2.4, 2.4)(2.4, 3.0)(3.0, 3.0)(3.0, 2.4)(2.4, 2.4)
};
\addplot[fill=white!90!black] coordinates {
(3.4, 0.4)(3.4, 1.0)(4.0, 1.0)(4.0, 0.4)(3.4, 0.4)
};
\addplot[fill=white!90!black] coordinates {
(3.4, 1.4)(3.4, 2.0)(4.0, 2.0)(4.0, 1.4)(3.4, 1.4)
};
\addplot[fill=white!90!black] coordinates {
(3.4, 2.4)(3.4, 3.0)(4.0, 3.0)(4.0, 2.4)(3.4, 2.4)
};
\addplot[fill=white!90!black] coordinates {
(3.4, 3.4)(3.4, 4.0)(4.0, 4.0)(4.0, 3.4)(3.4, 3.4)
};
\addplot[fill=white!90!black] coordinates {
(4.4, 0.4)(4.4, 1.0)(5.0, 1.0)(5.0, 0.4)(4.4, 0.4)
};
\addplot[fill=white!90!black] coordinates {
(4.4, 1.4)(4.4, 2.0)(5.0, 2.0)(5.0, 1.4)(4.4, 1.4)
};
\addplot[fill=white!90!black] coordinates {
(4.4, 2.4)(4.4, 3.0)(5.0, 3.0)(5.0, 2.4)(4.4, 2.4)
};
\addplot[fill=white!90!black] coordinates {
(4.4, 3.4)(4.4, 4.0)(5.0, 4.0)(5.0, 3.4)(4.4, 3.4)
};
\addplot[fill=white!90!black] coordinates {
(4.4, 4.4)(4.4, 5.0)(5.0, 5.0)(5.0, 4.4)(4.4, 4.4)
};
\addplot[fill=white!90!black] coordinates {
(5.4, 0.4)(5.4, 1.0)(6.0, 1.0)(6.0, 0.4)(5.4, 0.4)
};
\addplot[fill=white!90!black] coordinates {
(5.4, 1.4)(5.4, 2.0)(6.0, 2.0)(6.0, 1.4)(5.4, 1.4)
};
\addplot[fill=white!90!black] coordinates {
(5.4, 2.4)(5.4, 3.0)(6.0, 3.0)(6.0, 2.4)(5.4, 2.4)
};
\addplot[fill=white!90!black] coordinates {
(5.4, 3.4)(5.4, 4.0)(6.0, 4.0)(6.0, 3.4)(5.4, 3.4)
};
\addplot[fill=white!90!black] coordinates {
(5.4, 4.4)(5.4, 5.0)(6.0, 5.0)(6.0, 4.4)(5.4, 4.4)
};
\addplot[fill=white!90!black] coordinates {
(5.4, 5.4)(5.4, 6.0)(6.0, 6.0)(6.0, 5.4)(5.4, 5.4)
};
\addplot[fill=white!90!black] coordinates {
(6.4, 0.4)(6.4, 1.0)(7.0, 1.0)(7.0, 0.4)(6.4, 0.4)
};
\addplot[fill=white!90!black] coordinates {
(6.4, 1.4)(6.4, 2.0)(7.0, 2.0)(7.0, 1.4)(6.4, 1.4)
};
\addplot[fill=white!90!black] coordinates {
(6.4, 2.4)(6.4, 3.0)(7.0, 3.0)(7.0, 2.4)(6.4, 2.4)
};
\addplot[fill=white!90!black] coordinates {
(6.4, 3.4)(6.4, 4.0)(7.0, 4.0)(7.0, 3.4)(6.4, 3.4)
};
\addplot[fill=white!90!black] coordinates {
(6.4, 4.4)(6.4, 5.0)(7.0, 5.0)(7.0, 4.4)(6.4, 4.4)
};
\addplot[fill=white!90!black] coordinates {
(6.4, 5.4)(6.4, 6.0)(7.0, 6.0)(7.0, 5.4)(6.4, 5.4)
};
\addplot[fill=white!90!black] coordinates {
(6.4, 6.4)(6.4, 7.0)(7.0, 7.0)(7.0, 6.4)(6.4, 6.4)
};
\addplot[fill=white!90!black] coordinates {
(7.4, 0.4)(7.4, 1.0)(8.0, 1.0)(8.0, 0.4)(7.4, 0.4)
};
\addplot[fill=white!90!black] coordinates {
(7.4, 1.4)(7.4, 2.0)(8.0, 2.0)(8.0, 1.4)(7.4, 1.4)
};
\addplot[fill=white!90!black] coordinates {
(7.4, 2.4)(7.4, 3.0)(8.0, 3.0)(8.0, 2.4)(7.4, 2.4)
};
\addplot[fill=white!90!black] coordinates {
(7.4, 3.4)(7.4, 4.0)(8.0, 4.0)(8.0, 3.4)(7.4, 3.4)
};
\addplot[fill=white!90!black] coordinates {
(7.4, 4.4)(7.4, 5.0)(8.0, 5.0)(8.0, 4.4)(7.4, 4.4)
};
\addplot[fill=white!90!black] coordinates {
(7.4, 5.4)(7.4, 6.0)(8.0, 6.0)(8.0, 5.4)(7.4, 5.4)
};
\addplot[fill=white!90!black] coordinates {
(7.4, 6.4)(7.4, 7.0)(8.0, 7.0)(8.0, 6.4)(7.4, 6.4)
};
\addplot[fill=white!90!black] coordinates {
(7.4, 7.4)(7.4, 8.0)(8.0, 8.0)(8.0, 7.4)(7.4, 7.4)
};
\addplot[fill=white!90!black] coordinates {
(8.4, 0.4)(8.4, 1.0)(9.0, 1.0)(9.0, 0.4)(8.4, 0.4)
};
\addplot[fill=white!90!black] coordinates {
(8.4, 1.4)(8.4, 2.0)(9.0, 2.0)(9.0, 1.4)(8.4, 1.4)
};
\addplot[fill=white!90!black] coordinates {
(8.4, 2.4)(8.4, 3.0)(9.0, 3.0)(9.0, 2.4)(8.4, 2.4)
};
\addplot[fill=white!90!black] coordinates {
(8.4, 3.4)(8.4, 4.0)(9.0, 4.0)(9.0, 3.4)(8.4, 3.4)
};
\addplot[fill=white!90!black] coordinates {
(8.4, 4.4)(8.4, 5.0)(9.0, 5.0)(9.0, 4.4)(8.4, 4.4)
};
\addplot[fill=white!90!black] coordinates {
(8.4, 5.4)(8.4, 6.0)(9.0, 6.0)(9.0, 5.4)(8.4, 5.4)
};
\addplot[fill=white!90!black] coordinates {
(8.4, 6.4)(8.4, 7.0)(9.0, 7.0)(9.0, 6.4)(8.4, 6.4)
};
\addplot[fill=white!90!black] coordinates {
(8.4, 7.4)(8.4, 8.0)(9.0, 8.0)(9.0, 7.4)(8.4, 7.4)
};
\addplot[fill=white!90!black] coordinates {
(8.4, 8.4)(8.4, 9.0)(9.0, 9.0)(9.0, 8.4)(8.4, 8.4)
};
\addplot[fill=white!90!black] coordinates {
(9.4, 0.4)(9.4, 1.0)(10.0, 1.0)(10.0, 0.4)(9.4, 0.4)
};
\addplot[fill=white!90!black] coordinates {
(9.4, 1.4)(9.4, 2.0)(10.0, 2.0)(10.0, 1.4)(9.4, 1.4)
};
\addplot[fill=white!90!black] coordinates {
(9.4, 2.4)(9.4, 3.0)(10.0, 3.0)(10.0, 2.4)(9.4, 2.4)
};
\addplot[fill=white!90!black] coordinates {
(9.4, 3.4)(9.4, 4.0)(10.0, 4.0)(10.0, 3.4)(9.4, 3.4)
};
\addplot[fill=white!90!black] coordinates {
(9.4, 4.4)(9.4, 5.0)(10.0, 5.0)(10.0, 4.4)(9.4, 4.4)
};
\addplot[fill=white!90!black] coordinates {
(9.4, 5.4)(9.4, 6.0)(10.0, 6.0)(10.0, 5.4)(9.4, 5.4)
};
\addplot[fill=white!90!black] coordinates {
(9.4, 6.4)(9.4, 7.0)(10.0, 7.0)(10.0, 6.4)(9.4, 6.4)
};
\addplot[fill=white!90!black] coordinates {
(9.4, 7.4)(9.4, 8.0)(10.0, 8.0)(10.0, 7.4)(9.4, 7.4)
};
\addplot[fill=white!90!black] coordinates {
(9.4, 8.4)(9.4, 9.0)(10.0, 9.0)(10.0, 8.4)(9.4, 8.4)
};
\addplot[fill=white!90!black] coordinates {
(9.4, 9.4)(9.4, 10.0)(10.0, 10.0)(10.0, 9.4)(9.4, 9.4)
};

\addplot[dashed, line width=0.8mm] coordinates {
(0.3, 0.3)(0.3, 2.1)(2.1, 2.1)(2.1, 0.3)(0.3, 0.3)
};
\addplot[dashed, line width=0.8mm] coordinates {
(2.3, 0.3)(2.3, 2.1)(5.1, 2.1)(5.1, 0.3)(2.3, 0.3)
};
\addplot[dashed, line width=0.8mm] coordinates {
(2.3, 2.3)(2.3, 5.1)(5.1, 5.1)(5.1, 2.3)(2.3, 2.3)
};
\addplot[dashed, line width=0.8mm] coordinates {
(5.3, 0.3)(5.3, 2.1)(6.1, 2.1)(6.1, 0.3)(5.3, 0.3)
};
\addplot[dashed, line width=0.8mm] coordinates {
(5.3, 2.3)(5.3, 5.1)(6.1, 5.1)(6.1, 2.3)(5.3, 2.3)
};
\addplot[dashed, line width=0.8mm] coordinates {
(5.3, 5.3)(5.3, 6.1)(6.1, 6.1)(6.1, 5.3)(5.3, 5.3)
};
\addplot[dashed, line width=0.8mm] coordinates {
(6.3, 0.3)(6.3, 2.1)(10.1, 2.1)(10.1, 0.3)(6.3, 0.3)
};
\addplot[dashed, line width=0.8mm] coordinates {
(6.3, 2.3)(6.3, 5.1)(10.1, 5.1)(10.1, 2.3)(6.3, 2.3)
};
\addplot[dashed, line width=0.8mm] coordinates {
(6.3, 5.3)(6.3, 6.1)(10.1, 6.1)(10.1, 5.3)(6.3, 5.3)
};
\addplot[dashed, line width=0.8mm] coordinates {
(6.3, 6.3)(6.3, 10.1)(10.1, 10.1)(10.1, 6.3)(6.3, 6.3)
};

\end{axis}

\end{tikzpicture}
    \caption{Graphical example illustration of \textbf{step 4.b- batching}:
    given an initial solution with $k = 10$ in \textbf{step 4.a}, 
    we construct a feasible solution of program~\ref{eq:strongly factor-revealing quadratic program} with $\nHat = 4$.
    Here each solid box corresponds to an two-dimensional index $(a, b)$ of the initial solution.
    First, we identify $b_1 = 2, b_2 = 5, b_3 = 6, b_4 = 10$
    such that 
    for every $t\in[\nHat]$,
    $\sum_{b\in [b_{t - 1} + 1: b_t]}
    \sum_{a\in[b:\nHat]}\qHat(a, b) = 1$.
    Then, we take the average for each dashed box,
    which corresponds to an two-dimensional index $(t, \tau)$ in program~\ref{eq:strongly factor-revealing quadratic program}.
    }
    \label{fig:compression}
\end{figure}

It is straightforward to verify that 
the constructed solution has the same objective value.
Finally, we verify that the constructed solution 
satisfies all constraints in program~\ref{eq:strongly factor-revealing quadratic program}, which finishes our proof.
\begin{enumerate}
    \item[\SFRCmonotonicity] It is implied by the solution construction 
    and \SFRCmonotonicity\ in the original solution.
    \item[\SFRCtriangleineq] It follows a similar argument as 
    \textbf{step 3} and is omitted here.
    \item[\SFRCcontribution] It follows a similar argument as 
    \textbf{step 3}.
    Specifically, for any $t\in[\nHat]$, we have
    \begin{align*}
        \costscalar\cdot \costHat\primed 
        &\overset{(a)}{=}
        \costscalar\cdot \costHat
        \\
        &\overset{(b)}{\geq}
        \sum_{a'\in[b_t:k]}
     \sum_{b'\in[a]}
     \qHat(a', b')
     \cdot 
     \plus{
     \discountfactor\cdot \muHat(a', b') - \distanceHat(a', b')
     }
     \\
     &\qquad\qquad 
     +\sum_{a'\in[b_t:k]}
     \sum_{b'\in[a+1:a']}
     \qHat(a', b')
     \cdot 
     \plus{
     \discountfactor\cdot \muHat(b_t, b_t) - \distanceHat(a', b')
     }
        \\
        &\overset{}{\geq}
        \sum_{a'\in[b_t + 1:k]}
     \sum_{b'\in[a]}
     \qHat(a', b')
     \cdot 
     \plus{
     \discountfactor\cdot \muHat(a', b') - \distanceHat(a', b')
     }
     \\
     &\qquad
     +\sum_{a'\in[b_t + 1:k]}
     \sum_{b'\in[a+1:a']}
     \qHat(a', b')
     \cdot 
     \plus{
     \discountfactor\cdot \muHat(b_t, b_t) - \distanceHat(a', b')
     }
        \\
        &\overset{(c)}{\geq}
        \sum_{t' \in[t + 1:\nHat]}
        \sum_{\tau'\in[t]}
        \left(
        \sum_{(a', b')\in C(t', \tau')}\qHat(a', b')
        \right)
        \cdot 
        \plus
        { 
        \frac{
        \sum_{(a', b')\in C(t', \tau')}\qHat(a', b')\cdot 
        (\discountfactor\cdot \muHat(a',b') - \distanceHat(a', b'))
        }{
        \sum_{(a', b')\in C(t', \tau')}\qHat(a', b')
        }
        }
     \\
     &\qquad
        +\sum_{t' \in[t + 1:\nHat]}
        \sum_{\tau'\in[t + 1:t']}
        \left(
        \sum_{(a', b')\in C(t', \tau')}\qHat(a', b')
        \right)
        \cdot 
        \plus
        { 
        \frac{
        \sum_{(a', b')\in C(t', \tau')}\qHat(a', b')\cdot 
        (\discountfactor\cdot \muHat(b_t,b_t) - \distanceHat(a', b'))
        }{
        \sum_{(a', b')\in C(t', \tau')}\qHat(a', b')
        }
        }
        \\
        &\overset{(d)}{=}
        \sum_{t' \in[t + 1:\nHat]}
        \sum_{\tau'\in[t]}
        \qHat\primed(t',\tau')
        \cdot 
        \plus
        { 
        \discountfactor\cdot \muHat\primed(t',\tau') -
        \distanceHat\primed(t', \tau')
        }
     \\
     &\qquad
        +\sum_{t' \in[t + 1:\nHat]}
        \sum_{\tau'\in[t + 1:t']}
        \qHat\primed(t',\tau')
        \cdot 
        \plus
        {
        \discountfactor\cdot \muHat(b_t, b_t) - 
        \distanceHat\primed(t', \tau')
        }
        \\
        &\overset{(e)}{\geq}
        \sum_{t' \in[t + 1:\nHat]}
        \sum_{\tau'\in[t]}
        \qHat\primed(t',\tau')
        \cdot 
        \plus
        { 
        \discountfactor\cdot \muHat\primed(t',\tau') -
        \distanceHat\primed(t', \tau')
        }
     \\
     &\qquad
        +\sum_{t' \in[t + 1:\nHat]}
        \sum_{\tau'\in[t + 1:t']}
        \qHat\primed(t',\tau')
        \cdot 
        \plus
        {
        \discountfactor\cdot \muHat\primed(t, t) - 
        \distanceHat\primed(t', \tau')
        }
    \end{align*}
    where equalities~(a) (d) hold by construction;
    inequality~(b) holds due to \SFRCcontributionstronger\ at $a = b_t$
    in the original solution;
    inequality~(c) holds due to the convexity of $\plus{\cdot}$;
    and 
    inequality~(e) holds due to 
    the construction of 
    $\muHat\primed(t, t)$ and 
    \SFRCmonotonicity\ in the original solution,
    which guarantees $\muHat\primed(t, t) \leq 
    \muHat(b_t, b_t)$.
    \item[\SFRCdistance] It is implied by the solution construction
    and \SFRCdistance\ in the original solution.
    \item[\SFRCconnectcost] It is implied by the solution construction
    and \SFRCconnectcost\ in the original solution.
    \item[\SFRCtotalcost] It is implied by the solution construction
    and \SFRCtotalcost\ in the original solution.
    \item[\SFRCdensity] It is implied by the definition 
    of sequence $(b_1, \dots, b_{\nHat})$.
\end{enumerate}
\end{proof}

\end{document}